\documentclass[12pt]{article}

\usepackage{algorithm,algorithmicx,algpseudocode}
\usepackage[round,authoryear,compress,sort]{natbib}
\usepackage{amsmath,amssymb,amsfonts, amsbsy,amsthm, epsfig,  graphicx,rotating, tabularx, booktabs, titling}

\allowdisplaybreaks

\usepackage[usenames]{color}
\usepackage{multirow,xspace,setspace,booktabs,subcaption}
\usepackage{enumerate}
\usepackage{float}
\usepackage{tabularx}
\usepackage{capt-of}
\usepackage{lipsum}
\usepackage{tabulary}
\usepackage{thmtools}
\usepackage{diagbox}
\usepackage{xcolor}
\usepackage{enumitem}
\setlist[enumerate]{leftmargin=*}
\usepackage{makecell}
\RequirePackage[colorlinks,linkcolor=blue,citecolor=blue,urlcolor=blue]
{hyperref}

\usepackage{verbatim} 

\usepackage{scalerel}   
\newcommand{\bigtimes}{\mathop{\scalerel*{\times}{\sum}}}

\algdef{SE}[SUBALG]{AlgIndent}{AlgEndIndent}{}{\algorithmicend\ }%
\algtext*{AlgIndent}
\algtext*{AlgEndIndent}

\usepackage{fancyhdr}

\newcounter{dgp}                     

\usepackage{fancyhdr}
\newcommand{\beq}{\begin{eqnarray*}}
\newcommand{\eeq}{\end{eqnarray*}}
\newcommand{\beqn}{\begin{eqnarray}}
\newcommand{\eeqn}{\end{eqnarray}}

\newcommand{\ra}{\rightarrow}

\newcommand{\var}{{\rm var}}
\newcommand{\bi}{\begin{itemize}}
\newcommand{\ei}{\end{itemize}}

\newcommand{\be}{\begin{equation}}
\newcommand{\ee}{\end{equation}}

\newcommand{\bfm}[1]{\ensuremath{\mathbf{#1}}}

\def\bi{\bfm i}

\DeclareMathOperator{\plim}{plim}

\algrenewcommand\algorithmicrequire{\textbf{Input:}}
\algrenewcommand\algorithmicensure{\textbf{Output:}}

\renewcommand{\[}{\left\[}
\renewcommand{\]}{\right\]}

\newcommand{\argmin}{\mathop{\rm arg\min}}

\numberwithin{equation}{section}
\theoremstyle{plain}

\newtheorem{lemma}{Lemma}[section]

\newtheorem{theorem}{Theorem}[section]

\theoremstyle{definition}       
\newtheorem{assumption}{Assumption}[section]



\DeclareMathOperator{\npr}{pr}
\DeclareMathOperator{\ncov}{cov}

\newtheorem{remark}{Remark}[section]
\newtheorem{definition}{Definition}[section]

\def\be{\begin{equation}}
\def\ee{\end{equation}}

\usepackage{cleveref}
\usepackage{authblk}
\usepackage{booktabs}

\usepackage{bibunits}
\defaultbibliographystyle{bka}  

\author[1]{Shihao Zhang}
\affil[1]{School of Mathematical Sciences, University of Southampton, Southampton, UK}
\author[2]{Zudi Lu}
\affil[2]{Department of Biostatistics, City University of Hong Kong, Hong Kong SAR, China}  
\author[1]{Chao Zheng}

\begin{document}

\title{A Simple and Effective Random Forest Modelling for Nonlinear Time Series Data}

\date{}
\maketitle
\vspace{-1.6cm}
\begin{abstract}

In this paper, we propose Random Forests by Random Weights (RF-RW), a theoretically grounded and practically effective alternative RF modelling for nonlinear time series data, where existing RF-based approaches struggle to adequately capture temporal dependence. RF-RW reconciles the strengths of classic RF with the temporal dependence inherent in time series forecasting. Specifically, it avoids the bootstrap resampling procedure, therefore preserves the serial dependence structure, whilst incorporates independent random weights to reduce correlations among trees. We establish non-asymptotic concentration bounds and asymptotic uniform consistency guarantees, for both fixed- and high-dimensional feature spaces, which extend beyond existing theoretical analyses of RF. Extensive simulation studies demonstrate that RF-RW outperforms existing RF-based approaches and other benchmarks such as SVM and LSTM. It also achieves the lowest error among competitors in our real-data example of predicting UK COVID-19 daily cases.

\end{abstract}


\section{Introduction} \label{sec:intro}
\begin{bibunit}
Time series data, exhibiting temporal dependence and ordering that may be only partially observed, are prevalent across scientific and applied disciplines. Efficient modelling of such data is essential both for understanding the underlying generating processes and for improving forecasting performance.  However, this remains challenging due to the complex and often nonlinear dependence structures that characterize many real-world series. To address this, a wide range of nonparametric approaches have been proposed \citep{Masini2023,Fan2003,Gao2007,terasvirta2010}, offering greater flexibility than traditional parametric models, which are typically constrained by rigid distributional or structural assumptions. 

Random Forests (RF) \citep{Breiman2001} have gained increasing prominence in nonlinear classification and regression problems, owing to their robustness to overfitting, relatively low tuning requirements, and ability to handle high-dimensional features. The central idea is to construct a large collection of randomized decision trees \citep{Breiman1984}, each trained on a subsample of the data. While the prediction from an individual tree may be highly variable, aggregating across many trees substantially reduces variance and thereby improves predictive accuracy \citep{Hastie2009}. 
However, direct application of RF to time series forecasting \citep{YEŞİLKANAT2020,Galasso2022,Rady2021,Kane2014} fails to account for temporal dependence, as the bootstrap subsampling shuffles the data, thus disrupting the inherent sequential structure of the time series. To address this issue, recent studies have focused on modifying the structure of RF to better accommodate dependent data. The most straightforward approach is to remove the bootstrapping procedure altogether. For example, \citet{Chen2024} implement a quantile random forest without bootstrapping in a panel and time series setting. Beyond these, \citet{Saha2023} propose the generalized least squares Random Forest (RF-GLS), which decorrelates the data prior to subsampling, and \citet{Goehry2023} developed a block bootstrap Random Forest that preserves dependence within each data block while assuming approximate independence across blocks.


The use of bootstrap subsampling, random feature in splitting, and highly data-dependent partitioning makes it challenging to analyse the theoretical properties of RF, especially for consistency. Existing theoretical results on RF have largely been restricted to the i.i.d. setting, where pointwise consistency \citep{Wager2015}, $L_2$ consistency \citep{Biau2012,Scornet2015,Chi2022}, and pointwise asymptotic normality \citep{Wager2018,Athey2019} have been established. We refer to \cite{Scornet2025}  for a comprehensive review. Studies on the asymptotic properties of RF and its variants in time series settings remain quite limited and continue to attract growing interest. \cite{Davis2020} provides the pointwise consistency of RF without bootstrapping for nonlinear time series, while \cite{Shiraishi2024} and \cite{Chen2024} focus on quantile regression forests for time series and panel data, respectively.

In this paper, we propose a new framework for random forests modelling, called Random Forests by Random Weights (RF-RW), which is specially designed for nonlinear time series modelling. The key idea is to replace the classic bootstrapping, which distorts temporal dependence when subsampling the time series, with an alternative weighting scheme. Specifically, we assign a set of independent random weights to the series when constructing each individual tree in the forest. Beyond the random forest framework, the idea of random-weighting bootstrap is similar in spirit to that in \cite{Jin2001} and \citet{Zheng2018} for objective perturbation in estimating minimands and quantiles, respectively. It is also closely related to the multiplier (wild) bootstrap approach, which has been extensively applied in the Gaussian approximation for the maximum sums of independent high-dimensional random vectors \citep{Chernozhukov2013, Deng2020}. A primary contribution of this work is the establishment of rigorous theoretical guarantees for RF-RW in time series settings. We derive non-asymptotic concentration inequalities and asymptotic uniform consistency results that hold under both fixed- and high-dimensional feature spaces. To the best of our knowledge, these results extend beyond the scope of most existing theoretical analyses of Random Forests, which largely assume i.i.d. data, fixed-dimensional covariates, and typically address only pointwise consistency. We further demonstrate the practical performance of RF-RW through extensive simulation studies and an empirical application to modelling daily COVID-19 case counts in the United Kingdom. Across these studies, we can see RF-RW consistently and significantly outperforms existing benchmark methods. An \texttt{R} package providing an implementation of the proposed RF-RW algorithm is publicly available at \texttt{\url{https://github.com/shihaozhang73/RFRW}}.


The remainder of this paper is organized as follows. Section \ref{sec:method} introduces the methodology of RF-RW, with particular emphasis on the role of random weights and their effect on the construction of the forests. Section \ref{sec:theory} presents the theoretical results, including the uniform consistency for fixed- and high-dimensional nonlinear time series data. Section \ref{sec:sim} consists of simulation studies, while Section \ref{sec:covid} demonstrates the application of RF-RW to modelling UK COVID-19 daily cases. Additional simulation results and proofs of the theoretical results are deferred to the Supplementary Material.

\section{Random Forests by
Random Weights} \label{sec:method}

Suppose we have temporal dependent data pairs $D_T = \{(X_1,Y_1),\ldots,(X_T,Y_T)\}$, where $X_t\in  \mathbb{R}^p $ are $p$ explanatory features and $Y_t\in \mathbb{R}$ is the response over time $t\in\{1,\ldots,T\}$.  Consider the nonparametric time series regression model:
\begin{equation}
\label{eqn:rfrw_time_series}
    Y_t = f(X_t)+\epsilon_t,
\end{equation}
for some unknown measurable function $f:\mathbb{R}^p \rightarrow \mathbb{R}$ and i.i.d. centred random errors $(\epsilon_1,\ldots,\epsilon_T)^\top$.
As a special case, if we take $X_t=(Y_{t-1},\cdots,Y_{t-p})^\top$, i.e., the feature vector $X_t$ consists of $p$ lagged values of $Y_t$, then model (\ref{eqn:rfrw_time_series}) becomes a nonlinear autoregressive process of order $p$ (NLAR(p) process). 

When RF algorithm is applied directly to time series data, each regression tree within the ensemble is trained on a bootstrapped subsample of the original series. This random resampling and permutation process inherently disrupts the temporal dependence structure of the data, resulting in misspecified modelling. A similar phenomenon can also occur when employing block bootstrap, especially in the presence of long-run dependence. If we remove the bootstrap step to preserve the data dependence, for example, training every tree based on the full series $D_T$ as in \cite{Davis2020}, the inter-tree randomness will decrease, resulting in highly correlated trees that undermine the averaging effect of the forest and weaken the overall model robustness.









Our proposed solution to the above dilemma is conceptually straightforward: instead of removing the bootstrap, we modify the bootstrap procedure while retaining the original series, thereby preserving the temporal dependence structure without compromising inter-tree randomness. Specifically, consider assigning a set of random weights to the original series when growing each tree, where these weights are totally independent across different trees.  In this way, every regression tree $T_{\Lambda}$ is trained on a weighted version of $D_T$ with a set of i.i.d. random weights,  where $\Lambda$ is the partition that recursively divides the feature space into a set of non-overlapping hyper-rectangles (nodes). 
The forests RF-RW, ${H}_{\mathbf{\Lambda}}$, consists of B trees $\{T_{\Lambda_b}\}_1^B$ with $\boldsymbol\Lambda=\{\Lambda_1,\ldots,\Lambda_B\}$.  

We summarize the construction of RF-RW for time series modelling in Algorithm \ref{alg:rfrw} below.

\begin{algorithm}[H]
\caption{Random Forests by Random Weights (RF-RW) }
\label{alg:rfrw}
\begin{algorithmic}[1]
 \Require Training data $D_T$, test data $x$, model parameters $(B, k, m_{\text{try}})$:  
\For{$b=1,\ldots,B$}
\State (i) Sample i.i.d. random weights $(\omega_1,\ldots,\omega_T)^\top$.
\State (ii) Train a tree $T_{\Lambda_b}$ w.r.t  recursive partition $\Lambda_b$ : 
\AlgIndent
\State For each leaf node:
\While{\ leaf node sizes $\geq$ $k$}
\State (a) Sample $\mathcal{M}_{\text{try}}\subseteq\{1,\ldots,p\}$ of size $m_{\text{try}}$.
\State (b) Find splitting feature and relative  position  $\widehat{j}$ and $\widehat{\tau}$, respectively:
\begin{equation}
\label{eqn:rfrw_weight_splitting}
    (\widehat{j},\widehat{\tau})=\argmin_{j\in \mathcal{M}_{\text{try}},\tau}\left[\sum_{X_t\in A_\text{left}}\omega_t(Y_t-\hat{c}_\text{left})^2+\sum_{X_t\in A_\text{right}}\omega_t(Y_t-\hat{c}_\text{right})^2\right].
\end{equation}
\State (c) Split the current leaf node into $A_\text{left}$ and $A_\text{right}$.
\State (d) Set $A_\text{left}$ and $A_\text{right}$ as the new leaf nodes.
\EndWhile
\AlgEndIndent
\EndFor
\Ensure ${H}_{\mathbf{\Lambda}}(x)={B}^{-1}\sum_{b=1}^B{T}_{\Lambda_b}(x)$.

\end{algorithmic}
\end{algorithm}

The algorithm is similar to the vanilla Random Forests, except for the employment of random weights. Here, we sample a set of i.i.d. nonnegative weights with $\mathbb{E}[\omega_t]=1$ for $t=1\dots, T$ during the construction of each tree. In theoretical analysis, we require the weights to be sub-exponential, which allows for a wide range of distributions. Our simulation studies suggest that RF-RW achieves comparable predictive performance across various distributions of random weights.

In the recursive partition $\Lambda$, suppose that the current leaf node is $A=\bigtimes_{i=1}^p[v_i^-,v_i^+]\subseteq \mathbb{R}^p$, where $\bigtimes$ denotes the Cartesian product.  Based on splitting feature $j\in \mathcal{M}_{\text{try}}$ and relative position  $\tau \in (0,1)$,  we split $A$ into left child node $A_\text{left}:=\{x\in A: x^{(j)} \leq v_{j,\tau}\}$ and the right child node $A_\text{right}:=\{x\in A:x^{(j)} > v_{j,\tau}\}$, where $x=(x^{(1)},\ldots,x^{(p)})^\top\in\mathbb{R}^p$ is the observed explanatory feature with $x^{(j)}$ being the $j$-th feature, and $v_{j,\tau}:=\tau v_j^- + (1-\tau)v_j^+$ is the split position. To find the splitting feature and relative position in the partition, we need to calculate  $\hat{c}_\text{left}$ and $\hat{c}_\text{right}$ in the algorithm, which is done by minimizing the following weighted least-squares, 
\begin{equation}
\label{eqn:rfrw_weight_sum_of_square}
\min_{c}\sum_{t=1}^T\omega_tI(X_t\in A)(Y_t-c)^2, \quad 
\end{equation}
where we take $A=A_\text{left}$ and $A=A_\text{right}$, respectively.
Therefore, it it straightforward that $\hat{c}_\text{left}$ and $\hat{c}_\text{right}$ are just standardized weighted average of the response $Y_t$ in $A_\text{left}$ and $A_\text{right}$, respectively, i.e.,
\begin{equation*}
\label{eqn:rfrw_node_constant}
    \hat{c} = \frac{1}{\sum_{t=1}^T\omega_tI(X_t \in A)}\sum_{t=1}^T\omega_tY_tI(X_t \in A).
\end{equation*}

To grow a tree $T_\Lambda$, we recursively partition its nodes until the number of data points in the leaf node falls below a threshold $k$, which is referred as $k$-valid partition ($\Lambda\in \mathcal{V}_k$) in \cite{Wager2015} and \cite{Davis2020}. In this way, $T_\Lambda$ outputs a collection of non-overlapping hyper-rectangles (nodes) that together form a partition of $\mathbb{R}^p$. 

Using $T_\Lambda$ to predict the response corresponding to a new data point with 
explanatory value $x$ is similar to that in the vanilla RF. Let $L_{\Lambda}(x)$ be the 
unique leaf node, which is a hyper-rectangle subspace of $\mathbb{R}^p$, that contains $x$. The prediction is the standardized weighted average of the responses within $L_{\Lambda}$: 
\begin{equation}
\label{eqn:rfrw_tree_reg}
T_{\Lambda}(x)=\frac{1}{\sum_{t=1}^T\omega_tI\big(X_t \in L_{\Lambda}(x)\big)}\sum_{t=1}^T\omega_tY_tI\big(X_t \in L_{\Lambda}(x)\big).
\end{equation}

The construction of RF-RW involves two sources of injected randomness: the random weights assigned to each tree and the random sampling of features at each split. Random forests that omit the bootstrap procedure \citep{Davis2020, Chen2024} can be viewed as a special case of RF-RW with deterministic weights 
 $\omega_t=1$ for all $t\in\{1,\ldots,T\}$. For such models, variability arises only through feature sampling, resulting in reduced tree diversity. In contrast, the introduction of random weights in the construction of RF-RW provides an additional source of randomness that replaces the bootstrap subsampling used in classical random forest algorithms. The variation in random weights perturbs each observation’s contribution both to the selection of $(\widehat{j},\widehat{\tau})$ in the split criterion~(\ref{eqn:rfrw_weight_splitting}) and to the computation of 
$\hat{c}$ in the within-node weighted least-squares estimator~(\ref{eqn:rfrw_weight_sum_of_square}). Independent draws of random weights across trees also induce distinct perturbations in the splitting, thereby promoting additional heterogeneity in tree structures beyond feature sampling and consequently reducing inter-tree correlations.




\section{Theoretical properties}\label{sec:theory}


In this section, we establish convergence properties for the proposed RF-RW algorithm under the nonlinear time series setup. Specifically, we develop a non-asymptotic uniform concentration bound of RF-RW around their optimal counterparts (the so-called population forests, formally defined in Section \ref{sec:theory_concentration}), and further obtain asymptotic uniform consistency of RF-RW for both fixed-dimensional and high-dimensional feature spaces. 

\subsection{Transformation of explanatory features}\label{sec:theory_trans}

While it is common to assume the explanatory features $X_t\in[0,1]^p$ in the literature of random forests \citep{Wager2015,Chen2024,Meinshausen2006,Biau2012,Scornet2015}, this condition is restricted in the time series setting, for example, in the autoregressive model we should consider the $p$-variate lag features $X_t$ taking values in $\mathbb{R}^p$.

To address this issue,  we consider a transformation of $x=(x^{(1)},\ldots,x^{(p)})^\top\in\mathbb{R}^p$ to $[0,1]^p$, based on which we establish most of the theoretical proofs of RF-RW. A common way of the transformation is to apply some cumulative distribution function to each feature, i.e., $F_h(x^{(i)})=\int_{-\infty}^{x^{(i)}} h(y)dy$, where $h:\mathbb{R}\rightarrow[0,\infty)$ is a probability density function, which is strictly positive almost everywhere.

Define the domain of $F_h$ as $\bar{\mathbb{R}}:=\mathbb{R}\cup\{\pm\infty\}$ , such that $F_h(-\infty)=0$ and $F_h(+\infty)=1$. Denote the one-to-one mapping $\iota_h:\bar{\mathbb{R}}^p\rightarrow[0,1]^p$:
$$\iota_h:(x^{(1)},\ldots,x^{(p)})\mapsto\big(F_h(x^{(1)}),\ldots,F_h(x^{(p)})\big).$$
Note that $h$ and $F_h$ can be any density and corresponding cumulative distribution functions, we can choose a suitable $h$ that satisfies Assumption \ref{assum:rfrw_h_density} and impose the signal structure on the transformed feature $z$ in Section \ref{sec:theory_consistency_high} when deriving the theoretical properties of RF-RW.

The transformed feature can be denoted as $z=\iota_h(x) \in [0,1]^p$. During the tree construction procedure, a tree generated by recursively partitioning the feature space of $x\in \mathbb{R}^p$ can equivalently be viewed as being grown on the transformed space of $z\in[0,1]^p$. Since the regression trees are invariant to monotone transformation of the feature spaces, the selections of $(\widehat{j},\widehat{\tau})$ in the split criterion remain unchanged.

\subsection{Assumptions}\label{sec:theory_assump}

We introduce the following assumptions on the data-generating process and the tree construction algorithm, which help to develop the non-asymptotic concentration of the proposed method in Section \ref{sec:theory_concentration}.

\begin{assumption}
\label{assump:rfrw_all}
\edef\assprefix{\theassumption}
\begin{enumerate}[label=(\roman*), ref=\assprefix(\roman*)]
  \item[] 
    \item\label{assum:rfrw_h_density} (Density $h$). Let $h_X:\mathbb{R}^p\rightarrow[0,\infty)$ be the density of $X$. There exists $h$ and some constant $\zeta>1$, such that $\zeta^{-1} \prod^p_{i=1}{h(x_i)}\leq h_X(x)\leq \zeta\prod^p_{i=1}{h(x_i)}$, for $x=(x^{(1)},\ldots,x^{(p)})^\top\in\mathbb{R}^p$. 
  
  \item\label{assum:rfrw_eps} (Sub-gaussian error). Let $(\epsilon_1,\ldots,\epsilon_T)^\top$ be a sequence of i.i.d. sub-gaussian random error with $\mathbb{E}[\epsilon_t]=0$, and satisfies
  \begin{equation}
  \label{eqn:rfrw_subgaussian}
  \npr(|\epsilon_t|>x)\leq \gamma_1 \exp({-\gamma_2x^2}), \quad x>0,
  \end{equation}
  where $\gamma_1>0, \gamma_2>0$ are finite constants.


  \item\label{assum:rfrw_fbound} (Bounded $f$). The function $f$ in (\ref{eqn:rfrw_time_series}) satisfies
  \begin{equation*}
  \label{eqn:rfrw_assum_fbound}
  M:=\sup_{x\in\mathbb{R}^p}|f(x)|<\infty.
  \end{equation*}
  
  \item\label{assum:rfrw_alpha_mixing} (Stationarity and $\alpha$-Mixing). The time series $D_T$ is a strictly stationary and exponentially $\alpha$-mixing sequence, with the $\alpha$-mixing coefficient $$\alpha(t):= \sup_{A\in\sigma(X_1),B\in\sigma(X_{t+1})}\left|\npr(A\cap B)-\npr(A)\npr(B)\right|,$$ satisfying that $\alpha(t) \leq \exp(-C_{\alpha}t)$ for all $t\geq 1$ and some constant $C_{\alpha}\geq 1$. 
  
  \item\label{assum:rfrw_random_weights} (Sub-exponential random weights). Let the random weights $(\omega_1,\ldots,\omega_T)^\top$ in each tree be a set of i.i.d. nonnegative sub-exponential  random variables with $\mathbb{E}[\omega_t]=1$.
  
  \item\label{assum:rfrw_randomness} (Injected randomness). Let $\Theta_b=\{\theta^{\text{rw}}_{b},\theta^{\text{fs}}_{b}\}$ be the i.i.d. random object characterizing the randomness in the $b$-th tree for $b=1,\ldots,B$. Note that $\Theta_b$ consists of the set of random weights $\theta^{\text{rw}}_{b}$, which are independent of the random object $\theta^{\text{fs}}_{b}$ corresponding to the collection of random index sets for each feature splitting.
\end{enumerate}
\end{assumption}

Assumption \ref{assum:rfrw_h_density} is used to obtain a bounded density $h_Z:[0,1]^p\rightarrow[0,\infty)$ of the transformed random variable $Z$ (see Lemma \ref{lemma:rfrw_zeta} in the Supplementary Material for details). Similar results or assumptions can be found in \cite{Davis2020,Shiraishi2024,Chen2024,Wager2015}. Assumption \ref{assum:rfrw_eps} on sub-gaussian error is quite standard in RF literature.  Assumptions~\ref{assum:rfrw_fbound}–\ref{assum:rfrw_alpha_mixing}, which impose boundedness of the mean function and  $\alpha$-mixing data, are commonly adopted in the theoretical analysis of time series. Specifically, Assumption \ref{assum:rfrw_alpha_mixing} generalizes the widely used i.i.d. data assumption in existing works \citep{Wager2015,Wager2018,Athey2019,Scornet2015} to accommodate temporal dependence. Combined with Assumption~\ref{assum:rfrw_random_weights} on sub-exponential random weights, these conditions enable the application of Bernstein-type inequalities to derive exponential tail bounds for sums of dependent random variables. 
 Assumption \ref{assum:rfrw_randomness} is a standard technical condition that allows the separation of different sources of randomness in the theoretical analysis. 


\subsection{Non-asymptotic concentration of RF-RW} \label{sec:theory_concentration}

We obtained a $k$-valid weighted regression tree $T_{\Lambda}$ by taking the weighted average of the responses in the corresponding leaf nodes as in (\ref{eqn:rfrw_tree_reg}). A tree $T_{\Lambda}$ trained on the entire series gives rise to its partition-optimal counterpart $T_\Lambda^*$ (the population tree), which is defined as:
\begin{equation}
\label{eqn:rfrw_rw_tree_optimal}
T_\Lambda^*(x):= \mathbb{E}_\Lambda[ Y| X\in L_{\Lambda}(x)],
\end{equation}
where $\mathbb{E}_\Lambda$ denotes the expectation with respect to the conditional probability measure $\npr_\Lambda:=\npr(\cdot|D_T,\Theta)$ with $\Theta$ being the random object of an individual tree. The optimal $T_\Lambda^*$ is represented as the population-average response inside the leaf.

In Algorithm \ref{alg:rfrw}, the forest $H_{\boldsymbol\Lambda}$ is construed by averaging all the trees $T_{\Lambda}$. The corresponding partition-optimal forests $H_{\boldsymbol\Lambda}^*$ (the population forests) can be obtained by averaging the partition-optimal trees $T_\Lambda^*$. Specifically, let $\mathcal{W}_k:=\{\boldsymbol\Lambda\subseteq\mathcal{V}_k:|\boldsymbol\Lambda|<\infty\}$ contains all finite collections of $k$-valid partitions. Given an element $\boldsymbol\Lambda=\{\Lambda_1,\ldots,\Lambda_B\}$ of $\mathcal{W}_k$, the $k$-valid RF-RW $H_{\boldsymbol\Lambda}$ and its partition-optimal forest $H_{\boldsymbol\Lambda}^*$ are defined as,
\begin{equation*}
\label{eqn:rfrw_forest_define}
H_{\boldsymbol\Lambda}(x) =\frac{1}{B}\sum_{b=1}^BT_{\Lambda_b}(x), \quad H_{\boldsymbol\Lambda}^*(x) = \frac{1}{B}\sum_{b=1}^BT_{\Lambda_b}^*(x), \quad x\in \mathbb{R}^p.
\end{equation*}

Based on Assumption \ref{assump:rfrw_all}, we can show that the fitted forest $H_{\boldsymbol\Lambda}$ constitutes a good approximation to the partition-optimal forest $H_{\boldsymbol\Lambda}^*$ for all sufficiently large $T$.

\begin{theorem}
\label{theo:rfrw}
\edef\assprefix{\thetheorem}
    Under Assumption \ref{assump:rfrw_all}, there exists a positive constant $\beta$ such that the following two statements both hold with probability at least $1-6T^{-1}$ for all sufficiently large $T$:

    \begin{enumerate}[label=(\roman*), ref=\assprefix(\roman*)]
    \item\label{theo:rfrw_1} For fixed-dimensional feature space with  $p<\infty$,
    \begin{equation}
    \label{eqn:concentration_rfrw_low}
    \sup_{(x,\boldsymbol\Lambda)\in \mathbb{R}^p \times \mathcal{W}_k}|H_{\boldsymbol\Lambda}(x)-H_{\boldsymbol\Lambda}^*(x)|\leq \beta\frac{(\log T)^{3/2}}{k^{1/2}}.
    \end{equation}
    \item\label{theo:rfrw_2} For high-dimensional feature space with $\lim \inf p/T>0$, 
    \begin{equation}
    \label{eqn:concentration_rfrw_high}
    \sup_{(x,\boldsymbol\Lambda)\in \mathbb{R}^p \times \mathcal{W}_k}|H_{\boldsymbol\Lambda}(x)-H_{\boldsymbol\Lambda}^*(x)|\leq \beta\frac{(\log T)^{3/2}(\log p)^{1/2}}{k^{1/2}}.
    \end{equation}
\end{enumerate}
\end{theorem}

Theorem \ref{theo:rfrw} measures the deviation between $H_{\boldsymbol\Lambda}$ and the population forests $H_{\boldsymbol\Lambda}^*$ across $(x,\boldsymbol\Lambda)\in \mathbb{R}^p \times \mathcal{W}_k$. Note that in the high-dimensional case, the rate on the right-hand side of (\ref{eqn:concentration_rfrw_high}) incurs an additional factor $(\log p)^{1/2}$ compared to that of (\ref{eqn:concentration_rfrw_low}) for fixed $p$. As $T\rightarrow\infty$, we can choose a suitable slowly diverging threshold $k$ such that $(\log T)^{3/2}/k^{1/2} \rightarrow 0$ for fixed $p$, or $(\log T)^{3/2}(\log p)^{1/2}/k^{1/2} \rightarrow 0$ for $p\rightarrow\infty$ with $\liminf p/T >0$, which guarantees the uniform convergence in (\ref{eqn:concentration_rfrw_low}) and (\ref{eqn:concentration_rfrw_high}) with high probability, respectively. For example, we can choose $k=\lfloor T^c \rfloor$ for some $c\in(0,1)$ such that $k$ grows slowly with $T$. This condition of diverging leaf node sizes is similarly imposed in \cite{Wager2015} and \cite{Davis2020}. In contrast, the classic implementations of random forests \citep{Breiman2001} suggest a fixed $k$ that does not grow with $T$.

\begin{remark}
     Theorem \ref{theo:rfrw} is an extension of the fixed-dimensional random forests in \cite{Davis2020} to account for high-dimensional data under the time series setup. Under the fixed-dimensional case, we achieve a faster concentration rate $(\log T)^{3/2}/k^{1/2}$ for the proposed RF-RW than the rate $(\log T)^{2}/k^{1/2}$ obtained in \cite{Davis2020}. Under a restrictive setting of i.i.d. data and bounded response, \cite{Wager2015} considers the high-dimensional random forests and obtains a different concentration rate of $(\log T \log p)^{1/2}/k^{1/2}$.
\end{remark}

The non-asymptotic concentration of RF-RW provides the bound between $H_{\boldsymbol\Lambda}$ and $H_{\boldsymbol\Lambda}^*$. Next, we establish the asymptotic consistency for the proposed RF-RW under the cases of fixed-dimensional or high-dimensional feature spaces, by controlling the difference between the population forests $H_{\boldsymbol\Lambda}^*$ and the true function $f$.

\subsection{Consistency of fixed-dimensional RF-RW}\label{sec:theory_consistency_fixed}

To develop the consistency result of RF-RW in fixed dimension, we further pose the following additional requirements in constructing a single tree, which is similar to \cite{Davis2020}.
\begin{definition}[($\xi,k,m$)-valid partitions]
\label{def:rfrw_alpha_k_m_partition}
    Let $\xi \in (0,1/2)$, $k\geq 1$, and $m\geq 2k$. A partition $\Lambda$ is ($\xi,k,m$)-valid, denoted as $\Lambda\in \mathcal{V}_{\xi,k,m}$, if $\Lambda$ is a recursive partition, and the partition scheme follows the requirements below: (1) Every split is placed such that it puts at least a fraction $\xi$ of the data points in the parent node into each child node. (2) For each split, the probability $\pi_j$ that a node is split along the $j$-th feature is bounded from below by some $\rho>0$, for all $j\in\{1,\ldots,p\}$. (3) Node sizes of all leaves in the tree contain at least $k$ data points. (4) Any currently unsplit node will be split if it contains at least $m$ data points.
\end{definition}

Requirement\,(1) is often called the regular condition that imposes the restriction on the split position to prohibit the `edge splits'. Requirement\,(2) ensures that the split can be placed along any features, instead of being confined to only a subset of features. Requirement\,(3) guarantees that $\mathcal{V}_{\xi,k,m} \subseteq \mathcal{V}_{k}$, such that the non-asymptotic concentration result of RF-RW still applies.
Requirement (4) controls the maximum leaf node sizes. The splitting procedure proceeds if $m\geq2k$ and terminates when the new split violates Requirement\,(3). 

Moreover, we impose the following two assumptions on the continuity of the regression function $f$ and the tree depth.

\begin{assumption}
\label{assump:rfrw_fixed_consistency}
\edef\assprefix{\theassumption}
\begin{enumerate}[label=(\roman*), ref=\assprefix(\roman*)]
  \item[] 
  \item\label{assum:rfrw_lip} (Lipschitz continuity). For all $x^{\prime}$, $x^{\prime\prime} \in \mathbb{R}^p$, the function $f$ in (\ref{eqn:rfrw_time_series}) is Lipschitz continuous, satisfying that $|f(x^{\prime})-f(x^{\prime\prime})|\leq C_{f}\Vert x^{\prime}-x^{\prime\prime}\Vert$ for some positive constant $C_{f}$.

  \item\label{assum:rfrw_leave_alpha} (Tree depth). For $T\rightarrow\infty$, $\log (T/m)/\log (\xi^{-1}) \rightarrow \infty$.
\end{enumerate}
\end{assumption}

The Lipschitz continuity assumption \ref{assum:rfrw_lip} is standard in the RF literature \citep{Meinshausen2006,Scornet2015}. Let $d$ denote the depth of the tree. By Requirement (1) and (4) of Definition \ref{def:rfrw_alpha_k_m_partition}, we have $d\geq 1+\log (T/m)/\log (\xi^{-1})$. In conjunction with Assumption \ref{assum:rfrw_leave_alpha}, it implies that $d$ is diverging as $T\rightarrow \infty$, such that the maximum Euclidean distance between any two data points (diameter) of each leaf is restricted to zero, which helps to obtain the consistency of RF-RW for fixed $p$.

\begin{theorem}
\label{theo:rfrw_consistent}
\edef\assprefix{\thetheorem}
Let $\hat{H}_T$ be an $(\xi,k,m)$-forest and suppose that Assumption \ref{assump:rfrw_all}-\ref{assump:rfrw_fixed_consistency} are satisfied. For $T\rightarrow\infty$, fixed $p$, and $(\log T)^{3/2}/k^{1/2} \rightarrow 0$, the following statements hold:
\begin{enumerate}[label=(\roman*), ref=\assprefix(\roman*)]
    \item\label{theo:rfrw_consistent_1} $\hat{H}_T$ is a pointwise consistent estimator of $f$ for any $x\in \mathbb{R}^p$:
    $$\hat{H}_T(x) \overset{P}\longrightarrow f(x).$$ 
    \item\label{theo:rfrw_consistent_2} $\hat{H}_T(X)$ is a consistent estimator of $\mathbb{E}[Y|X]$:
    $$\hat{H}_T(X) \overset{P}\longrightarrow \mathbb{E}[Y|X].$$ 
    \item\label{theo:rfrw_consistent_3} If $\xi\in(0,0.2]$, then $\hat{H}_T$ is a uniformly consistent estimator of $f$:
    $$\sup_{x\in \mathbb{R}^p}\left|\hat{H}_T(x) - f(x)\right|\overset{P}\longrightarrow 0.$$
\end{enumerate}
\end{theorem}

The above results are all derived under the convergence in probability. For fixed $p$, the minimum leaf node sizes $k$ can be chosen such that $(\log T)^{3/2}/k^{1/2}\rightarrow0$, which ensures that the concentration bound from Theorem \ref{theo:rfrw_1} holds. Theorem \ref{theo:rfrw_consistent_1} demonstrates the pointwise consistency of the random forest estimator, which represents the most common type of result in the theoretical study of random forests \citep{Wager2018,Athey2019}. Theorem \ref{theo:rfrw_consistent_2} develops the consistency of the estimator evaluated at a new random data $X$, $\hat{H}_T(X)$, compared to the conditional mean $\mathbb{E}[Y|X]$. Theorem \ref{theo:rfrw_consistent_3} establishes the uniform consistency of RF-RW, which extends the pointwise convergence results of existing works \citep{Wager2015,Davis2020}.

\subsection{Consistency of high-dimensional RF-RW} \label{sec:theory_consistency_high}

In high-dimensional feature space, we generally assume that only a few (sparse) signal features contribute to estimating the true function $f$, while numerous noisy features are uninformative \citep{Biau2012}. This suggests the idea of distinguishing between informative and non-informative features.

Let $Q\subseteq\{1,\ldots,p\}$ be the informative feature index set. For the transformed feature $z=(z^{(1)},\ldots,z^{(p)})^\top\in [0,1]^p$, define $[z]_Q=(z^{(i_1)},\ldots,z^{(i_{|Q|})})^\top$, where $i_1<i_2<\ldots<i_{|Q|}$ and $i_j\in Q$ for all $j$. For any $z$, $z^\prime$, and $Q$, denote $u=[z]_Q \otimes [z^\prime]_{Q^C}$ with $u=(u^{(1)},\ldots,u^{(p)})^\top$, such that $u^{(j)}=z^{(j)}$ if $j\in Q$ and $u^{(j)}=z^{\prime^{(j)}}$ if $j\in Q^C$ for $j\in\{1,\ldots,p\}$. Thus, $z=[z]_Q \otimes [z]_{Q^C}$ holds for all $z$. Define the diameter of the leaf $\iota_h(L_{\Lambda}):=\{z=\iota_h(x):x\in L_{\Lambda}\}$ as $\text{diam}\big(\iota_h(L_{\Lambda})\big):=\sup_{z, z^{\prime} \in \iota_h(L_{\Lambda})}||z-z^{\prime}||$.  We impose an assumption on the signal structure under high-dimensional feature space in Assumption \ref{assump:rfrw_high_lip} below, which is similar to \cite{Chen2024}.

\begin{assumption}[Sparse signal and Lipschitz continuity]
\label{assump:rfrw_high_lip}
    There exists an informative feature index set $Q$ with $|Q|=o(\log T)$. For any $z$ and $y$,
    $$h_{Y|Z}(y|z)=h_{Y|[Z]_Q}(y|[z]_Q),$$
    where $h_{Y|Z}$ is the density of $Y$ given $Z$. In addition, for any $z$, $z^{\prime}$, $z^{\prime\prime}$ and $y$, it holds that
    $$\left|h(y,[z]_Q\otimes [z^{\prime\prime}]_{Q^C})-h(y,[z^{\prime}]_Q\otimes [z^{\prime\prime}]_{Q^C})\right|\leq L(y)\left \Vert [z]_Q-[z^{\prime}]_Q\right \Vert,$$
    where $\int_{-\infty}^{\infty}L(y)dy=L_1<\infty$, and $\int_{-\infty}^{\infty}|y|L(y)dy=L_2<\infty$.
\end{assumption}

Under Assumption \ref{assump:rfrw_high_lip}, $Y$ is independent of $[Z]_{Q^C}$ conditional on $[Z]_Q$, and the joint density $(Y,Z)$ is $L(y)$-Lipschitz in $[z]_Q$. In contrast, the signal structure of \cite{Wager2015} for high-dimensional random forests requires that the uninformative features are independent of the response, which is not appropriate for time series data. Assumption \ref{assump:rfrw_high_lip} also extends the signal condition in \cite{Chen2024,Wager2015} to accommodate a growing $|Q|$ as $T\rightarrow\infty$. 

Based on Assumption \ref{assump:rfrw_all} and \ref{assump:rfrw_high_lip}, we can control the difference between the partition-optimal forest $H_{\boldsymbol\Lambda}^*$ and the true function $f$, which is shown in the following theorem.

\begin{theorem}
\label{theo:rfrw_high_h_optimal_fx}
Under Assumption \ref{assump:rfrw_all} and \ref{assump:rfrw_high_lip}, there exists a positive constant $C_{\text{SIG}}$ such that
\begin{equation}
\label{eqn:rfrw_high_h_optimal_fx}
\sup_{(x,\boldsymbol\Lambda)\in \mathbb{R}^p \times \mathcal{W}_k}|H_{\boldsymbol\Lambda}^*(x)-f(x)|\leq C_{\text{SIG}}\sup_{x\in\mathbb{R}^p}\mathbb{E}_\Theta\normalfont{\text{diam}}\left([\iota_h\big(L_{\Lambda}(x)\big)]_Q\right),
\end{equation}
where $[\iota_h(L_{\Lambda})]_Q=\left\{[\iota_h(x)]_Q:x\in L_{\Lambda}\right\}$ and $\mathbb{E}_{\Theta}$ is the expectation with respect to $\Theta$.
\end{theorem}

In Theorem \ref{theo:rfrw_high_h_optimal_fx}, we obtain the uniform convergence of the population forests approximating the true regression function that only depends on the informative features. To develop the consistency of high-dimensional RF-RW, we need to ensure the upper bound on the right-hand side of (\ref{eqn:rfrw_high_h_optimal_fx}) goes to $0$, which requires shrinking the maximum leaf diameter of the informative features. To this end, we need the following requirements on the recursive partition in growing individual trees.

\begin{definition}[($\tilde{\xi},k,m$)-valid partitions]
\label{def:rfrw_xi_k_m_partition}
    Let $\tilde{\xi} \in (0,1/2)$, $\xi \in (0,1/2)$, $k\geq 1$, and $m\geq 2k$. A partition $\Lambda$ is $(\tilde{\xi},k,m)$-valid $(\Lambda\in \mathcal{V}_{\tilde{\xi},k,m})$, if $\Lambda$ is a recursive partition, and the partition scheme follows the same requirements as the $(\xi,k,m)$-valid partitions (Definition \ref{def:rfrw_alpha_k_m_partition}) but excludes Requirement (2), which ensures that every splitting feature has a nonzero probability to be split.
\end{definition}

In Definition \ref{def:rfrw_xi_k_m_partition}, the fraction $\tilde{\xi}$ is used to control the split position (see Lemma \ref{lemma:rfrw_high_xi_range} in the Supplementary Material for details). Moreover, we relax Requirement (2) in Definition \ref{def:rfrw_alpha_k_m_partition}, such that the probability of splitting along the uninformative features is not bounded away from zero. When random forests are applied to high-dimensional data, we prefer that the informative features are prioritized in the splitting process, while the uninformative features are rarely chosen. Otherwise, allocating numerous splits to uninformative features results in little or no bias reduction in the averaging effect across trees \citep{Chen2024}. 

Under Assumption \ref{assump:rfrw_all}, \ref{assump:rfrw_high_lip}, and additional assumptions of informative feature splitting and uninformative feature controlling imposed in Assumption \ref{assump:rfrw_high_signal} in the Supplementary Material, we establish the following consistency result of RF-RW for $p\rightarrow\infty$.

\begin{theorem}
\label{theo:rfrw_high_consistent}
Let $\hat{H}_T$ be an $(\tilde{\xi},k,m)$-forest with $\tilde{\xi} \in (0,1/2)$ and suppose that Assumption \ref{assump:rfrw_all}, \ref{assump:rfrw_high_lip}, and \ref{assump:rfrw_high_signal} are satisfied. For $T\rightarrow\infty$, $p\rightarrow\infty$ with $\liminf p/T >0$, $m_{\text{try}}\asymp p$, $|Q|=o(\log T)$, and $(\log T)^{3/2}(\log p)^{1/2}/k^{1/2} \rightarrow 0$, it holds that
$$\sup_{x\in \mathbb{R}^p}\left|\hat{H}_T(x) - f(x)\right|\overset{P}\longrightarrow 0.$$
\end{theorem}

Theorem \ref{theo:rfrw_high_consistent} establishes the uniform consistency of RF-RW in probability under a high-dimensional regime such that $\lim \inf p/T>0$. The minimum leaf node sizes $k$ are required to grow with the sample sizes $T$ and the dimension $p$ at a slow rate of $(\log T)^{3/2}(\log p)^{1/2}/k^{1/2}$ to ensure that Theorem \ref{theo:rfrw_2} holds. Theorem \ref{theo:rfrw_high_consistent} extends the results of \cite{Wager2015,Biau2012} in the i.i.d case to accommodate time series data, and does not impose the bounded response condition as required in \cite{Wager2015,Chen2024}. The result is also not restricted to a specific tree structure \citep{Wager2015,Biau2012} and can be easily applied to many variants of random forests. Moreover, Theorem \ref{theo:rfrw_high_consistent} allows for an increasing $|Q|=o(\log T)$, which extends the results of \cite{Chen2024,Wager2015} that focus only on the fixed $|Q|$. 

\section{Simulation Studies} \label{sec:sim}

In this section, we measure the predictive performance of the proposed RF-RW in simulated data. Consider the following three data-generating processes of nonlinear time series:
\begin{itemize}
\item[M1:] \label{dgp:1} $Y_t = Y_{t-1}\exp({-0.4Y_{t-1}^2}) + 4Y_{t-3}^2\exp({-0.55Y_{t-3}^2})+ 4Y_{t-5}\exp({-0.3Y_{t-5}^2})  + \epsilon_t.$
    \item[M2:]\label{dgp:2} $Y_t = \sin{(5Y_{t-1})}\exp({-Y_{t-1}^2}) + 2\text{sign}(Y_{t-3})\cos{(\pi Y_{t-3})}  + \epsilon_t.$
    \item[M3:] \label{dgp:3} $Y_t = 6Y_{t-1}\exp({-0.6Y_{t-1}^2})+3\cos(X_{t-6})\exp({-0.5X_{t-6}^2})+4\sin(Y_{t-8}^2)\exp({-0.4Y_{t-8}^2})  + \epsilon_t,$
    with $X_t=2X_{t-2}\exp({-0.7X_{t-2}^2})+3\sin(X_{t-8})\text{sign}\left(\exp({-0.5X_{t-8}^2})\right) + \varepsilon_t$,\end{itemize}
    where $\epsilon_t \overset{i.i.d.}\sim N(0,1)$ and $\varepsilon_t \overset{i.i.d.}\sim N(0,1)$ are random errors. We compute the Root Mean Square Test Error (RMSTE)  over 100 Monte Carlo simulations to evaluate the predictive performance. In each simulation, we generate training data of size $T\in \{500, 1000, 2000\}$ and a subsequent test data set with size $100$. For the proposed RF-RW, the number of feature sampling at splitting $m_\text{try}$ is chosen to be $p/3$  or $p$ for the proposed method, which are denoted as RF-RW-1 and RF-RW-2, respectively. The random weights are drawn from $\text{Exp}(1)$. In Section \ref{sec:appen_sim} of the Supplementary Material, we also consider random weights following from different distributions, such as lognormal and square root of Gamma distribution. Our results indicate that the predictive performance of RF-RW is largely insensitive to the choice of random weights.

In reality, the temporal dependency of the data is typically unknown, which motivates examining the performance of methods when including different numbers of lags in the explanatory features. We consider three scenarios of lag features: fewer lag features than those in the data-generating processes (underspecified); exact lag features (correctly specified); more lag features (overspecified). 
For data-generating processes M1 and M2, methods including 3 and 2 lag features of 
$Y_t$, respectively, are regarded as underspecified, whereas those including 15 lag features are considered overspecified. For the true model M3, methods incorporating 6 lag features of both 
$Y_t$  and 	$X_t$  are treated as underspecified, while those including 15 lag features are viewed as overspecified. Among these three scenarios, including more lag features is a common approach to better incorporate the temporal information. Additional numerical studies on the effect of overspecification can be found in Section \ref{sec:appen_sim} of the Supplementary Material.

We compare the predictive performance of RF-RW to several other benchmark methods, including \texttt{RF}: classic Random Forest, \texttt{tsRF}: time series block bootstrapping Random Forest \citep{Goehry2023},  
\texttt{ExtraTrees}: extremely randomized trees \citep{Geurts2006}, \texttt{SVM}: support vector machine \citep{Chang2011}, and \texttt{LSTM}: long short-term memory neural network \citep{Paszke2019}. For forest-based methods, we fix the number of trees $B=100$, and set the minimum node size $k=\lfloor 0.02(\log T)^4\log(\log T) \rfloor$. For \texttt{tsRF}, we select the block size to be 10 in the moving block bootstrapping. \texttt{SVM} is implemented with the radial basis kernel. For \texttt{LSTM}, we employ a network with parameters optimized from three candidate sets, number of layers: $\{1,\ldots,5\}$, width: $\{16,32,64,128\}$, and dropout rate: $\{0.1,\ldots,0.5\}$. 

\begin{figure}[!ht]
\centering
\includegraphics[width=1\textwidth]{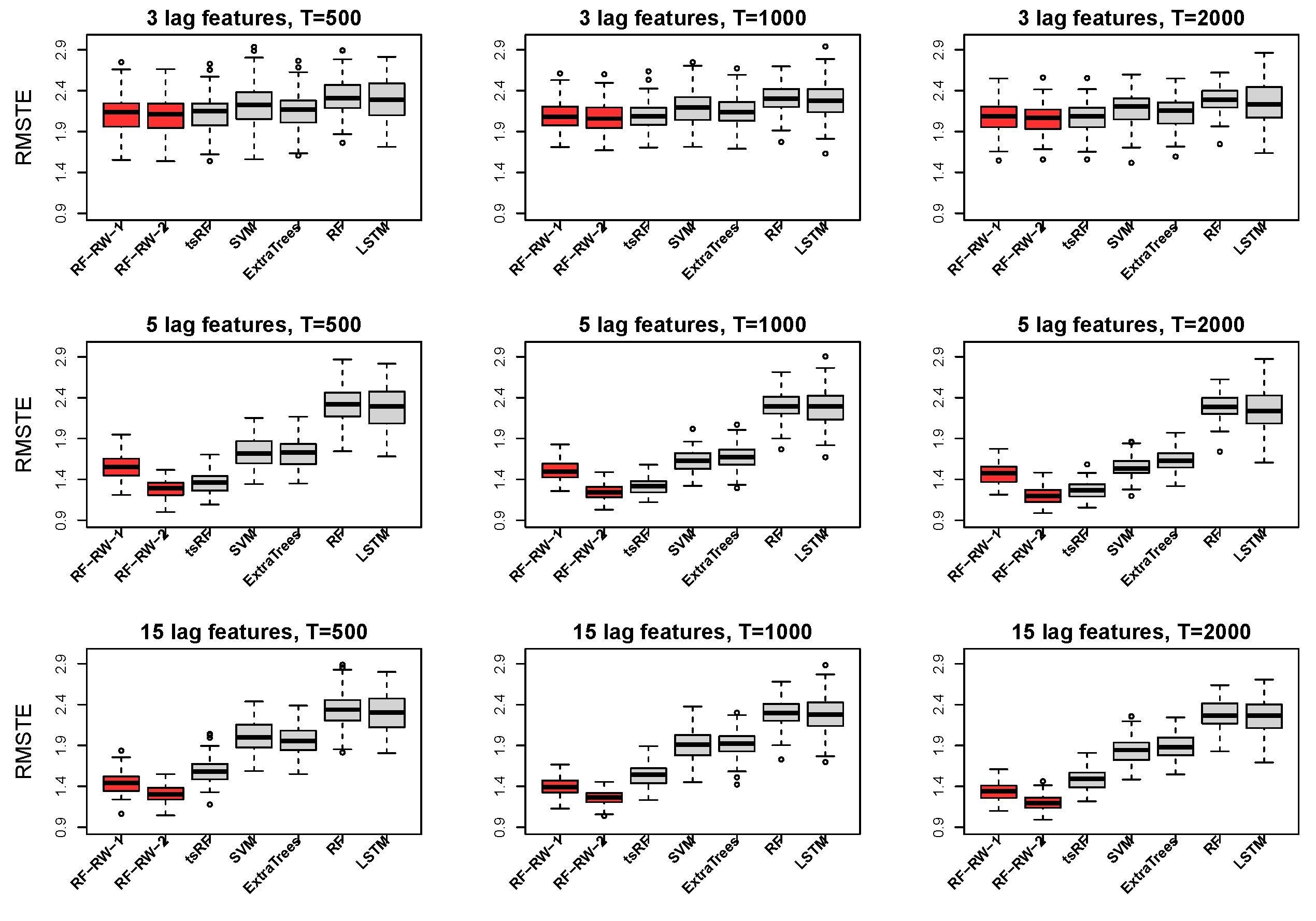}
\caption{RMSTE of  M1, including 3 lag features (underspecified), 5 lag features (correctly specified), 15 lag features (overspecified).}
\label{fig:rfrw_dgp1}
\end{figure}

\begin{figure}[!ht]
\centering
\includegraphics[width=1\textwidth]{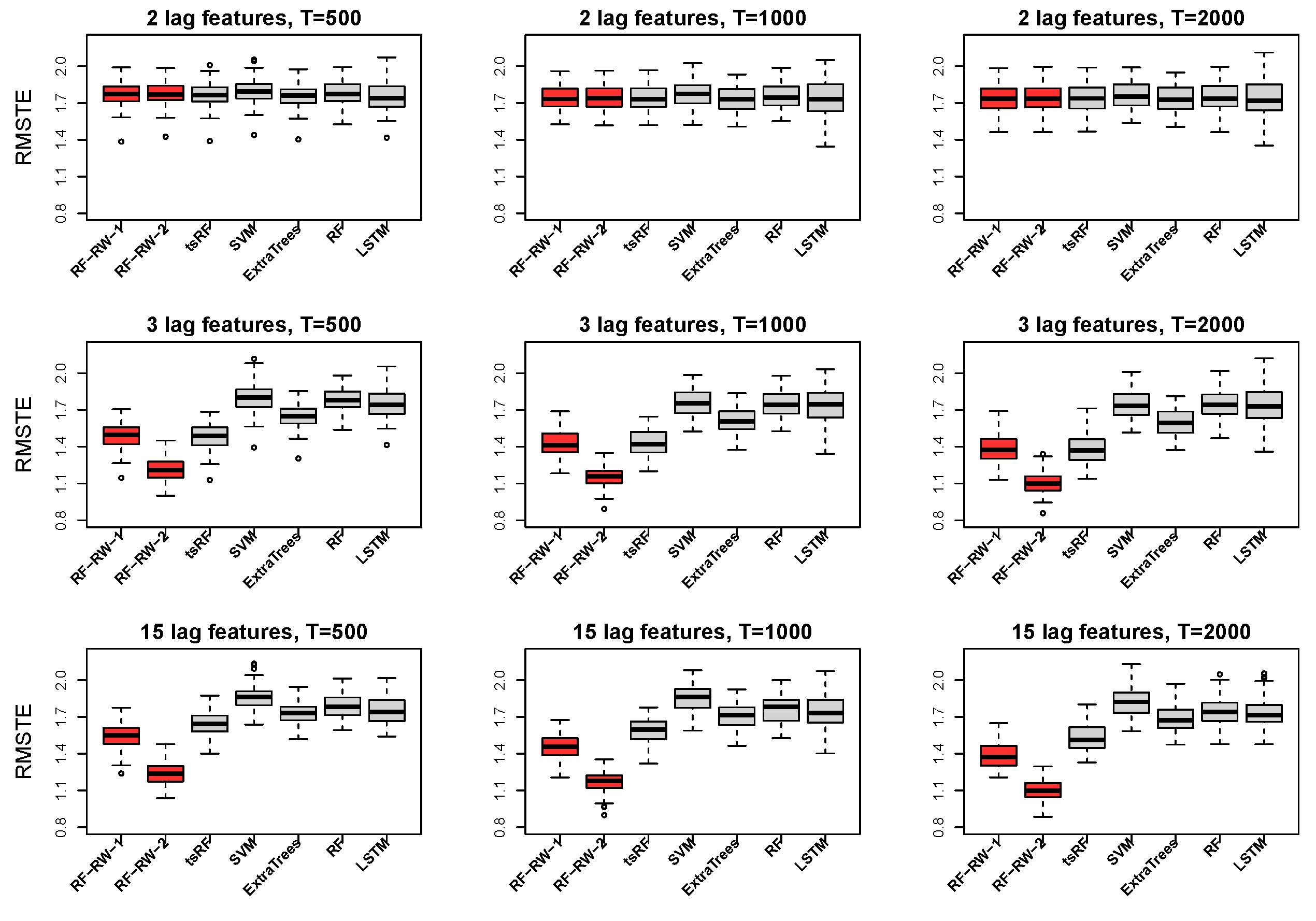}
\caption{RMSTE of M2, including 2 lag features (underspecified), 3 lag features (correctly specified), 15 lag features (overspecified).}
\label{fig:rfrw_dgp2}
\end{figure}

\begin{figure}[!ht]
\centering
\includegraphics[width=1\textwidth]{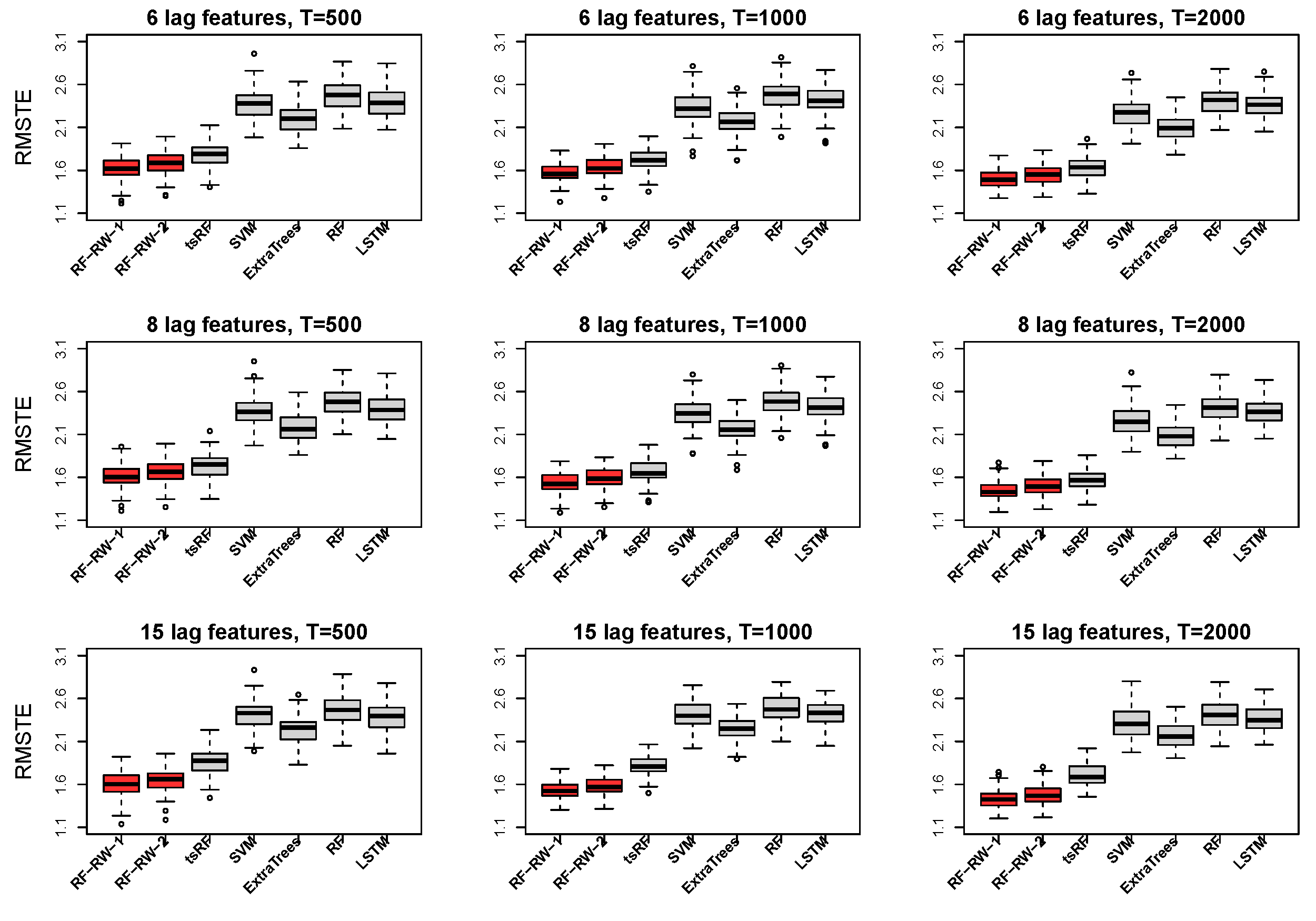}
\caption{RMSTE of M3, including 6 lag features (underspecified), 8 lag features (correctly specified), 15 lag features (overspecified).}
\label{fig:rfrw_dgp3}
\end{figure}

In Figures \ref{fig:rfrw_dgp1}-\ref{fig:rfrw_dgp3}, as the sample sizes $T$ increase, the RMSTE of the proposed method decreases in all settings, which matches our theoretical results that the RF-RW estimator converges to the true function $f$ as $T\ra \infty$. For underspecified scenarios, all the methods cannot adequately account for serial dependence due to missing important features, which yield pervasively higher RMSTE. RF-RW uniformly outperforms other methods under correctly and overspecified scenarios. Taking the predictive performance under the correctly specified scenarios as the baseline, we also find that the RMSTE of RF-RW remains stable as more lag features are included, indicating that the method tolerates feature overspecification and may mitigate the curse of dimensionality. This finding is further supported by additional simulation studies in the Supplementary Material, where RF-RW with 30 and 50 lag features exhibits performance comparable to that under the correctly specified scenarios. 

In contrast, although \texttt{tsRF} exhibits competitive predictive performance under correct specification, its RMSTE deteriorates as more lag features are included, limiting its applicability in high-dimensional settings.

\section{An application to UK COVID-19 daily cases modelling} \label{sec:covid}

We highlight the strength of the proposed RF-RW in predicting the COVID-19 daily cases in the UK. The data spans from 15 Apr 2020 to 15 May 2022, which consists of 761 days and can be accessed at \texttt{\url{https://ukhsa-dashboard.data.gov.uk/}}. Cases recorded from 15 Apr 2020 to 24 Apr 2022 (740 data) are taken as the training data, while the remaining 21 data recorded from 25 Apr 2022 to 15 May 2022 are used as test data.

To stationarize the series, we take the log difference of the original data, and apply a 7-day difference to eliminate the weekly seasonality, which is used as the response. The lag values of the response are used as explanatory features. The predicted cases on the test set, obtained by transforming the predicted response, are used to assess forecasting accuracy.

We evaluate the performance of RF-RW with random weights sampled from $\text{Exp}(1)$, and compare it with \texttt{tsRF}, \texttt{ExtraTrees}, \texttt{RF},  \texttt{RF-GLS} \citep{Saha2023} and the classic Linear Seasonal Autoregressive Integrated Moving Average (\texttt{LS-ARIMA}). For all the forest-based methods, we consider 7, 14 and 21 lags of the response included in the explanatory features. The number of trees is fixed at $B=250$ and the minimum node size is $k=5$. For \texttt{tsRF}, the block size is set to be 5, as suggested by \cite{Goehry2023} for weekly data. For \texttt{LS-ARIMA}, we choose the non-seasonal AR order from $\{7,14,21,28\}$, choose the seasonal AR order from  $\{0,1,2,3,4,5,6\}$, and set the seasonal period to be $7$, while the moving average orders and degrees of differencing are all set to be $0$. 

\begin{table}[!ht]
\centering
\caption{RMSTE of COVID-19 cases of the forest-based methods with optimal LS-ARIMA$(7,0,0)(1,0,0)_7$}
\label{table:covid_data_rmse}
\begin{tabular}{c | c | c  | c } 
 \hline
 \backslashbox[50mm]{Method}{Lags}  & 7 & 14 & 21 \\
 \hline
 RF-RW-1 & 931.78 & 958.92 & \textbf{888.01} \\ 
 RF-RW-2 & \textbf{897.76} & \textbf{950.60} & 917.84 \\ 
 tsRF & 920.31 & 999.59 & 965.38 \\ 
 ExtraTrees & 921.64 & 973.76 & 951.51 \\ 
 RF & 979.30 & 974.58 & 930.86 \\
 RF-GLS & 984.45 & 1061.84 & 1146.93 \\ 
 \hline
 LS-ARIMA$(7,0,0)(1,0,0)_7$ &  \multicolumn{3}{c}{1046.71} \\
 \hline
 
\end{tabular}
\end{table}

Table \ref{table:covid_data_rmse} shows that the proposed RF-RW generally outperforms other methods, with the lowest RMSTE of 888.01. The forest-based methods have better predictive performance than the \texttt{LS-ARIMA} methods in most cases, which is similar to the results of many previous studies. The inferior predictive performances of RF and its variants reveal their intrinsic limitation in adequately capturing the underlying temporal information.

In Figure \ref{fig:covid}, we present the results of RF-RW-1, which includes 21 lags of the response and achieves the minimum RMSTE. The left panel shows the estimated and predicted cases align with the actual cases against the date. The right panel highlights the predicted cases reproducing the overall pattern of the actual cases. These results exhibit the strength of the proposed RF-RW that attains accurate fitting and provides reliable forecasting.

\begin{figure}[H]
\centering
\includegraphics[width=1\textwidth]{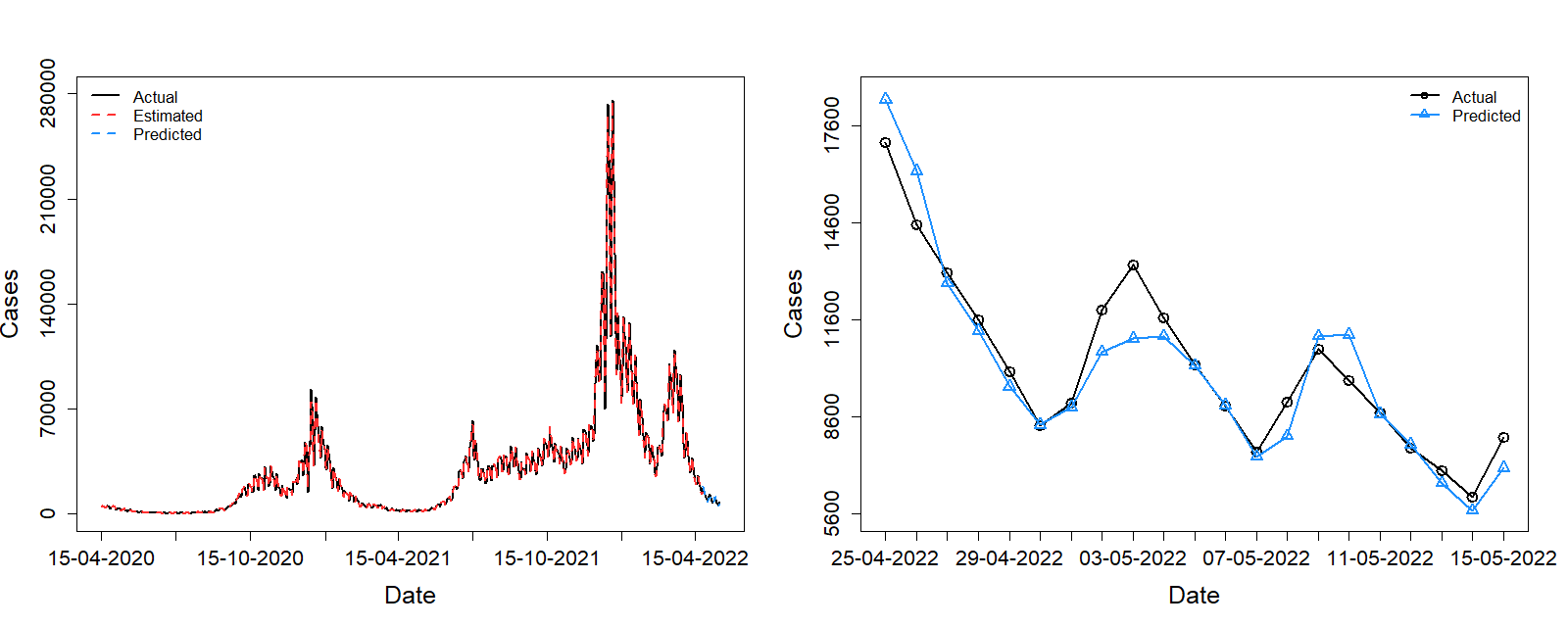}
\caption{RF-RW-1 that includes 21 lags of the response. Left: Actual, Estimated, and Predicted cases against Date. Right: Actual and Predicted cases against Date.}
\label{fig:covid}
\end{figure}

\section{Discussion}\label{sec:discussion}

In this paper, we propose the Random Forests by Random Weights (RF-RW) algorithm, a practically effective random forest framework for modeling nonlinear time series data. It is also of interest to extend the RF-RW methodology to spatio-temporal settings, enabling the model to capture both spatial and temporal dependencies as well as complex nonlinear interactions across space and time. A detailed investigation of such spatio-temporal extensions of RF-RW is left for future work. Moreover, the proposed random weighting scheme is broadly applicable beyond conditional mean prediction, for example, to estimating conditional quantile. From a theoretical perspective, similar to \cite{Wager2015} and \cite{Davis2020}, our results share the limitation that the concentration bounds for the forests are of the same order as those for individual trees, and therefore do not reflect the averaging effect of the ensemble. This theoretical gap remains an open problem in the literature and is left for future research.


\putbib[literature]
\end{bibunit}

\appendix
\newpage
\setcounter{page}{1}
\pagestyle{fancy}
\fancyhf{}

\setcounter{figure}{0}
\makeatletter
\renewcommand{\thefigure}{A\@arabic\c@figure}
\makeatother

\setcounter{table}{0}
 \makeatletter
\renewcommand{\thetable}{A\@arabic\c@table}
\makeatother
\rhead{\bfseries\thepage}
\setcounter{equation}{0}
\noindent


\section*{\centering Supplementary Material to ``A Simple and Effective Random Forest Modelling for Nonlinear Time Series Data"}\label{sec:supp}
\begin{center}

\large{Shihao Zhang, Zudi Lu, and Chao Zheng}\\



\end{center}

\bigskip

The organization of this supplementary material is as follows. Section \ref{sec:appen_transformation_supporting_lemma} introduces the density properties of the transformed feature and several auxiliary lemmas. Section \ref{sec:appen_proof_concentration} provides the proofs of the concentration bounds for the weighted regression trees and for RF-RW around their population counterparts. The consistency results for RF-RW in fixed and high dimensions are proved in Sections \ref{sec:appen_proof_fixed_consistency} and \ref{sec:appen_proof_high_consistency}, respectively. Additional simulation results are presented in Section \ref{sec:appen_sim}.

Throughout the supplementary material, we denote constants $C$ and $c$, whose values do not depend on $T$ and may vary from place to place. For two sequence $(a_t)_{t\geq 1}$ and $(b_t)_{t\geq 1}$, we write $a_t \lesssim b_t$ if there exists a constant $c\geq 1$ such that $a_t \leq cb_t$ for all $t$. The notation $a_t \asymp b_t$ represents that $a_t \lesssim b_t$ as well as $b_t \lesssim a_t$. We define $a_t=O(b_t)$ if $\lim \sup_{t\rightarrow\infty} a_t/b_t\leq c$.


\begin{bibunit}



\section{Density of transformed feature and supporting lemmas}\label{sec:appen_transformation_supporting_lemma}

\subsection{Bounded density $h_Z$ of the transformed feature $Z$}

Based on Assumption \ref{assum:rfrw_h_density}, we can show that the density $h_Z$ of $Z$ is lower and upper bounded, which is shown in the lemma below.

\begin{lemma}
\label{lemma:rfrw_zeta}
Under Assumption \ref{assum:rfrw_h_density}, there exists a constant $\zeta>1$, such that the density $h_Z$ satisfies
\begin{equation}
\label{eqn:rfrw_lemma_zeta}
\zeta^{-1}\leq h_Z(z) \leq \zeta,
\end{equation}
for almost all $z=(z^{(1)},\ldots,z^{(p)})^\top\in[0,1]^p$.
\end{lemma}
\begin{proof}
By applying the mapping $\iota_h$, the density $h_Z$ becomes
\begin{equation}
\label{eqn:rfrw_lemma_zeta_1}
h_Z(z)=h_X\big(\iota_h^{-1}(z)\big)\left|\det\left(J_{\iota_h^{-1}}(z)\right)\right|,
\end{equation}
where $\det\left(J_{\iota_h^{-1}}(z)\right)$ is the determinant of the $p\times p$ Jacobian matrix $J_{\iota_h^{-1}}$ that is evaluated at $z$. Since the mapping $\iota_h$ is one-to-one monotone, and $x_i=F_h^{-1}(z_i)$, the Jacobian matrix $J_{\iota_h^{-1}}$ is diagonal:
\begin{equation*}
  \left[J_{\iota_h^{-1}}(z)\right]_{i,j} =
    \begin{cases}
      {\partial F_h^{-1}(z_i)}/{\partial z_j}, & \text{if}\ i=j.\\
      0, & \text{if}\ i\neq j. \nonumber
    \end{cases}       
\end{equation*}
Applying the inverse function theorem, we have 
$$\frac{\partial F_h^{-1}(z_i)}{\partial z_i}=\frac{1}{h\big(F_h^{-1}(z_i)\big)}.$$
It follows that 
$$\det\left(J_{\iota_h^{-1}}(z)\right)=\frac{1}{\prod^p_{i=1}{h\big(F_h^{-1}(z_i)\big)}},$$ 
and (\ref{eqn:rfrw_lemma_zeta_1}) becomes
$$h_Z(z)=\frac{h_X\big(\iota_h^{-1}(z)\big)}{\prod^p_{i=1}{h\big(F_h^{-1}(z_i)\big)}}.$$
By Assumption \ref{assum:rfrw_h_density}, we have
$$\zeta^{-1} \leq \frac{h_X(x)}{\prod^p_{i=1}{h(x_i)}}\leq \zeta,$$
for some constant $\zeta>1$, which finishes the proof.

\end{proof}

By applying the mapping $\iota_h$ with $h_Z$ satisfying (\ref{eqn:rfrw_lemma_zeta}), a node $R=\bigtimes_{i=1}^p[r_i^-,r_i^+]\subseteq [0,1]^p$ can be viewed as the image of a node $A=\bigtimes_{i=1}^p[v_i^-,v_i^+]\subseteq \mathbb{R}^p$ under $\iota_h$, that is, $R=\iota_h(A)=\bigtimes_{i=1}^p[F_h(v_i^-),F_h(v_i^+)]$. For the node $R$ and the transformed feature $Z_t$, we introduce the following notations. Let $\#R:=|t=\{1,\ldots,T\}:z_t\in R|$ be the number of observations $\{z_t\}_{t=1}^T$ in $R$, $\mu(R):=\npr(Z\in R)$ be the probability that $Z$ belongs to $R$, and $\eta(R):=\mathbb{E}[Y|Z\in R]$ be the population-average response inside $R$.

\subsection{Supporting lemmas}

\begin{lemma}[\citealt{Merlevède2009}, Theorem 2]
\label{lemma:rfrw_bernstein}
(Bernstein inequality) Let $(g_j)_{j\geq 1}$ be a sequence of centred real-valued random variables, and $K_T=\sum_{j=1}^T g_j$. Suppose the sequence is $\alpha$-mixing with coefficient $\alpha(t) \leq -2ct$ for all $t\geq1$ and some positive constant $c$, and $\sup_{t\geq 1}||g_t||_{\infty}\leq C$ for some  positive constant $C$. Then there exists constant $C_1$ depending only on $c$,  such that for all $x>0$ 
$$\npr(|K_T|\geq x)\leq \exp\left(-\frac{C_1x^2}{v^2T+C^2+xC(\log T)^2}\right),$$
where
$$v^2=\sup_{t>0}\left(\var(g_t)+2\sum_{j>t}|\ncov(g_t,g_j)|\right).$$
\end{lemma}

\medskip

\begin{lemma}[\citealt{de1999}, Theorem 8.2.2]
\label{lemma:rfrw_martingale}
(Freedman type inequality) Let $\{N_j,\mathcal{F}_j\} $ be a martingale difference sequence with $P_T
=\sum_{j=1}^T N_j$ and the filtration $\mathcal{F}_j=\sigma(Y_i:i\leq j)$. Suppose that $\mathbb{E}(N_t|\mathcal{F}_{t-1})=0$, $\mathbb{E}(N_t^2|\mathcal{F}_{t-1})=\sigma_t^2<\infty$, and $V_T^2=\sum_{t=1}^T \sigma_t^2$. Further assume that there exists a positive constant $c$ such that, almost surely, $\mathbb{E}(|N_t|^m|\mathcal{F}_{t-1})\leq (m!/2) \sigma_t^2c^{m-2}$ for all $m>2$, or $\npr(|N_t|\leq c |\mathcal{F}_{t-1})=1$. Then, for all $x,y>0$,
\begin{eqnarray}
    \npr\left(P_T>x, V_T^2\leq y\ \text{for some}\ T\right) &\leq& \exp\left(-\frac{x^2}{y\{1+(1+2cx/y)^{1/2}\}+cx}\right) \nonumber \\ 
    &\leq& \exp\left(-\frac{x^2}{2(y+cx)}\right).\nonumber
\end{eqnarray}

\end{lemma}

\medskip

\begin{lemma}
\label{lemma:rfrw_high_xi_range}
Under Assumption \ref{assump:rfrw_all}, and every split is placed such that it puts at least a fraction $\xi$ of the data points in the parent node into each child node (which satisfies Requirement (1) in Definition \ref{def:rfrw_alpha_k_m_partition}). For any parent node $R=\bigtimes_{i=1}^p[r_i^-,r_i^+]\subseteq [0,1]^p$ and the splitting feature $j$, the split position lies in $\left[(1-\tilde{\xi})r_j^- + \tilde{\xi} r_j^+, \tilde{\xi} r_j^- + (1-\tilde{\xi})r_j^+\right]$ with probability going to 1, where $\xi>1.1\zeta^2\tilde{\xi}$, $\zeta>1$ is given in Lemma \ref{lemma:rfrw_zeta}, $\xi \in (0,1/2)$, and $\tilde{\xi} \in (0,1/2)$.
\end{lemma}

\begin{proof}
    The proof is immediate using the same argument as in Proposition 1(A) in \cite{Chen2024-sup}, thus it is omitted here.
\end{proof}

\section{Proofs of results in Section \ref{sec:theory_concentration}}\label{sec:appen_proof_concentration}

\subsection{Concentration of weighted regression trees around the population trees}

The partition-optimal (population) tree $T_{\Lambda}^*$ arises if we can train a weighted regression tree $T_{\Lambda}$ on the entire series. Under Assumption \ref{assump:rfrw_all}, we can obtain the following concentration result for a tree $T_{\Lambda}$.

\begin{theorem}
\label{theo:rfrw_rw_regression_tree}
\edef\assprefix{\thetheorem}
    Under Assumption \ref{assump:rfrw_all}, there exists a positive constant $\beta$ such that the following two statements both hold with probability at least $1-6T^{-1}$ for all sufficiently large $T$:

    \begin{enumerate}[label=(\roman*), ref=\assprefix(\roman*)]
    
    \item\label{theo:rfrw_rw_regression_tree_1} For fixed-dimensional feature space with  $p<\infty$,
    \begin{equation*}
    \label{eqn:rfrw_theo_low}
    \sup_{(x,\Lambda)\in \mathbb{R}^p \times \mathcal{V}_k}|T_\Lambda(x)-T_\Lambda^*(x)|\leq \beta\frac{(\log T)^{3/2}}{k^{1/2}}.
    \end{equation*}

    \item\label{theo:rfrw_rw_regression_tree_2} For high-dimensional feature space with $\lim \inf p/T>0$,
    \begin{equation}
    \label{eqn:rfrw_theo_high}
    \sup_{(x,\Lambda)\in \mathbb{R}^p \times \mathcal{V}_k}|T_\Lambda(x)-T_\Lambda^*(x)|\leq \beta\frac{(\log T)^{3/2}(\log p)^{1/2}}{k^{1/2}}.
    \end{equation}
\end{enumerate}
\end{theorem}

To prove Theorem \ref{theo:rfrw_rw_regression_tree}, consider the partition $\Lambda$ of $\mathbb{R}^p$, define $\bar{\Lambda}$ to be the partition of $\bar{\mathbb{R}}=\mathbb{R}\cup\{\pm\infty\}$, which is obtained by extending the feature space of each rectangles in $\Lambda$ to $\bar{\mathbb{R}}^p$. In this way, if the partition $\Lambda$ is a $k$-valid partition of $\mathbb{R}^p$, then the partition $\bar{\Lambda}$ of $\bar{\mathbb{R}}^p$ is also a $k$-valid partition. For all $x\in\mathbb{R}^p$, it follows that 
\begin{equation}
\label{eqn:rfrw_T_lambda_bar}
T_\Lambda(x)-T_\Lambda^*(x)=T_{\bar{\Lambda}}(x) - T_{\bar{\Lambda}}^*(x).
\end{equation}
In addition, for any $x\in\mathbb{R}^p$, and leaf $L=\iota_h\big(L_{\bar{\Lambda}}(x)\big)$ that contains $z=\iota_h(x)$, we define $G_T(L)$ as follows,
\begin{equation}
\label{eqn:rfrw_GTL_define}
G_T(L):=T_{\bar{\Lambda}}(x) - T_{\bar{\Lambda}}^*(x)=\frac{1}{\sum\limits_{t:Z_t\in L}\omega_t}\sum_{t:Z_t\in L}\omega_tY_t-\eta(L).
\end{equation}
Let $L_{\bar{\Lambda}}$ be the leaf that is generated by the partition $\bar{\Lambda}$ of $\bar{\mathbb{R}}^p$. Then $L_{\bar{\Lambda}}$ (that contains $X_t$) is a leaf of a $k$-valid partition of $\bar{\mathbb{R}}^p$ if and only if $\iota_h(L_{\bar{\Lambda}})$ (that contains $Z_t$) is a leaf of a $k$-valid partition of $[0,1]^p$. By (\ref{eqn:rfrw_T_lambda_bar}) and (\ref{eqn:rfrw_GTL_define}), we have
\begin{equation}
\label{eqn:rfrw_GTL_T_lambda}
\sup_{(x,\Lambda)\in \mathbb{R}^p \times \mathcal{V}_k}|T_\Lambda(x)-T_\Lambda^*(x)|=\sup_{L\in \mathcal{L}_k}|G_T(L)|,
\end{equation}
where $\mathcal{L}_k$ contains all sets of $k$-valid partitions of $[0,1]^p$. In the following, we introduce several lemmas to obtain the concentration result in (\ref{eqn:rfrw_GTL_T_lambda}) to prove Theorem \ref{theo:rfrw_rw_regression_tree}.

Note that each leaf $L\in\mathcal{L}_k$ is a rectangle subspace of $[0,1]^p$, similar to \cite{Wager2015-sup} and \cite{Davis2020-sup}, we can construct a set of rectangles that constitutes good approximations to it, such that the concentration bound of (\ref{eqn:rfrw_GTL_T_lambda}) only depends on the cardinality of the approximating rectangles. To study both the cases of fixed and high dimension $p$, we combine the results of Theorem 7 in \cite{Wager2015-sup} and Theorem 3 in \cite{Davis2020-sup} to introduce the approximating rectangles in Lemma \ref{lemma:rfrw_rectangle}.


\begin{lemma}
\label{lemma:rfrw_rectangle}
\edef\assprefix{\thelemma}
    Let $\varsigma \asymp k^{-1/2}$, and $w \asymp k/T$. There exists a collection of rectangles $\mathcal{R}_{\varsigma,w}$ such that the following two statements hold:
    \begin{enumerate}[label=(\roman*), ref=\assprefix(\roman*)]
        \item\label{lemma:rfrw_rectangle_1} For any rectangle $R\subseteq[0,1]^p$ with its volume $\normalfont \text{Leb}(R)\geq w$, there exists rectangles $R_{-},R_{+}\in \mathcal{R}_{\varsigma,w}$ such that
        \begin{equation*}
        \label{eqn:rfrw_theo_3_rectangle}
        R_{-}\subseteq R \subseteq R_{+}, \quad  \exp({-\varsigma})\normalfont \text{Leb}(R_{+})\leq \normalfont \text{Leb}(R) \leq \exp(\varsigma) \normalfont \text{Leb}(R_{-}).
        \end{equation*}
        \item\label{lemma:rfrw_rectangle_2} The cardinality $|\mathcal{R}_{\varsigma,w}|$ of the set $\mathcal{R}_{\varsigma,w}$ satisfies that 
        \begin{equation*}
        \log |\mathcal{R}_{\varsigma,w}| \lesssim
        \begin{cases}
        \log T \log p, & \text{if}\ \lim \inf p/T>0.\\
        \log T, & \text{if}\ p<\infty. \nonumber
        \end{cases}       
        \end{equation*}
    \end{enumerate}
\end{lemma}

\begin{proof}
    The proof of Statement (i) can be referred to Theorem 7 in \cite{Wager2015-sup} and Theorem 3(i) in \cite{Davis2020-sup}. For Statement (ii), for the dimension $p$ that satisfies $\lim \inf p/T>0$, we have $\log |\mathcal{R}_{\varsigma,w}|\lesssim\log T \log p$, whose proof can be found in Corollary 8 in \cite{Wager2015-sup}. When the dimension $p$ is fixed, we have $\log|\mathcal{R}_{\varsigma,w}|\lesssim\log T$, which coincides with the results of Theorem 3(ii) in \cite{Davis2020-sup}.
\end{proof}

For any leaf $L\in\mathcal{L}^w_k:=\{L\in\mathcal{L}_k:\text{Leb}(L)\geq w\}$, there exists a approximating rectangle $L_{-}^{\varsigma}$ from $\mathcal{R}_{\varsigma,w}$ with $L_{-}^{\varsigma}\subseteq L$ and $ \text{Leb}(L) \leq \exp(\varsigma) \text{Leb}(L_{-}^{\varsigma})$, such that
\begin{eqnarray}
\label{eqn:rfrw_rw_tree_general}
\sup_{L\in \mathcal{L}_k^w}|G_T(L)|&\leq& \sup_{L\in \mathcal{L}_k^w}\left |\frac{1}{\sum\limits_{t:Z_t\in L}\omega_t}\sum_{t:Z_t\in L}\omega_tY_t - \frac{1}{\sum\limits_{t:Z_t\in L_{-}^\varsigma}\omega_t}\sum_{t:Z_t\in L_{-}^\varsigma} \omega_tY_t\right| \nonumber \\
&+& \sup_{L\in \mathcal{L}_k^w}|G_T(L_{-}^\varsigma)| +\sup_{L\in \mathcal{L}_k^w}|\eta(L_{-}^\varsigma)-\eta(L)|.
\end{eqnarray}
Note that (\ref{eqn:rfrw_rw_tree_general}) applies across all leaf $L\in\mathcal{L}^w_k$. We can follow a similar argument as in Lemma 4 in \cite{Davis2020-sup} to show that for all sufficiently large $T$, the set $\mathcal{L}_k=\mathcal{L}_k^w$, which coincides with (\ref{eqn:rfrw_GTL_T_lambda}) that holds across the leaf $L\in\mathcal{L}_k$. To develop the concentration result of weighted regression trees, it remains to prove that for large $T$ and with high probability, the right-hand side of (\ref{eqn:rfrw_rw_tree_general}) is upper bounded, and this bound goes to zero. As the concentration bound in (\ref{eqn:rfrw_rw_tree_general}) only depends on the cardinality $|\mathcal{R}_{\varsigma,w}|$, we can derive the concentration result for weighted regression trees under the high-dimensional case, while the same arguments can be applied to the fixed-dimensional case, which are omitted here.

Next, we introduce Lemma \ref{lemma:rfrw_R_tmu}-\ref{lemma:rfrw_w_l}, which are used in the derivation of the bound of (\ref{eqn:rfrw_rw_tree_general}). The establishments of the bonds for the three terms on the right-hand side of (\ref{eqn:rfrw_rw_tree_general}) are presented in Lemma \ref{lemma:rfrw_weighted_average}-\ref{lemma:rfrw_eta}, respectively. 

Lemma \ref{lemma:rfrw_R_tmu} establishes the concentration of $\#R$ across all the rectangles in $\mathcal{R}_{\varsigma,w}$. In this way, we can develop the concentration inequalities for $\#L$ and $\#L_-^\varsigma$ across all $L\in \mathcal{L}_k^w$.

\begin{lemma}
\label{lemma:rfrw_R_tmu}
Suppose that Assumption \ref{assump:rfrw_all} is satisfied,  $\varsigma = k^{-1/2}$, and $w = k/(4\zeta T)$, where $\zeta>1$ is given in Lemma \ref{lemma:rfrw_zeta}. There exists a positive constant $c_1$, such that event $\mathcal{E}_1$:
$$\mathcal{E}_1:=\left \{ \sup\limits_{R\in \mathcal{R}_{\varsigma,w}:\mu(R)\geq \zeta w}\frac{\left |\#R-T\mu(R) \right |}{\{T\mu(R)\}^{1/2}} \leq c_1\log T (\log p)^{1/2}\right\},$$
holds with probability at least $1-T^{-1}$ for all sufficiently large $T$.

\end{lemma}

\begin{proof}
When the cardinality of the rectangles $\mathcal{R}_{\varsigma,w}$ satisfies that $\log |\mathcal{R}_{\varsigma,w}|\lesssim \log T \log p$ as in Lemma \ref{lemma:rfrw_rectangle_2}, the proof of Lemma \ref{lemma:rfrw_R_tmu} immediately follows the similar arguments of Lemma 3 in \cite{Davis2020-sup}.
\end{proof}

Lemma \ref{lemma:rfrw_w_l} shows that for a given node $R\in \mathcal{R}_{\varsigma,w}$, the sum of the random weights $\omega_t$ in $R$ is asymptotically identical to $\#R$, which simplifies the establishment of the bounds of (\ref{eqn:rfrw_rw_tree_general}).

\begin{lemma}
\label{lemma:rfrw_w_l}
Suppose that Assumption \ref{assump:rfrw_all} is satisfied, $\varsigma = k^{-1/2}$, and $w = k/(4\zeta T)$, where $\zeta>1$ is given in Lemma \ref{lemma:rfrw_zeta}. There exists a positive constant $c_2$, such that event $\mathcal{E}_2$: 
\begin{equation*}
\label{eqn:rfrw_w_l}
\mathcal{E}_2:=\left \{\sup\limits_{R\in \mathcal{R}_{\varsigma,w}:\mu(R)\geq \zeta w} \frac{1}{\mu(R)}\left|\sum\limits_{t:Z_t\in R}\omega_t - \#R\right|\leq c_2 \frac{(\log T\log p)^{1/2}}{k^{1/2}}\right\} \cup \mathcal{E}_1^c,
\end{equation*}
holds with probability at least $1-T^{-1}$ for all sufficiently large $T$.

\end{lemma}

\begin{proof}
   Define the random variable $S_t:=\omega_t-1$, so that $\sum\limits_{t:Z_t\in R}S_t=\sum\limits_{t:Z_t\in R}\omega_t - \#R$. For any rectangle $R\in \mathcal{R}_{\varsigma,w}$ with $\mu(R)\geq \zeta w$, the sequence $\{S_tI(Z_t \in R)\}_{t\geq 1}$ is a martingale difference sequence with respect to the filtration $\mathcal{F}_t=\sigma(Y_j:j\leq t)$, and random variable $S_t$ is independent of $\mathcal{F}_{t-1}$ with its moment satisfying that
   $$\mathbb{E}[|S_tI(Z_t \in R)|^m | \mathcal{F}_{t-1}]\leq m!c^{m-2}I(Z_t \in R), \quad m=3, 4,\ldots,$$
   for some $c>0$. By applying the Freedman type inequality (Lemma \ref{lemma:rfrw_martingale}) for unbounded summands, it holds that
    \begin{equation}
    \label{eqn:rfrw_w_l_1}
    \log\npr\left( \left|\sum\limits_{t:Z_t\in R}S_t\right|>x,\#R\leq y \right)  \lesssim - \frac{x^2T}{y/T+x},
    \end{equation}
    for any $x,y>0$. Based on Lemma \ref{lemma:rfrw_R_tmu}, we consider the following choice of $y$, 
    $$y=T\mu(R)+c_1\log T(\log p)^{1/2}\{T\mu(R)\}^{1/2},$$
    where $c_1$ is the constant given in  event $\mathcal{E}_1$. It follows from (\ref{eqn:rfrw_w_l_1}) that
    \begin{eqnarray}
    \label{eqn:rfrw_w_l_2}
    \log\npr\left( \left\{\left|\sum\limits_{t:Z_t\in R}S_t\right|>x\right\} \cap \mathcal{E}_1  \right)  \lesssim - \frac{x^2T}{\max\{\mu(R),x\}}.
    \end{eqnarray}
    We might choose a constant $C$ such that, 
    \begin{equation}
    \label{eqn:rfrw_w_l_3}
    \npr\left( \left\{\left|\sum\limits_{t:Z_t\in R}S_t\right|>x\right\} \cap \mathcal{E}_1  \right)\leq \frac{1}{|\mathcal{R}_{\varsigma,w}|T},
    \end{equation}
    where 
    \begin{equation}
    \label{eqn:rfrw_w_l_4}
    x=C\max\left\{\left(\frac{\log (|\mathcal{R}_{\varsigma,w}|T)\mu(R)}{T}\right)^{1/2},\frac{\log(|\mathcal{R}_{\varsigma,w}|T)}{T}\right\}.
    \end{equation}
    The maximum value on the right-hand side of (\ref{eqn:rfrw_w_l_4}) is reached by the first term if 
    \begin{equation*}
    \label{eqn:rfrw_w_l_5}
    k\geq 4\log (|\mathcal{R}_{\varsigma,w}|T).
    \end{equation*}
    For large $T$, (\ref{eqn:rfrw_w_l_3}) and (\ref{eqn:rfrw_w_l_4}) indicate that
    $$\npr\left( \left\{\left|\sum\limits_{t:Z_t\in R}\omega_t - \#R\right|>c_2\left(\log T\log p\right)^{1/2}\left(\frac{\mu(R)}{4T}\right)^{1/2} \right\} \cap \mathcal{E}_1 \right)\leq \frac{1}{|\mathcal{R}_{\varsigma,w}|T},$$
    for a suitable constant $c_2.$ By rearranging terms and using that $T\mu(R)\geq T \zeta k/(4\zeta T)=k/4$, it holds that 
    $$\npr\left( \left\{\frac{1}{\mu(R)}\left|\sum\limits_{t:Z_t\in R}\omega_t - \#R\right|>c_2\frac{(\log T\log p)^{1/2}}{k^{1/2}} \right\} \cap \mathcal{E}_1 \right)\leq \frac{1}{|\mathcal{R}_{\varsigma,w}|T}.$$
    We can apply the union bound over all rectangles $R \in \mathcal{R}_{\varsigma,w}$, such that $\npr(\mathcal{E}_2)\geq1-T^{-1}$ for all sufficiently large $T$. 
\end{proof}

By Lemma \ref{lemma:rfrw_w_l}, as $T\rightarrow\infty$, we can choose some $k\rightarrow\infty$, such that $(\log T \log p)^{1/2}/k^{1/2} \rightarrow 0$, then we have $\sum_{t:Z_t\in R}\omega_t \rightarrow \#R$ in probability.

Based on Lemma \ref{lemma:rfrw_R_tmu}-\ref{lemma:rfrw_w_l}, we can establish the bounds of the three terms on the right-hand side of (\ref{eqn:rfrw_rw_tree_general}), which are detailed in Lemma \ref{lemma:rfrw_weighted_average}-\ref{lemma:rfrw_eta}, respectively.

\begin{lemma}
\label{lemma:rfrw_weighted_average}
Suppose that Assumption \ref{assump:rfrw_all} is satisfied, $\varsigma = k^{-1/2}$, and $w = k/(4\zeta T)$, where $\zeta>1$ is given in Lemma \ref{lemma:rfrw_zeta}. Based on Lemma \ref{lemma:rfrw_R_tmu} and \ref{lemma:rfrw_w_l}, it holds that
\begin{eqnarray*}
    &&\sup_{L\in \mathcal{L}_k^w}\left |\frac{1}{\sum\limits_{t:Z_t\in L}\omega_t}\sum_{t:Z_t\in L}\omega_tY_t - \frac{1}{\sum\limits_{t:Z_t\in L_{-}^\varsigma}\omega_t}\sum_{t:Z_t\in L_{-}^\varsigma} \omega_tY_t\right|  \nonumber \\
    &\leq& 6(M+\max_{t=1,\ldots,T}|\epsilon_t|)\left\{\frac{\zeta^2+2c_1\log T(\log p)^{1/2}}{k^{1/2}} +O_p\left(\frac{(\log T\log p)^{1/2}}{Tk^{1/2}}\right)\right\}\nonumber,
\end{eqnarray*}
for all sufficiently large $T$, where $c_1$ is given in event $\mathcal{E}_1$, and $M$ is given in Assumption \ref{assum:rfrw_fbound}.
\end{lemma}

\begin{proof}
Note that
\begin{eqnarray}
\label{eqn:rfrw_lemma_rfrw_weighted_average_1}
    &&\sup_{L\in \mathcal{L}_k^w} \left |\frac{1}{\sum\limits_{t:Z_t\in L}\omega_t}\sum_{t:Z_t\in L}\omega_tY_t - \frac{1}{\sum\limits_{t:Z_t\in L_{-}^\varsigma}\omega_t}\sum_{t:Z_t\in L_{-}^\varsigma} \omega_tY_t\right| \nonumber \\
    &=& \sup_{L\in \mathcal{L}_k^w} \left |\frac{1}{\sum\limits_{t:Z_t\in L}\omega_t}\left(\sum_{t:Z_t\in L_{-}^\varsigma} \omega_tY_t + \sum_{t:Z_t\in L \setminus L_{-}^\varsigma}\omega_tY_t\right) - \frac{1}{\sum\limits_{t:Z_t\in L_{-}^\varsigma}\omega_t}\sum_{t:Z_t\in L_{-}^\varsigma} \omega_tY_t\right| \nonumber \\
    &\leq& \sup_{L\in \mathcal{L}_k^w} \left|\frac{\sum\limits_{t:Z_t\in L_{-}^\varsigma} \omega_tY_t \sum\limits_{t:Z_t\in L \setminus L_{-}^\varsigma}\omega_t}{\sum\limits_{t:Z_t\in L}\omega_t \sum\limits_{t:Z_t\in L_{-}^\varsigma}\omega_t}\right| + \sup_{L\in \mathcal{L}_k^w}\left|\frac{\sum\limits_{t:Z_t\in L \setminus L_{-}^\varsigma}\omega_t Y_t}{\sum\limits_{t:Z_t\in L} \omega_t}\right| \nonumber \\
     &\leq& 2(M+\max\limits_{t=1,\ldots,T}|\epsilon_t|) \left\{\sup_{L\in \mathcal{L}_k^w}\frac{\#L-\#L_{-}^\varsigma}{\#L} +O_p\left(\frac{(\log T\log p)^{1/2}}{Tk^{1/2}}\right)\right\}. 
\end{eqnarray}
As proven in Lemma 4 in \cite{Davis2020-sup}, we have
\begin{equation}
\label{eqn:rfrw_lemma_rfrw_weighted_average_2}
\sup_{L\in \mathcal{L}_k^w}\frac{\#L-\#L_{-}^\varsigma}{\#L} \leq \frac{3\{\zeta^2+2c_1\log T(\log p)^{1/2}\}}{k^{1/2}}.
\end{equation}
By (\ref{eqn:rfrw_lemma_rfrw_weighted_average_1}) and (\ref{eqn:rfrw_lemma_rfrw_weighted_average_2}), we finish the proof.
\end{proof}

\begin{lemma}
\label{lemma:rfrw_gtl_minus}
Suppose that Assumption \ref{assump:rfrw_all} is satisfied,  $\varsigma = k^{-1/2}$, and $w = k/(4\zeta T)$, where $\zeta>1$ is given in Lemma \ref{lemma:rfrw_zeta}. Based on Lemma \ref{lemma:rfrw_R_tmu} and \ref{lemma:rfrw_w_l}, it holds that
\begin{eqnarray*}
    &&\sup_{L\in \mathcal{L}_k^w}|G_T(L_{-}^\varsigma)|\nonumber\\ &\leq& \frac{4Mc_1\log T(\log p)^{1/2}}{k^{1/2}}\left\{1+O_p\left(\frac{(\log T\log p)^{1/2}}{Tk^{1/2}}\right)\right\} \nonumber \\
    &+& 2 \sup\limits_{R\in \mathcal{R}_{\varsigma,w}:\mu(R)\geq \zeta w} \frac{1}{T\mu(R)}\left| \sum\limits_{t:Z_t\in R}\omega_t\epsilon_t \right|\left\{1+O_p\left(\frac{(\log T\log p)^{1/2}}{Tk^{1/2}}\right)\right\} \nonumber \\
    &+& 2 \sup\limits_{R\in \mathcal{R}_{\varsigma,w}:\mu(R)\geq \zeta w} \frac{\left| \frac{1}{T} \sum\limits_{t:Z_t\in R} \omega_tf(X_t)-\mathbb{E}[\omega f(X)I(Z\in R)]\right|}{\mu(R)}\left\{1+O_p\left(\frac{(\log T\log p)^{1/2}}{Tk^{1/2}}\right)\right\},\nonumber\\ \nonumber
\end{eqnarray*}
for all sufficiently large $T$, where $c_1$ is is given in event $\mathcal{E}_1$, and $M$ is given in Assumption \ref{assum:rfrw_fbound}.
\end{lemma}

\begin{proof}
For any given rectangle $R$, we have 
\begin{eqnarray}
\label{eqn:rfrw_lemma_rfrw_gtl_minus_1}
&& |G_T(R)|\nonumber\\
&\leq& M\frac{\left|\sum\limits_{t:Z_t\in R}\omega_t-T\mu(R)\right|}{\sum\limits_{t:Z_t\in R}\omega_t} + \frac{1}{\sum\limits_{t:Z_t\in R}\omega_t} 
\left|\sum\limits_{t:Z_t\in R}\omega_t\epsilon_t\right|  \nonumber \\ &+& \frac{T}{\sum\limits_{t:Z_t\in R}\omega_t}\left|\frac{1}{T}\sum\limits_{t:Z_t\in R}\omega_tf(X_t)-\mathbb{E}[\omega f(X)I(Z\in R)]\right|\nonumber \\
&\leq& M\frac{|\#R-T\mu(R)|}{\#R}\left\{1+O_p\left(\frac{(\log T\log p)^{1/2}}{Tk^{1/2}}\right)\right\} \nonumber\\
&+& \frac{1}{\#R} 
\left|\sum\limits_{t:Z_t\in R}\omega_t\epsilon_t\right|\left\{1+O_p\left(\frac{(\log T\log p)^{1/2}}{Tk^{1/2}}\right)\right\}  \nonumber \\ &+& \frac{T}{\#R}\left|\frac{1}{T}\sum\limits_{t:Z_t\in R}\omega_tf(X_t)-\mathbb{E}[\omega f(X)I(Z\in R)]\right|\left\{1+O_p\left(\frac{(\log T\log p)^{1/2}}{Tk^{1/2}}\right)\right\}.
\end{eqnarray}
Denote a collection of rectangles $\mathcal{R}^{\prime} :=\{R\in\mathcal{R}_{\varsigma,w}:\mu(R)\geq\zeta w\}$. By Lemma 5 in \cite{Davis2020-sup} and Lemma \ref{lemma:rfrw_R_tmu}, when $T$ is sufficiently large, we have $\mathcal{R}^{\prime}\supseteq\{L_{-}^\varsigma:L\in\mathcal{L}^w_k\}$.

For any rectangle $R\in \mathcal{R}^{\prime}$, it holds from event $\mathcal{E}_1$ that 
\begin{equation}
\label{eqn:rfrw_lemma_rfrw_gtl_minus_3}
\#R \geq T\mu(R)\left(1-\frac{2c_1\log T(\log p)^{1/2}}{k^{1/2}}\right)\geq \frac{T\mu(R)}{2},
\end{equation}
and also 
\begin{equation}
\label{eqn:rfrw_lemma_rfrw_gtl_minus_4}
\frac{|\#R-T\mu(R)|}{\#R} \leq \frac{4c_1\log T(\log p)^{1/2}}{k^{1/2}}.
\end{equation}
By (\ref{eqn:rfrw_lemma_rfrw_gtl_minus_1})-(\ref{eqn:rfrw_lemma_rfrw_gtl_minus_4}), we complete the proof.
\end{proof}

\begin{lemma}
\label{lemma:rfrw_eta}
Suppose that Assumption \ref{assum:rfrw_h_density}-\ref{assum:rfrw_alpha_mixing} are satisfied, $\varsigma \asymp k^{-1/2}$, and $w \asymp k/T$, it holds that 
$$\sup_{L\in \mathcal{L}_k^w}|\eta(L_{-}^\varsigma)-\eta(L)|\leq 2M\zeta^2\varsigma,$$
where $\zeta>1$ is given in Lemma \ref{lemma:rfrw_zeta} and $M$ is given in Assumption \ref{assum:rfrw_fbound}.

\end{lemma}

\begin{proof}
    By Assumption \ref{assum:rfrw_fbound}, for an any leaf $L\in \mathcal{L}_k^w$, it holds that 
    \begin{eqnarray}
    \label{eqn:rfrw_lemma_rfrw_eta_1}
        &&|\eta(L_{-}^\varsigma)-\eta(L)| \nonumber \\
        &=& \left|\frac{1}{\mu(L_{-}^\varsigma)}\int_{L_{-}^\varsigma}|f(\iota_h^{-1}(z))|h_Z(z) \text{d}z-\frac{1}{\mu(L)}\int_{L}|f(\iota_h^{-1}(z))|h_Z(z)\text{d}z\right| \nonumber \\
        &\leq&  \left|\frac{1}{\mu(L)}\int_{L \setminus L_{-}^\varsigma}|f(\iota_h^{-1}(z))|h_Z(z)\text{d}z\right| + \left|\frac{\mu(L)-\mu(L_{-}^\varsigma)}{\mu(L)\mu(L_{-}^\varsigma)}\int_{L_{-}^\varsigma}|f(\iota_h^{-1}(z))|h_Z(z)\text{d}z\right| \nonumber \\
        &\leq& M \left(\frac{\mu(L)-\mu(L_{-}^\varsigma)}{\mu(L)}+\frac{\mu(L)-\mu(L_{-}^\varsigma)}{\mu(L)\mu(L_{-}^\varsigma)}\mu(L_{-}^\varsigma)\right)\nonumber \\
        &\leq& 2M\frac{\mu(L)-\mu(L_{-}^\varsigma)}{\mu(L)}.
    \end{eqnarray}
    Based on Lemma \ref{lemma:rfrw_zeta} and \ref{lemma:rfrw_rectangle_1}, we have 
    \begin{equation}
    \label{eqn:rfrw_lemma_rfrw_eta_2}
    \mu(L)-\mu(L_{-}^\varsigma)\leq\zeta\{1-\exp({-\varsigma})\}\text{Leb}(L)\leq \zeta^2\varsigma\mu(L).
    \end{equation}
    By (\ref{eqn:rfrw_lemma_rfrw_eta_1}) and (\ref{eqn:rfrw_lemma_rfrw_eta_2}), the proof is finished.
\end{proof}

Based on Lemma \ref{lemma:rfrw_weighted_average}-\ref{lemma:rfrw_eta}, we introduce the proof of Theorem \ref{theo:rfrw_rw_regression_tree} as follows.

\begin{proof}[Proof of Theorem \ref{theo:rfrw_rw_regression_tree}]

To prove Theorem \ref{theo:rfrw_rw_regression_tree}, we should obtain the bounds for Lemma \ref{lemma:rfrw_weighted_average}-\ref{lemma:rfrw_gtl_minus}. We first define event $\mathcal{E}_3$ to obtain the bounds for $\max\limits_{t=1,\ldots,T}\omega_t$, which enables the application of Bernstein inequalities for the $\alpha$-mixing process. We further define events $\mathcal{E}_4$, $\mathcal{E}_5$, and $\mathcal{E}_6$ to provide the bounds for $\max\limits_{t=1,\ldots,T}|\epsilon_t|$ from Lemma \ref{lemma:rfrw_weighted_average}, $\{T\mu(R)\}^{-1}\left| \sum\limits_{t:Z_t\in R}\omega_t\epsilon_t \right|$ from Lemma \ref{lemma:rfrw_gtl_minus}, and $\{\mu(R)\}^{-1}\left|T^{-1}\sum\limits_{t:Z_t\in R}\omega_tf(X_t)-\mathbb{E}[\omega f(X)I(Z\in R)]\right|$ from Lemma \ref{lemma:rfrw_gtl_minus}, respectively.

Suppose that Assumption \ref{assump:rfrw_all} is satisfied. Let $\varsigma = k^{-1/2}$, and $w = k/(4\zeta T)$, where $\zeta>1$ is given in Lemma \ref{lemma:rfrw_zeta}. We summarize the events $\mathcal{E}_1$ and $\mathcal{E}_2$ defined in Lemma \ref{lemma:rfrw_R_tmu} and \ref{lemma:rfrw_w_l}, respectively, and define the events $\mathcal{E}_3$, $\mathcal{E}_4$, $\mathcal{E}_5$, and $\mathcal{E}_6$ as follows: 

\begin{enumerate}[align=left]
    \item[$\mathcal{E}_1:=$] $\left \{ \sup\limits_{R\in \mathcal{R}_{\varsigma,w}:\mu(R)\geq \zeta w}\frac{1}{\{T\mu(R)\}^{1/2}}\left |\#R-T\mu(R) \right | \leq c_1\log T(\log p)^{1/2} \right\}.$
    \item[$\mathcal{E}_2:=$] $\left \{\sup\limits_{R\in \mathcal{R}_{\varsigma,w}:\mu(R)\geq \zeta w} \frac{1}{\mu(R)}\left|\sum\limits_{t:Z_t\in R}\omega_t - \#R\right|\leq c_2 \frac{(\log T\log p)^{1/2}}{k^{1/2}}\right\} \cup \mathcal{E}_1^c.$
    \item[$\mathcal{E}_3:=$] $\left \{\max_{t=1,\ldots,T}\omega_t\leq c_3 \log T\right\}.$
    \item[$\mathcal{E}_4:=$] $\left \{\max_{t=1,\ldots,T}|\epsilon_t|\leq c_4 (\log T)^{1/2}\right\}.$
    \item[$\mathcal{E}_5:=$] $\left \{\sup\limits_{R\in \mathcal{R}_{\varsigma,w}:\mu(R)\geq \zeta w} \frac{1}{T\mu(R)}\left| \sum\limits_{t:Z_t\in R}\omega_t\epsilon_t \right|\leq c_5 \frac{(\log T\log p)^{1/2}}{k^{1/2}}\right\} \cup \mathcal{E}_1^c.$
    \item[$\mathcal{E}_6:=$] $\left \{\sup\limits_{R\in \mathcal{R}_{\varsigma,w}:\mu(R)\geq \zeta w} \frac{1}{\mu(R)}\left| \frac{1}{T} \sum\limits_{t:Z_t\in R} \omega_tf(X_t)-\mathbb{E}[\omega f(X)I(Z\in R)]\right|\leq c_6 \frac{(\log T)^{3/2}(\log p)^{1/2}}{k^{1/2}}\right\} \cup \mathcal{E}_3^c.$

\end{enumerate}

The idea of the proof of Theorem \ref{theo:rfrw_rw_regression_tree} is to show that for all sufficiently large $T$, each event $\mathcal{E}_i$ holds with probability at least $1-T^{-1}$, and under the joint event $\mathcal{E}_1\cap\mathcal{E}_2\cap\mathcal{E}_3\cap\mathcal{E}_4\cap\mathcal{E}_5\cap\mathcal{E}_6$, the deviation bound between $T_{\Lambda}(x)$ and $T^\ast_{\Lambda}(x)$ holds.

Both event $\mathcal{E}_1$ and $\mathcal{E}_2$ hold with probability at least $1-T^{-1}$, which is proved in Lemma  \ref{lemma:rfrw_R_tmu} and \ref{lemma:rfrw_w_l}, respectively.



Using the union tail bound of sub-exponential and sub-gaussian random variables, we have that both events $\mathcal{E}_3$ and $\mathcal{E}_4$ hold with probability at least $1-T^{-1}$.

To prove $\npr(\mathcal{E}_5)\geq1-T^{-1}$, note that the sequence $\{\omega_t\epsilon_tI(Z_t \in R)\}_{t\geq 1}$ is a martingale difference sequence with respect to the filtration $\mathcal{F}_t=\sigma(Y_j:j\leq t)$, and random variable $\omega_t\epsilon_t$ is independent of $\mathcal{F}_{t-1}$ with its moment satisfying that,
$$\mathbb{E}[|\omega_t\epsilon_tI(Z_t \in R)|^m | \mathcal{F}_{t-1}]\leq m!c^{m-2}I(Z_t \in R), \quad m=3, 4,\ldots,$$
for some $c>0$. Similar to Lemma \ref{lemma:rfrw_w_l}, we can apply the Freedman type inequality (Lemma \ref{lemma:rfrw_martingale}) for unbounded summands by constructing that
$$\log\npr\left( \frac{1}{T}\left|\sum\limits_{t:Z_t\in R}\omega_t\epsilon_t\right|>x,\#R\leq y \right)  \lesssim - \frac{x^2T}{y/T+x},$$
for any $x, y >0$. Based on Lemma \ref{lemma:rfrw_R_tmu}, we consider the following choice of $y$, 
$$y=T\mu(R)+c_1\log T(\log p)^{1/2}\{T\mu(R)\}^{1/2},$$
where $c_1$ is the constant given in event $\mathcal{E}_1$. Similarly, we have
\begin{eqnarray}
\label{eqn:rfrw_wtet_e4}
\log\npr\left( \left\{ \frac{1}{T}\left|\sum\limits_{t:Z_t\in R}\omega_t\epsilon_t\right|>x\right\} \cap \mathcal{E}_1  \right)  \lesssim - \frac{x^2T}{\max\{\mu(R),x\}}.
\end{eqnarray}
Since the right-hand side of (\ref{eqn:rfrw_wtet_e4}) is the same as in (\ref{eqn:rfrw_w_l_2}), we can follow the same arguments in Lemma \ref{lemma:rfrw_w_l} and verify that a suitable constant $c_5$ can be chosen such that
$$\npr\left( \left\{\frac{1}{T}\left|\sum\limits_{t:Z_t\in R}\omega_t\epsilon_t\right|>c_5\left(\log T\log p\right)^{1/2}\left(\frac{\mu(R)}{4T}\right)^{1/2} \right\} \cap \mathcal{E}_1 \right)\leq \frac{1}{|\mathcal{R}_{\varsigma,w}|T},$$
for sufficiently large $T$. By rearranging terms and using that $T\mu(R)\geq T \zeta k/(4\zeta T)=k/4$, we have 
$$\npr\left( \left\{\frac{1}{T\mu(R)}\left|\sum\limits_{t:Z_t\in R}\omega_t\epsilon_t\right|>c_5\frac{(\log T\log p)^{1/2}}{k^{1/2}} \right\} \cap \mathcal{E}_1 \right)\leq \frac{1}{|\mathcal{R}_{\varsigma,w}|T}.$$
By applying the union bound over all the rectangles in $\mathcal{R}_{\varsigma,w}$, it holds that $\npr(\mathcal{E}_5)\geq1-T^{-1}$.

Next, we proceed with proving $\npr(\mathcal{E}_6)\geq1-T^{-1}$. By Assumption \ref{assum:rfrw_fbound}-\ref{assum:rfrw_alpha_mixing} and event $\mathcal{E}_1$, the sequence $\{\omega_tf(X_t)I(Z_t\in R)\}_{t\geq 1}$ is bounded and $\alpha$-mixing, with its mixing coefficients bounded by those coefficients $\{\alpha(t)\}_{t\geq 1}$ of $(X_t)_{t\geq 1}$.

We can apply the Bernstein type inequality (Lemma \ref{lemma:rfrw_bernstein}) for strongly mixing processes to obtain the bound that 
\begin{eqnarray}
\label{eqn:rfrw_bernstein_e5_1}
&&\log \npr \left(\left\{\left|\frac{1}{T}\sum\limits_{t:Z_t\in R}\omega_tf(X_t)-\mathbb{E}[\omega f(X)I(Z\in R)]\right|>x\right\}\cap \mathcal{E}_3 \right) \nonumber \\
&\lesssim& - \frac{x^2T}{\nu^2_R+(\log T)^2T^{-1}+x(\log T)^3},
\end{eqnarray}
where
\begin{eqnarray*}
    \nu^2_R&:=&\var\big(\omega_1f(X_1)I(Z_1\in R)\big) \nonumber \\&+& 2\sum_{t=1}^\infty \left|\ncov\big(\omega_{t+1}f(X_{t+1})I(Z_{t+1}\in R),\omega_1f(X_1)I(Z_1\in R)\big)\right|. \nonumber
\end{eqnarray*}
Note that 
$$\inf\{y>0:\npr(|\omega f(X)|I(Z\in R)>y)\leq u\}\leq c_3MI_{(u\leq \mu(R))}\log T,$$
where $c_3$ is given in event $\mathcal{E}_3$ and $M$ is given in Assumption \ref{assum:rfrw_fbound}. By applying Rio's covariance inequality \citep[Theorem 1.1]{rio1993}, it holds that 
$$\left|\ncov\big(\omega_{t+1}f(X_{t+1})I(Z_{t+1}\in R),\omega_1f(X_1)I(Z_1\in R)\big)\right|\leq \gamma(\log T)^2\min\{\alpha(t),\mu(R)\},$$
where $\gamma=4c_3^2M^2$. Similar to the proof of Lemma 3 in \cite{Davis2020-sup}, we have $\nu^2_R \lesssim \mu(R)(\log T)^2$. Then, (\ref{eqn:rfrw_bernstein_e5_1}) implies that
\begin{eqnarray*}
\label{eqn:rfrw_bernstein_e5_2}
&&\log \npr \left(\left\{\left|\frac{1}{T}\sum\limits_{t:Z_t\in R}\omega_tf(X_t)-\mathbb{E}[\omega f(X)I(Z\in R)]\right|>x\right\}\cap \mathcal{E}_1 \right) \nonumber \\
&&\lesssim - \frac{x^2T}{\max\{\mu(R)(\log T)^2,x(\log T)^3\}}.
\end{eqnarray*}
We can choose a constant $C$ such that,
\begin{equation}
\label{eqn:rfrw_bernstein_e5_3}
\npr \left(\left\{\left|\frac{1}{T}\sum\limits_{t:Z_t\in R}\omega_tf(X_t)-\mathbb{E}[\omega f(X)I(Z\in R)]\right|>x\right\}\cap \mathcal{E}_3 \right) \leq \frac{1}{|\mathcal{R}_{\varsigma,w}|T},
\end{equation}
where 
\begin{equation}
\label{eqn:rfrw_bernstein_e5_4}
x= C \max\left\{\left(\frac{\mu(R)(\log T)^2\log(|\mathcal{R}_{\varsigma,w}|T)}{T}\right)^{1/2},\frac{(\log T)^3\log(|\mathcal{R}_{\varsigma,w}|T)}{T}\right\}.
\end{equation}
The maximum value on the right-hand side of (\ref{eqn:rfrw_bernstein_e5_4}) is reached by the first term if 
\begin{equation*}
\label{eqn:rfrw_bernstein_e5_5}
k\geq 4(\log T)^4\log(|\mathcal{R}_{\varsigma,w}|T).
\end{equation*}
For large $T$, (\ref{eqn:rfrw_bernstein_e5_3}) and (\ref{eqn:rfrw_bernstein_e5_4}) indicate that
\begin{eqnarray}
    &&\npr \left(\left\{\left|\frac{1}{T}\sum\limits_{t:Z_t\in R}\omega_tf(X_t)-\mathbb{E}[\omega f(X)I(Z\in R)]\right|>c_6\left(\frac{\mu(R)(\log T)^{3}\log p}{4T}\right)^{1/2}\right\}\cap \mathcal{E}_3 \right) \nonumber \\
    &&\leq \frac{1}{|\mathcal{R}_{\varsigma,w}|T}, \nonumber
\end{eqnarray}
for a suitable constant $c_6$. By rearranging terms and using that $T\mu(R)\geq k/4$, we have 
\begin{eqnarray}
    &&\npr\left( \left\{\frac{1}{\mu(R)}\left| \frac{1}{T}\sum\limits_{t:Z_t\in R} \omega_tf(X_t) -\mathbb{E}[\omega f(X)I(Z\in R)]\right|> c_6 \frac{(\log T)^{3/2}(\log p)^{1/2}}{k^{1/2}}\right\} \cap \mathcal{E}_3 \right) \nonumber \\
    &&\leq \frac{1}{|\mathcal{R}_{\varsigma,w}|T}. \nonumber
\end{eqnarray}
We can apply the union bound over all the rectangles $R \in \mathcal{R}_{\varsigma,w}$, such that $\npr(\mathcal{E}_6)\geq1-T^{-1}$ for all sufficiently large $T$.

Finally, suppose that the joint event $\mathcal{E}_1\cap\mathcal{E}_2\cap\mathcal{E}_3\cap\mathcal{E}_4\cap\mathcal{E}_5\cap\mathcal{E}_6$ holds. Based on Lemma \ref{lemma:rfrw_weighted_average}-\ref{lemma:rfrw_eta} and (\ref{eqn:rfrw_rw_tree_general}), it follows that
\begin{eqnarray}
\label{eqn:rw_tree_final}
\sup_{L\in \mathcal{L}_k}|G_T(L)|&\leq & 6\{M+c_4(\log T)^{1/2}\}\left\{\frac{\zeta^2+2c_1\log T(\log p)^{1/2}}{k^{1/2}}+O_p\left(\frac{(\log T\log p)^{1/2}}{Tk^{1/2}}\right)\right\}\nonumber\\ 
 &+ &\frac{4Mc_1\log T(\log p)^{1/2}}{k^{1/2}}\left\{1+O_p\left(\frac{(\log T\log p)^{1/2}}{Tk^{1/2}}\right)\right\} \nonumber \\
 &+& \frac{2c_5(\log T\log p)^{1/2}}{k^{1/2}}\left\{1+O_p\left(\frac{(\log T\log p)^{1/2}}{Tk^{1/2}}\right)\right\} \nonumber \\
 &+& \frac{2c_6(\log T)^{3/2}(\log p)^{1/2}}{k^{1/2}}\left\{1+O_p\left(\frac{(\log T\log p)^{1/2}}{Tk^{1/2}}\right)\right\} \nonumber \\
& +& \frac{2M\zeta^2}{k^{1/2}} . 
\end{eqnarray}
A suitable constant $\beta$ can be chosen such that (\ref{eqn:rfrw_theo_high}) is satisfied on $\mathcal{E}_1\cap\mathcal{E}_2\cap\mathcal{E}_3\cap\mathcal{E}_4\cap\mathcal{E}_5\cap\mathcal{E}_6$, which completes the proof.

\end{proof}

\subsection{Proof of Theorem \ref{theo:rfrw}}

As the concentration bound of weighted regression trees in Theorem \ref{theo:rfrw_rw_regression_tree} holds uniformly across all the $k$-valid trees, in Theorem \ref{theo:rfrw}, we can extend the concentration result of trees to that of the forests. We similarly present the concentration result of RF-RW under the high-dimensional case. The same arguments can be applied to the fixed-dimensional case, which are omitted here.

\begin{proof}[Proof of Theorem \ref{theo:rfrw}]
    By the triangular inequality and Theorem \ref{theo:rfrw_rw_regression_tree}, it holds that 
\begin{eqnarray*}
    &&\sup_{(x,\boldsymbol\Lambda)\in \mathbb{R}^p \times \mathcal{W}_k}|H_{\boldsymbol\Lambda}(x)-H_{\boldsymbol\Lambda}^*(x)| \nonumber \\
    &=&\frac{1}{B}\sup_{(x,\boldsymbol\Lambda)\in \mathbb{R}^p \times \mathcal{W}_k}\left |\sum_{b=1}^B T_{\Lambda_b}(x)-\sum_{b=1}^B T_{\Lambda_b}^*(x)\right | \nonumber \\
    &=& \frac{1}{B} \sup_{(x,\boldsymbol\Lambda)\in \mathbb{R}^p \times \mathcal{W}_k}\left| T_{\Lambda_1}(x) - T_{\Lambda_1}^*(x)+\ldots+T_{\Lambda_B}(x) - T_{\Lambda_B}^*(x)\right|  \nonumber \\
    &\leq& \frac{1}{B} \left (\sup_{(x,\Lambda)\in \mathbb{R}^p \times \mathcal{V}_k}\left| T_{\Lambda_1}(x) - T_{\Lambda_1}^*(x)\right  | + \ldots +\sup_{(x,\Lambda)\in \mathbb{R}^p \times \mathcal{V}_k}\left|T_{\Lambda_B}(x) - T_{\Lambda_B}^*(x)\right| \right) \nonumber \\
    &\leq& \beta\frac{(\log T)^{3/2}(\log p)^{1/2}}{k^{1/2}}, \nonumber
\end{eqnarray*}
which finishes the proof.
\end{proof}


\section{Proofs of results in Section \ref{sec:theory_consistency_fixed}}\label{sec:appen_proof_fixed_consistency}

\subsection{Proof of Theorem \ref{theo:rfrw_consistent_1} and \ref{theo:rfrw_consistent_2}}

For each leaf $L:=\bigtimes_{i=1}^p[r_i^-,r_i^+]\subseteq[0,1]^p$ that contains $z\in[0,1]^p$, denote the diameter of $L$ as: 
$$\text{diam}(L):=\sup_{z^{\prime}, z^{\prime\prime}\in L}||z^{\prime}-z^{\prime\prime}||=\left(\sum_{i=1}^p\left|r_i^+ - r_i^-\right|^2\right)^{1/2}.$$
Similarly, the diameter of leaf $L_{\Lambda} \subseteq \mathbb{R}^p$ that contains $x\in\mathbb{R}^p$ is denoted as $\text{diam}(L_{\Lambda}):=\sup_{x^{\prime}, x^{\prime\prime} \in L_{\Lambda}}||x^{\prime}-x^{\prime\prime}||$. To obtain the consistency of RF-RW under fixed dimension $p$, we restrict the maximal diameter of each leaf node to zero as $T\rightarrow \infty$, which is shown in Lemma \ref{lemma:rfrw_diam}.

\begin{lemma}
\label{lemma:rfrw_diam}
    Suppose that Assumption \ref{assump:rfrw_all} and \ref{assum:rfrw_leave_alpha} are satisfied, and the partition $\Lambda\in \mathcal{V}_{\xi,k,m}$. For any $x\in\mathbb{R}^p$, as $T\rightarrow \infty$, \normalfont{$\text{diam}\big(L_{\Lambda}(x)\big)\rightarrow 0$} with probability one.
    
\end{lemma}

\begin{proof}
    The proof of Lemma \ref{lemma:rfrw_diam} can be referred to Lemma 6 in \cite{Davis2020-sup}.
\end{proof}

Based on Theorem \ref{theo:rfrw} and Lemma \ref{lemma:rfrw_diam}, we introduce the proof of Theorem \ref{theo:rfrw_consistent_1} and \ref{theo:rfrw_consistent_2} as follows.

\begin{proof}[Proof of Theorem \ref{theo:rfrw_consistent_1} and \ref{theo:rfrw_consistent_2}]
By Theorem \ref{theo:rfrw}, for partitions $\Lambda_1,\ldots,\Lambda_B\in \mathcal{V}_{\xi,k,m}$ and $T\rightarrow\infty$, it holds that
\begin{eqnarray*}
    \left|\hat{H}_T(x)-f(x)\right|&=&\left|\hat{H}_T(x)-f(x)-H^*_T(x)+H^*_T(x)\right|\nonumber \\ &\leq&  \max_{b=1,\ldots,B}\left|T^*_{\Lambda_b}(x)-f(x)\right|+o_p(1), \nonumber
\end{eqnarray*}
and similarly,
$$\left|\hat{H}_T(X)-f(X)\right|\leq  \max_{b=1,\ldots,B}\left|T^*_{\Lambda_b}(X)-f(X)\right|+o_p(1).$$
By Assumption \ref{assum:rfrw_lip}, we have 
$$\left|T^*_{\Lambda}(x)-f(x)\right|\leq \frac{\mathbb{E}_\Lambda\left[|f(X)-f(x)|I\big(X\in L_{\Lambda}(x)\big)\right]}{\npr_\Lambda\big(X\in L_{\Lambda}(x)\big)}\leq C_{f}\text{diam}\big(L_{\Lambda}(x)\big). $$
By Lemma \ref{lemma:rfrw_diam}, as $T\rightarrow \infty$, we have $\text{diam}\big(L_{\Lambda}(x)\big) \rightarrow 0$ with probability one. It follows that $T^*_{\Lambda}(x) \rightarrow f(x)$ almost surely and correspondingly $\hat{H}_T(x)\rightarrow f(x)$ in probability for all $x \in \mathbb{R}^p$.
\end{proof}

\subsection{Proof of Theorem \ref{theo:rfrw_consistent_3}}

Based on the proof of Theorem \ref{theo:rfrw_consistent_1} and \ref{theo:rfrw_consistent_2}, we can extend the current pointwise consistency result of $\hat{H}_T$ to uniform consistency, by showing that $\sup_{x\in \mathbb{R}^p}\left\{\text{diam}\big(L_{\Lambda}(x)\big)\right\}\overset{P}{\longrightarrow}0$ as $T\rightarrow\infty$. Specifically, we can consider the following moment condition in Lemma \ref{lemma:rfrw_diam_x_uniform}.

\begin{lemma}
\label{lemma:rfrw_diam_x_uniform}
    Suppose that Assumption \ref{assump:rfrw_all} and \ref{assum:rfrw_leave_alpha} are satisfied, and the partition $\Lambda\in \mathcal{V}_{\xi,k,m}$. If $\xi\in(0,0.2]$, it holds that
    $$\mathbb{E}\left[\sup_{x\in \mathbb{R}^p}\left\{\normalfont\text{diam}\big(L_{\Lambda}(x)\big)\right\}\right]=O\left(\left(T/m\right)^{-\gamma}\right),$$
    where 
    $$\gamma=\frac{\rho}{2}\frac{\log\left((1-\xi)^{-1}\right)}{\log(\xi^{-1})},$$
    with $m$, $\rho$, and $\xi$ are given in Definition \ref{def:rfrw_alpha_k_m_partition} of ($\xi,k,m$)-valid partition.
\end{lemma}

\begin{proof}
    Let $s$ be the number of splits leading to the leaf $L\subseteq[0,1]^p$, and $s_j$ be the number of these splits along the $j$-th feature for $j\in\{1,\ldots,p\}$. Definition \ref{def:rfrw_alpha_k_m_partition} implies that $s_j\rightarrow\infty$, and there holds that 
    $$\text{Binom}(s,\pi_j)\sim s_j \overset{d}{\geq}\bar{s}_j\sim \text{Binom}(\bar{s},\rho),$$
    where $\bar{s}={\log(T/m)}/{\log(\xi^{-1})}$, and $\pi_j$ is the probability that a node is split along the $j$-th feature.
    
    The main idea of this step is similar to Lemma 3 in \cite{Shiraishi2024-sup} and Lemma 2 in \cite{Wager2018-sup}. By Chernoff's inequality and for any $\delta\in(0,1)$ with $\delta\downarrow0$, we have
    \begin{eqnarray}
    \label{eqn:rfrw_diam_chernoff}
        \npr\left(s_j\leq (1+\delta)\frac{\bar{s}\rho}{2}\right)&\leq&\npr\left(\bar{s}_j\leq (1+\delta)\frac{\bar{s}\rho}{2}\right)=\npr\left(\bar{s}_j\leq \left(1-\frac{1-\delta}{2}\right)\bar{s}\rho\right) \nonumber \\
        &\leq& \exp\left(-\frac{1}{2}\left(\frac{1-\delta}{2}\right)^2 \bar{s}\rho\right) \nonumber \\
        &=&O\left(\exp\left(-\frac{\bar{s}\rho}{8}\right)\right),
    \end{eqnarray}
    where 
    $$\frac{\bar{s}\rho}{2}=\gamma\frac{\log(T/m)}{\log((1-\xi)^{-1})}.$$ 
    If $\xi\in(0,0.2]$, we have
    \begin{equation}
    \label{eqn:rfrw_xi_0.2}
        -\frac{\bar{s}\rho}{8}<-\gamma\log(T/m). 
    \end{equation}
    By (\ref{eqn:rfrw_diam_chernoff}) and (\ref{eqn:rfrw_xi_0.2}), it holds that
    \begin{eqnarray}
    \label{eqn:rfrw_diam_chernoff_final}
    \npr\left(s_j\leq \gamma(1+\delta)\frac{\log(T/m)}{\log\left((1-\xi)^{-1}\right)}\right)&=&O\left(\exp\left(-\gamma\log\frac{T}{m}\right)\right) \nonumber \\
    &=& O\left(\left(T/m\right)^{-\gamma}\right).
    \end{eqnarray}
    Denote that $\text{diam}_i(L)=\text{Leb}(L_i)=|r_i^+-r_i^-|$ for any leaf $L=\bigtimes_{i=1}^p[r_i^-,r_i^+]$. Let $d_j\in\{1,\ldots,s_j\}$ be the $d_j$-th split among the $s_j$ splits. Define $L^{d_j-}=\bigtimes_{i=1}^pL_i^{d_j-}$ and $L^{d_j}=\bigtimes_{i=1}^pL_i^{d_j}$ as the leaf before and after the $d_j$-th splitting, respectively. Following the similar proof of Lemma 3 in \cite{Shiraishi2024-sup} and Lemma 6 in \cite{Davis2020-sup}, the $\text{diam}_j(L)$ can be written as
    $$\text{diam}_j(L)=\prod_{d_j=1}^{s_j} \frac{\text{diam}_j(L^{d_j})}{\text{diam}_j(L^{d_j-})}.$$
    In addition, for each leaf $L^d=\bigtimes_{i=1}^pL_i^d$ with $L_i^d=[r_i^{d-},r_i^{d+}]\subset (0,1)$ and $d\in\{1,\ldots,s\}$ is the $d$-th split among the $s$ splits, there exists $L_{\Lambda}^d=\bigtimes_{i=1}^pL_{\Lambda_i}^d$ with $L_{\Lambda_i}^d=[v_i^{d-},v_i^{d+}]\subseteq \mathbb{R}$, such that $L_i^d=[F_h(v_i^{d-}),F_h(v_i^{d+})]=:F_h(L_{\Lambda_i}^d)$. It implies that 
    $$\frac{\text{diam}_j(L^{d_j})}{\text{diam}_j(L^{d_j-})}=\frac{\text{Leb}\left(F_h(L_{\Lambda_j}^{d_j})\right)}{\text{Leb}\left(F_h(L_{\Lambda_j}^{d_j-})\right)}=\frac{\text{Leb}\left(\iota_h(L_{\Lambda}^{d_j})\right)}{\text{Leb}\left(\iota_h(L_{\Lambda}^{d_j-})\right)}=1-\frac{\text{Leb}\left(\iota_h(L_{\Lambda}^{d_j-}\setminus L_{\Lambda}^{d_j})\right)}{\text{Leb}\left(\iota_h(L_{\Lambda}^{d_j-})\right)},$$
    where $\iota_h(L_{\Lambda}^{d})=\bigtimes_{i=1}^pF_h(L_{\Lambda_i}^d)$. By Lemma \ref{lemma:rfrw_zeta}, for all sufficiently large $T$, we have $$0<\zeta^{-2}\xi\leq\frac{\text{diam}_j(L^{d_j})}{\text{diam}_j(L^{d_j-})}\leq 1-\zeta^{-2}\xi<1.$$
    Given that $s_j\rightarrow\infty$, by Glivenko-Cantelli theorem for ergodic process, it holds that
    \begin{eqnarray*}
    \frac{1}{s_j}\log(\text{diam}_j(L)) &=& \frac{1}{s_j}\sum_{d_j=1}^{s_j}\log\left(\frac{\text{diam}_j(L^{d_j})}{\text{diam}_j(L^{d_j-})}\right) \nonumber \\
    &\overset{a.s.}\rightarrow& \mathbb{E}\left[\log\left(\frac{\text{diam}_j(L^1)}{\text{diam}_j(L^{1-})}\right)\right] \in [\log(\zeta^{-2}\xi),\log(1-\zeta^{-2}\xi)], \\
    \end{eqnarray*}
    which implies that there exists some $\delta\in(0,1)$ with $\delta\downarrow0$ such that
    \begin{equation}
    \label{eqn:rfrw_diam_L_j}
    \text{diam}_j(L)\leq (1-\xi)^{s_j/(1+\delta)},
    \end{equation}
    with probability one. By (\ref{eqn:rfrw_diam_chernoff_final}) and (\ref{eqn:rfrw_diam_L_j}),
    \begin{eqnarray}
    \npr\left(\text{diam}_j(L)\geq (T/m)^{-\gamma}\right) &\leq& \npr\left((1-\xi)^{s_j/(1+\delta)}\geq (T/m)^{-\gamma}\right)\nonumber \\
    &=& \npr\left(s_j\leq \gamma(1+\delta)\frac{\log(T/m)}{\log((1-\xi)^{-1})}\right) \nonumber \\
    &=& O\left(\left(T/m\right)^{-\gamma}\right).
    \end{eqnarray}
    Denote that $\text{diam}_{j*}(L):=\max_{j=1,\ldots,p}\text{diam}_j(L)$, it holds that
    \begin{eqnarray*}
    \mathbb{E}[\text{diam}(L)] &\leq& \sqrt{p}\big (\mathbb{E}[\text{diam}_{j*}(L)I_{(\text{diam}_{j*}(L)\geq (T/m)^{-\gamma})}] 
    + \mathbb{E}[\text{diam}_{j*}(L)I_{(\text{diam}_{j*}(L)< (T/m)^{-\gamma})}] \big) \nonumber \\
    && =O\left(\left(T/m\right)^{-\gamma}\right).
    \end{eqnarray*}
    Note that the above result does not depend on the observed data $z=\iota_h(x)$, we have
    $$\mathbb{E}\left[\sup_{z\in [0,1]^p}\left\{\normalfont\text{diam}\big(L(z)\big)\right\}\right]=O\left(\left(T/m\right)^{-\gamma}\right).$$
    Moreover, since the mapping $\iota_h$ is one-to-one, following the similar proof of Corollary 1 in \cite{Shiraishi2024-sup}, it holds that
    $$\mathbb{E}\left[\sup_{x\in \mathbb{R}^p}\left\{\normalfont\text{diam}\big(L_{\Lambda}(x)\big)\right\}\right]=O\left(\left(T/m\right)^{-\gamma}\right).$$
\end{proof}

The proof of Theorem \ref{theo:rfrw_consistent_3} immediately follows from Lemma \ref{lemma:rfrw_diam_x_uniform} and is shown below.

\begin{proof}[Proof of Theorem \ref{theo:rfrw_consistent_3}]
As $T\rightarrow \infty$, Lemma \ref{lemma:rfrw_diam_x_uniform} implies that $\sup_{x\in \mathbb{R}^p}\left\{\text{diam}\big(L_{\Lambda}(x)\big)\right\}\overset{P}{\longrightarrow}0$ as $T\rightarrow\infty$, which completes the proof.
\end{proof}

\section{Proofs of results in Section \ref{sec:theory_consistency_high}}\label{sec:appen_proof_high_consistency}

\subsection{Proof of Theorem \ref{theo:rfrw_high_h_optimal_fx}}

We first introduce some implications of Assumption \ref{assump:rfrw_high_lip}, which are used in the following proofs. 

For any $z$, $z^\prime$, $z^{\prime\prime}$, and $z^{\prime\prime\prime}$ , the joint density $h_Q(y,[z]_Q)$ of $(Y,[Z]_Q)$ is $L(y)$-Lipschitz with respective to $[z]_Q$: 
$$\left|h_Q(y,[z]_Q)-h_Q(y,[z^\prime]_Q)\right|\leq L(y) \left \Vert [z]_Q-[z^\prime]_Q\right \Vert.$$ 
The marginal density $h_{[Z]_Q}([z]_Q)$ of $[Z]_Q$ is $L_1$-Lipschitz: $$\left|h_{[Z]_Q}([z]_Q)-h_{[Z]_Q}([z^\prime]_Q)\right|\leq L_1 \left \Vert [z]_Q-[z^\prime]_Q\right \Vert.$$
The marginal density $h_Z([z^\prime]_Q,[z^{\prime\prime}]_{Q^C})$ of $Z$ is also $L_1$-Lipschitz: 
$$\left|h_Z([z^{\prime}]_Q,[z^{\prime\prime\prime}]_{Q^C})-h_Z([z^{\prime\prime}]_Q,[z^{\prime\prime\prime}]_{Q^C})\right|\leq L_1 \left \Vert [z^\prime]_Q-[z^{\prime\prime}]_Q\right \Vert.$$ 


\begin{proof}[Proof of Theorem \ref{theo:rfrw_high_h_optimal_fx}]
    For predicting $f(\iota_h^{-1}(z))$ using $\mathbb{E}[Y|Z\in L(z)]$, and $z$, $z^\prime \in [0,1]^p$ we have
\begin{eqnarray}
\label{eqn:rfrw_proof_t3}
\mathbb{E}[Y|Z\in L(z)] &=& \frac{\int_{z\in L(z)}\int_{-\infty}^{\infty}yh_{Y|Z}(y|z)h_Z(z)dydz}{\mu\big(L(z)\big)}\nonumber \\
&=& \frac{\int_{z\in L(z)}h_Z(z)\int_{-\infty}^{\infty}yh_{Y|[Z]_Q}(y|[z]_Q)dydz}{\mu\big(L(z)\big)}\nonumber \\
&=& \frac{\int_{z\in L(z)}h_Z(z)\int_{-\infty}^{\infty}yh_{Y|[Z]_Q}(y|[z^\prime]_Q)dydz}{\mu\big(L(z)\big)} + \nonumber \\
&& \frac{\int_{z\in L(z)}h_Z(z)\int_{-\infty}^{\infty}y\left\{h_{Y|[Z]_Q}(y|[z]_Q)-h_{Y|[Z]_Q}(y|[z^\prime]_Q)\right\}dydz}{\mu\big(L(z)\big)}  \nonumber \\
&=& \mathbb{E}[Y|Z=[z]_Q]  \nonumber \\
&+& \frac{\int_{z\in L(z)}h_Z(z)\int_{-\infty}^{\infty}y\left\{h_{Y|[Z]_Q}(y|[z]_Q)-h_{Y|[Z]_Q}(y|[z^\prime]_Q)\right\}dydz}{\mu\big(L(z)\big)}. \qquad \quad
\end{eqnarray}
Note that
\begin{eqnarray*}
    &&\left|h_{Y|[Z]_Q}(y|[z]_Q)-h_{Y|[Z]_Q}(y|[z^\prime]_Q)\right| \nonumber \\ &\leq&\left|\frac{h_Q(y,[z]_Q)}{h_{[Z]_Q}([z]_Q)}-\frac{h_Q(y,[z^\prime]_Q)}{h_{[Z]_Q}([z]_Q)}\right| + \left|\frac{h_Q(y,[z^\prime]_Q)}{h_{[Z]_Q}([z]_Q)}-\frac{h_Q(y,[z^\prime]_Q)}{h_{[Z]_Q}([z^\prime]_Q)}\right| \nonumber\\
    &\leq& \frac{L(y)\left \Vert [z]_Q-[z^\prime]_Q\right \Vert}{h_{[Z]_Q}([z]_Q)} + \frac{h_Q(y,[z^\prime]_Q)\left| h_{[Z]_Q}([z]_Q) - h_{[Z]_Q}([z^\prime]_Q) \right|}{h_{[Z]_Q}([z]_Q)h_{[Z]_Q}([z^\prime]_Q)}\nonumber \\
    &\leq& \zeta L(y) \left \Vert [z]_Q-[z^\prime]_Q \right \Vert + \frac{L_1\zeta h_Q(y,[z^\prime]_Q)\left \Vert [z]_Q-[z^\prime]_Q \right \Vert}{h_{[Z]_Q}([z^\prime]_Q)},\nonumber
\end{eqnarray*}
where $\zeta$ is given in Lemma \ref{lemma:rfrw_zeta}. Since $\left \Vert [z]_Q-[z^\prime]_Q \right \Vert\leq \text{diam}\big([L(z)]_Q\big)$, we can obtain from the second term on the right-hand side of (\ref{eqn:rfrw_proof_t3}) that
\begin{eqnarray}
\label{eqn:rfrw_proof_t3_2}
    &&\left|\frac{\int_{z\in L(z)}h_Z(z)\int_{-\infty}^{\infty}y\left\{h_{Y|[Z]_Q}(y|[z]_Q)-h_{Y|[Z]_Q}(y|[z^\prime]_Q)\right\}dydz}{\mu\big(L(z)\big)}\right| \nonumber \\
    &\leq& \frac{\int_{z\in L(z)}h_Z(z)\int_{-\infty}^{\infty}|y|\zeta L(y) \text{diam}\big([L(z)]_Q\big)dydz}{\mu\big(L(z)\big)} \nonumber \\
    &+&  \frac{\int_{z\in L(z)}h_Z(z)\left|\int_{-\infty}^{\infty}y\left\{h_{Y|[Z]_Q}(y|[z^\prime]_Q)\right\}L_1\zeta\text{diam}\big([L(z)]_Q\big)dydz\right|}{\mu\big(L(z)\big)} \nonumber \\
    &\leq& \zeta L_2 \text{diam}\big([L(z)]_Q\big) +\zeta ML_1 \text{diam}\big([L(z)]_Q\big) \nonumber \\
    &\leq& C_{SIG}\text{diam}\big([L(z)]_Q\big), 
\end{eqnarray}
for some constant $C_{SIG}>0$. By (\ref{eqn:rfrw_proof_t3}) and (\ref{eqn:rfrw_proof_t3_2}), it holds that 
$$\big| \mathbb{E}[Y|Z\in L(z)]-E[Y|Z=z]\big|\leq C_{SIG}\text{diam}\big([L(z)]_Q\big).$$
Since trees are invariant to monotone transformations of the feature spaces, $L(z)=\iota_h(L_{\bar\Lambda}(x))=\iota_h(L_{\Lambda}(x))$ for $x\in\mathbb{R}^p$, and $z=\iota_h(x)$, we have
$$\left| T_{\Lambda}^*(x)-f(x) \right|\leq C_{\text{SIG}}\normalfont{\text{diam}}\left([\iota_h\big(L_{\Lambda}(x)\big)]_Q\right).$$
It follows that
$$\left| H_{\boldsymbol \Lambda}^*(x)-f(x) \right|\leq C_{\text{SIG}}\mathbb{E}_\Theta\normalfont{\text{diam}}\left([\iota_h\big(L_{\Lambda}(x)\big)]_Q\right),$$
where $\mathbb{E}_{\Theta}$ is the expectation with respect to $\Theta$. It further implies that
$$\sup_{(x,\boldsymbol\Lambda)\in \mathbb{R}^p \times \mathcal{W}_k}|H_{\boldsymbol\Lambda}^*(x)-f(x)|\leq C_{\text{SIG}}\sup_{x\in\mathbb{R}^p}\mathbb{E}_\Theta\normalfont{\text{diam}}\left([\iota_h\big(L_{\Lambda}(x)\big)]_Q\right),$$
which finishes the proof.
\end{proof}

\subsection{Split criterion}

To obtain the consistency of high-dimensional RF-RW, we investigate the split criterion to control the splitting of informative and uninformative features.

For a node $A=\bigtimes_{i=1}^p[v_i^-,v_i^+]\subseteq \mathbb{R}^p$ in the tree that contains $x\in\mathbb{R}^p$, the splitting feature $j$, and the relative position $\tau$, we denote the empirical split criterion (\ref{eqn:rfrw_weight_splitting}) as $\hat{l}_\Lambda(A,j,\tau)$, which can be written as
\begin{equation*}
\label{eqn:rfrw_tree_split_min_empirical}
\hat{l}_\Lambda(A,j,\tau) = \frac{\sum_{X_t\in A_\text{left}} \omega_t}{\sum_{X_t\in A} \omega_t}\hat{\sigma}^2(Y|X \in A_\text{left}) + \frac{\sum_{X_t\in A_\text{right}} \omega_t}{\sum_{X_t\in A} \omega_t}\hat{\sigma}^2(Y|X \in A_\text{right}),
\end{equation*}
where $\hat{\sigma}^2(Y|X \in A) = \hat{Y}^2(A) - \{\hat{Y}(A)\}^2$, with
$$\hat{Y}^2(A)=\frac{1}{\sum_{X_t\in A} \omega_t}\sum_{X_t\in A} \omega_tY_t^2, \quad \hat{Y}(A)=\frac{1}{\sum_{X_t\in A} \omega_t}\sum_{X_t\in A} \omega_tY_t.$$
The population split criterion $l_\Lambda(A,j,\tau)$ of $\hat{l}_\Lambda(A,j,\tau)$ is denoted as
\begin{eqnarray}
\label{eqn:rfrw_tree_split_min_population}
    l_\Lambda(A,j,\tau) &=& \npr(X \in A_\text{left}|X\in A) \sigma^2(Y|X \in A_\text{left}) \nonumber\\
      &+& \npr(X \in A_\text{right}|X\in A) \sigma^2(Y|X \in A_\text{right}), \quad
\end{eqnarray}
where $\sigma^2(Y|X \in A)=\mathbb{E}[Y^2|X\in A]-(\mathbb{E}[Y|X\in A]
)^2$.

Let a node $R=\bigtimes_{i=1}^p[r_i^-,r_i^+]\subseteq [0,1]^p$ in the tree that contains $z=(z^{(1)},\ldots,z^{(p)})^\top \in[0,1]^p$ is transformed from a node $A=\bigtimes_{i=1}^p[v_i^-,v_i^+]\subseteq \mathbb{R}^p$. Based on the splitting feature $j\in \mathcal{M}_{\text{try}}$, and the relative position $\tau \in (0,1)$, we define the left child and right child node of $R$ as $R_\text{left}:=\{z\in R: z^{(j)} \leq r_{j,\tau}\}$, and $R_\text{right}:=\{z\in R:z^{(j)} > r_{j,\tau}\}$, respectively, where $r_{j,\tau}:=\tau r_j^- + (1-\tau)r_j^+$ is the split position.

We can similarly define the empirical split criterion $\hat{l}(R,j,\tau)$ and the population counterpart $l(R,j,\tau)$ with respect to $Z_t\in[0,1]^p$ as
\begin{equation*}
\label{eqn:rfrw_tree_split_min_empirical_z}
\hat{l}(R,j,\tau) = \frac{\sum_{Z_t\in R_\text{left}} \omega_t}{\sum_{Z_t\in R} \omega_t}\hat{\sigma}^2(Y|Z \in R_\text{left}) + \frac{\sum_{Z_t\in R_\text{right}} \omega_t}{\sum_{Z_t\in R} \omega_t}\hat{\sigma}^2(Y|Z \in R_\text{right}).
\end{equation*}
\begin{eqnarray}
\label{eqn:rfrw_tree_split_min_population_z}
l(R,j,\tau) &=& \npr(Z \in R_\text{left}|Z\in R) \sigma^2(Y|Z \in R_\text{left})\nonumber\\
&+& \npr(Z \in R_\text{right}|Z\in R) \sigma^2(Y|Z \in R_\text{right}).
\end{eqnarray}
Equivalently, the population split criterion for $Z_t \in [0,1]^p$ can also be expressed as the variance reduction version:
\begin{eqnarray}
\label{eqn:rfrw_tree_split_max_population}
    \mathcal{I}(R,j,\tau) = \sigma^2(Y|Z \in R)- l(R,j,\tau).
\end{eqnarray}
Similarly, the empirical split criterion is denoted as $\hat{\mathcal{I}}(R,j,\tau)=\hat\sigma^2(Y|Z \in R)- \hat l(R,j,\tau)$. If we have access to the whole population, and since the mapping between $X_t$ and $Z_t$ is one-to-one, then minimizing (\ref{eqn:rfrw_tree_split_min_population}) is equivalent to minimizing (\ref{eqn:rfrw_tree_split_min_population_z}), or maximizing (\ref{eqn:rfrw_tree_split_max_population}).

Under Assumption \ref{assump:rfrw_all}, we can show that the empirical split criterion constitutes a good approximation to the population counterpart in Lemma \ref{lemma:rfrw_high_emp_l_pop_l}, which helps control the splitting process.

\begin{lemma}
\label{lemma:rfrw_high_emp_l_pop_l}
Suppose Assumption \ref{assump:rfrw_all} is satisfied, it holds that
\begin{equation}
\label{eqn:rfrw_high_emp_l_pop_l}
\sup_{R\in \mathcal{R}_{\varsigma,w}:\mu(R)\geq \zeta w,j,\tilde{\xi}\leq \tau \leq 1-\tilde{\xi}}\left|\hat{l}(R,j,\tau)-l(R,j,\tau)\right|=O_p\left(\frac{(\log T)^{3/2}(\log p)^{1/2}}{k^{1/2}}\right).
\end{equation}
\end{lemma}

\begin{proof}

From event $\mathcal{E}_1$, we can obtain the following result by rearranging terms:
\begin{equation*}
\label{eqn:rfrw_R_tmuR_prob}
\sup\limits_{R\in \mathcal{R}_{\varsigma,w}:\mu(R)\geq \zeta w} \frac{1}{\mu(R)}\left|\frac{\#R}{T}-\mu(R)\right| \leq c_2\frac{\log T (\log p)^{1/2}}{k^{1/2}},
\end{equation*}
holds with with probability at least $1-T^{-1}$ for all sufficiently large $T$. As $T\rightarrow\infty$, we can choose some $k\rightarrow\infty$, such that $\log T (\log p)^{1/2}/k^{1/2} \rightarrow 0$, which implies that $\#R/T \rightarrow \mu(R)$ in probability.

Since we also have $\sum_{t:Z_t\in R}\omega_t \rightarrow \#R$ in probability by Lemma \ref{lemma:rfrw_w_l} as $T\rightarrow\infty$, it holds that
\begin{eqnarray*}
\label{eqn:rfrw_high_emp_l_pop_l_5}
&&\left|\hat{l}(R,j,\tau)-l(R,j,\tau)\right| \nonumber \\
&\leq& \left|\frac{\sum_{Z_t\in R_\text{left}} \omega_t}{\sum_{Z_t\in R} \omega_t}\hat{\sigma}^2(Y|Z \in R_\text{left}) - \npr(Z \in R_\text{left}| Z \in R) \sigma^2(Y|Z \in R_\text{left})\right| \nonumber \\
&+&  \left|\frac{\sum_{Z_t\in R_\text{right}} \omega_t}{\sum_{Z_t\in R} \omega_t}\hat{\sigma}^2(Y|Z \in R_\text{right}) - \npr(Z \in R_\text{right}| Z \in R) \sigma^2(Y|Z \in R_\text{right})\right| \nonumber \\
&\leq& \left|\frac{\#R_\text{left}}{\#R}\hat{\sigma}^2(Y|Z \in R_\text{left}) - \npr(Z \in R_\text{left}| Z \in R) \sigma^2(Y|Z \in R_\text{left})\right| \nonumber \\
&+&  \left|\frac{\#R_\text{right}}{\#R}\hat{\sigma}^2(Y|Z \in R_\text{right}) - \npr(Z \in R_\text{right}| Z \in R) \sigma^2(Y|Z \in R_\text{right})\right| \nonumber \\
&\leq& \left|\npr(Z \in R_\text{left}| Z \in R)\left \{\hat{\sigma}^2(Y|Z \in R_\text{left}) -\sigma^2(Y|Z \in R_\text{left})\right\}\right| \nonumber \\
&+&  \left|\npr(Z \in R_\text{right}| Z \in R)\left\{\hat{\sigma}^2(Y|Z \in R_\text{right}) - \sigma^2(Y|Z \in R_\text{right})\right\}\right| \nonumber \\
&+& O_p\left(\frac{\log T(\log p)^{1/2}}{k^{1/2}}\right)\nonumber\\
&\times&\left\{\npr(Z \in R_\text{left}| Z \in R)\hat{\sigma}^2(Y|Z \in R_\text{left}) + \npr(Z \in R_\text{right}| Z \in R)\hat{\sigma}^2(Y|Z \in R_\text{right})\right\}. \nonumber
\end{eqnarray*}
Since Assumption \ref{assum:rfrw_eps}-\ref{assum:rfrw_fbound} ensure that ${\sigma}^2(Y|Z \in R_\text{right})$ and ${\sigma}^2(Y|Z \in R_\text{right})$ are upper bounded, to prove (\ref{eqn:rfrw_high_emp_l_pop_l}), it suffices to show that
\begin{equation}
\label{eqn:rfrw_high_emp_l_pop_l_1}
\sup_{R\in \mathcal{R}_{\varsigma,w}:\mu(R)\geq \zeta w}\left| \hat{\sigma}^2(Y|Z \in R) -\sigma^2(Y|Z \in R)\right| = O_p\left( \frac{(\log T)^{3/2}(\log p)^{1/2}}{k^{1/2}}\right).
\end{equation}
By Lemma \ref{lemma:rfrw_gtl_minus} and the proof of Theorem \ref{theo:rfrw_rw_regression_tree}, we already showed that
\begin{equation*}
\label{eqn:rfrw_high_emp_l_pop_l_2}
\sup_{R\in \mathcal{R}_{\varsigma,w}:\mu(R)\geq \zeta w} \left| \frac{1}{\sum\limits_{t:Z_t\in R}\omega_t}\sum\limits_{t:Z_t\in R}\omega_tY_t-\mathbb{E}[Y|Z\in R]\right| = O_p\left( \frac{(\log T)^{3/2}(\log p)^{1/2}}{k^{1/2}}\right).
\end{equation*}
It remains to show that
\begin{equation}
\label{eqn:rfrw_high_emp_l_pop_l_3}
\sup_{R\in \mathcal{R}_{\varsigma,w}:\mu(R)\geq \zeta w} \left| \frac{1}{\sum\limits_{t:Z_t\in R}\omega_t}\sum\limits_{t:Z_t\in R}\omega_tY_t^2-\mathbb{E}[Y^2|Z\in R]\right| =O_p\left( \frac{(\log T)^{3/2}(\log p)^{1/2}}{k^{1/2}}\right).
\end{equation}
For any rectangle $R$, it holds that
\begin{eqnarray}
\label{eqn:rfrw_high_emp_l_pop_l_4}
&&\left| \frac{1}{\sum\limits_{t:Z_t\in R}\omega_t}\sum\limits_{t:Z_t\in R}\omega_tY_t^2-\mathbb{E}[Y^2|Z\in R]\right| \nonumber \\
&=& \left | \frac{\sum\limits_{t:Z_t\in R}\omega_tf(X_t)^2 + 2\sum\limits_{t:Z_t\in R}\omega_tf(X_t)\epsilon_t+\sum\limits_{t:Z_t\in R}\omega_t\epsilon_t^2}{\sum\limits_{t:Z_t\in R}\omega_t} -\mathbb{E}[Y^2|Z\in R]\right | \nonumber \\
&\leq& \left| \frac{\sum\limits_{t:Z_t\in R}\omega_tf(X_t)^2}{\sum\limits_{t:Z_t\in R}\omega_t} - \mathbb{E}[\omega f(X)^2|Z\in R]\right| + \left| \frac{2M\sum\limits_{t:Z_t\in R}\omega_t\epsilon_t}{\sum\limits_{t:Z_t\in R}\omega_t}\right| \nonumber \\
& +&\left| \frac{\sum\limits_{t:Z_t\in R}\omega_t\epsilon_t^2}{\sum\limits_{t:Z_t\in R}\omega_t}-\mathbb{E}[\omega\epsilon^2|Z\in R]\right|.
\end{eqnarray}
For the first term on the right-hand side of (\ref{eqn:rfrw_high_emp_l_pop_l_4}), similar to the proof of event $\mathcal{E}_6$, we can apply a Bernstein type inequality (Lemma \ref{lemma:rfrw_bernstein}) for the exponential $\alpha$-mixing process, such that
$$\sup_{R\in \mathcal{R}_{\varsigma,w}:\mu(R)\geq \zeta w}\left| \frac{\sum\limits_{t:Z_t\in R}\omega_tf(X_t)^2}{\sum\limits_{t:Z_t\in R}\omega_t} - \mathbb{E}[\omega f(X)^2|Z\in R]\right| =O_p\left( \frac{(\log T)^{3/2}(\log p)^{1/2}}{k^{1/2}}\right).$$
For the second term, we have proved it in Lemma \ref{lemma:rfrw_gtl_minus}, such that
$$\sup_{R\in \mathcal{R}_{\varsigma,w}:\mu(R)\geq \zeta w} \left| \frac{2M\sum\limits_{t:Z_t\in R}\omega_t\epsilon_t}{\sum\limits_{t:Z_t\in R}\omega_t}\right| =O_p\left( \frac{(\log T\log p)^{1/2}}{k^{1/2}}\right).$$
For the third term, we can similarly construct a martingale difference and apply the Freedman type inequality (Lemma \ref{lemma:rfrw_martingale}) for unbounded summands, such that
$$\sup_{R\in \mathcal{R}_{\varsigma,w}:\mu(R)\geq \zeta w} \left| \frac{\sum\limits_{t:Z_t\in R}\omega_t\epsilon_t^2}{\sum\limits_{t:Z_t\in R}\omega_t}-\mathbb{E}[\omega\epsilon^2|Z\in R]\right| =O_p\left( \frac{(\log T\log p)^{1/2}}{k^{1/2}}\right).$$
Accordingly, (\ref{eqn:rfrw_high_emp_l_pop_l_3}) can be proved, and we derive (\ref{eqn:rfrw_high_emp_l_pop_l_1}), which finishes the proof.
\end{proof}

\subsection{Proof of Theorem \ref{theo:rfrw_high_consistent}}

To develop the consistency of high-dimensional RF-RW, as discussed in Section \ref{sec:theory_consistency_high}, we need to restrict the maximum leaf diameter of the informative features to 0. It requires that each informative feature be selected at the feature sampling process and split with growing number of times.

Suppose for some fixed integer $d>0$, $p\rightarrow\infty$ with $\lim \inf p/T>0$, and the number of candidate features sampled $m_{\text{try}}\asymp p$. Let $s$ be the number of splits conducted. Definition \ref{def:rfrw_xi_k_m_partition} and Lemma \ref{lemma:rfrw_high_xi_range} implies that $s=O(\log T)$, and every split position lies in $[\tilde{\xi},1-\tilde{\xi}]$ with $\tau\in[\tilde{\xi},1-\tilde{\xi}]$ for the splits placed along the node $R\subseteq[0,1]^p$. As $T\rightarrow\infty$, we have $s\rightarrow\infty$, it follows that the probability of a given feature being selected as a candidate with less than or equal to $d$ times goes to $0$. Moreover, by Lemma \ref{lemma:rfrw_high_emp_l_pop_l}, as $T\rightarrow\infty$, we have $\hat{l}(R,j,\tau) \overset{P} \rightarrow l(R,j,\tau)$, which implies that $\hat{l}_\Lambda(A,j,\tau) \overset{P} \rightarrow l_\Lambda(A,j,\tau)$ and $\hat{\mathcal{I}}(R,j,\tau) \overset{P} \rightarrow \mathcal{I}(R,j,\tau)$. We can impose the following assumptions on $\mathcal{I}(R,j,\tau)$ to ensure that the informative features are prioritized to be selected as candidates and to be split, which is similar to \cite{Chen2024-sup}. 

\begin{assumption}
\label{assump:rfrw_high_signal}
\edef\assprefix{\theassumption}
\begin{enumerate}[label=(\roman*), ref=\assprefix(\roman*)]
  \item[] 
  \item\label{assum:rfrw_high_Q_delta} (Splitting of informative features). For an index set $Q$ with $|Q|=o(\log T)$, define $\Delta(\upsilon)$ as follows,
  $$\Delta(\upsilon)=\inf_{j\in Q}\inf_{\{R=\bigtimes_{i=1}^p[r_j^-,r_j^+]\subseteq [0,1]^p:r_j^+-r_j^-=\upsilon\}}\mathcal{I}(R,j,\frac{1}{2}).$$
  It holds that $\Delta(\upsilon)>0$ for any $\upsilon>0$.

  \item\label{assum:rfrw_high_q1_q} (Uninformative feature controlling). There exists a index set $Q_1 \supseteq Q$ with $|Q_1|=o(\log T)$ and $p-m_{\text{try}}=O(p/|Q_1|)$, such that for any rectangle $R$ and $i\notin Q_1$, it holds that 
  $$\sup_{\tilde{\xi}\leq\tau\leq 1-\tilde{\xi}}\mathcal{I}(R,i,\tau)<\psi\max_{j\in Q_1}\sup_{\tilde{\xi}\leq\tau\leq 1-\tilde{\xi}}\mathcal{I}(R,j,\tau),$$
  where constant $\psi\in (0,1)$.
\end{enumerate}

\end{assumption}

Assumption \ref{assum:rfrw_high_Q_delta} imposes the assumption on the lower bound of the variance reduction of splitting along the informative features, which is similar to Assumption 6 in \cite{Chen2024-sup}. Assumption \ref{assum:rfrw_high_q1_q} guarantees that the variance reduction of splitting along the uninformative features is uniformly less than the maximum variance reduction of splitting the features in $|Q_1|$ by a multiplicative constant $\psi$, which is similar to Assumption 7 in \cite{Chen2024-sup}. It implies that as $T\rightarrow\infty$, when all the features in $Q_1$ are simultaneously chosen to be the candidates, the split will be placed along a feature in $Q_1$. 

Compared to \cite{Wager2015-sup}, the conditions we impose do not require the informative features to be independent of the response. In contrast to \cite{Wager2015-sup,Chen2024-sup} that assume the informative feature set $Q$ to be fixed, we allow for a growing $|Q|$ as sample sizes $T$ increase.

In the following, we introduce Lemma \ref{lemma:rfrw_high_delta_lower_bound}-\ref{lemma:rfrw_high_I_upper_bound}, where Lemma \ref{lemma:rfrw_high_delta_lower_bound} provides the lower bound of $\Delta(\upsilon)$ and Lemma \ref{lemma:rfrw_high_I_upper_bound} obtains the upper bound of $\mathcal{I}(R,j,\tau)$ for splitting feature $j\in Q_1$. These Lemmas are used in Lemma \ref{lemma:rfrw_high_split_q1}-\ref{lemma:rfrw_high_split_q_infty}, which ensures that all informative features are selected as candidates and split with increasing number of times in the splitting process.

\begin{lemma}
\label{lemma:rfrw_high_delta_lower_bound}
Suppose that Assumption \ref{assump:rfrw_all}, \ref{assump:rfrw_high_lip}, and \ref{assum:rfrw_high_Q_delta} are satisfied, it holds that 
$$\inf_{c\leq \upsilon\leq 1}\Delta(\upsilon)>0,$$
where $0<c<1$.
\end{lemma}

\begin{proof}
The proof of Lemma \ref{lemma:rfrw_high_delta_lower_bound} is similar to Lemma B.5 in \cite{Chen2024-sup}, where they focus on the situation of $|Q|=1$. The same arguments can be applied to the case of $|Q|\rightarrow\infty$. Suppose that $\inf_{c\leq \upsilon\leq 1}\Delta(\upsilon)>0$ does not hold for any $c>0$, then $\inf_{c\leq \upsilon\leq 1}\Delta(\upsilon)=0$ for some $c$, and there exists $\{\upsilon_n\}_{n=1}^\infty\subseteq[c,1]$ such that $\Delta(\upsilon_n)\rightarrow0$. Assume that $\upsilon_n\downarrow\upsilon^*$ for some constant $\upsilon^*>0$, there exists a sequence of rectangles $\{R_n\}_{n=1}^\infty$ with $R_n=\bigtimes_{i=1}^p[r_{i,n}^-,r_{i,n}^+]$ and $r_{i,n}^+-r_{i,n}^-=\upsilon_n$, such that for splitting feature $j$, we have
$$\mathcal{I}(R_n,j,1/2)<(1+2^{-n})\Delta(\upsilon_n).$$
Define $\{\bar{R}_n\}_{n=1}^\infty$, where $\bar{R}_n=\bigtimes_{i=1}^p[\bar{r}_{i,n}^-,\bar{r}_{i,n}^+]$ and $\bar{R}_n=R_n\cap\{Z^{(i)}\leq \upsilon^*+r_{i,n}^-\}$ such that $\bar{r}_{i,n}^+-\bar{r}_{i,n}^-=\upsilon^*$. For splitting feature $j$, we have $\mathcal{I}(\bar{R}_n,j,1/2)\geq\Delta(\upsilon^*)$, and we would like to show that
\begin{equation}
\label{eqn:rfrw_high_delta_lower_bound_1}
|\mathcal{I}(\bar{R}_n,j,1/2)-\mathcal{I}(R_n,j,1/2)|\leq c^{-1}C(\upsilon_n-\upsilon^*),
\end{equation}
where $C>0$ is a constant independent of $c$, the rectangles $R_n$, $\bar{R}_n$, and splitting feature $j$. Denote $R_n^{(-j)}:=\bigtimes_{i\neq j}[r_{i,n}^-,r_{i,n}^+]$, consider that
\begin{eqnarray*}
&&|\mathbb{E}[Y|Z\in R_n] - \mathbb{E}[Y|Z\in\bar{R}_n]| \nonumber \\
&=& \left|\frac{\int\limits_{z^{(-j)}\in R_n^{(-j)}}\int\limits_{z^{(j)}=r_{j,n}^-}^{r_{j,n}^-+\upsilon_n}f(\iota_h^{-1}(z))h_Z(z)dz}{\int\limits_{z^{(-j)}\in R_n^{(-j)}}\int\limits_{z^{(j)}=r_{j,n}^-}^{r_{j,n}^-+\upsilon_n}h_Z(z)dz}-\frac{\int\limits_{z^{(-j)}\in R_n^{(-j)}}\int\limits_{z^{(j)}=r_{j,n}^-}^{r_{j,n}^-+\upsilon^*}f(\iota_h^{-1}(z))h_Z(z)dz}{\int\limits_{z^{(-j)}\in R_n^{(-j)}}\int\limits_{z^{(j)}=r_{j,n}^-}^{r_{j,n}^-+\upsilon^*}h_Z(z)dz}\right| \nonumber \\
&\leq& \left|\frac{\int\limits_{z^{(-j)}\in R_n^{(-j)}}\int\limits_{z^{(j)}=r_{j,n}^-}^{r_{j,n}^-+\upsilon_n}f(\iota_h^{-1}(z))h_Z(z)dz}{\int\limits_{z^{(-j)}\in R_n^{(-j)}}\int\limits_{z^{(j)}=r_{j,n}^-}^{r_{j,n}^-+\upsilon_n}h_Z(z)dz}-\frac{\int\limits_{z^{(-j)}\in R_n^{(-j)}}\int\limits_{z^{(j)}=r_{j,n}^-}^{r_{j,n}^-+\upsilon^*}f(\iota_h^{-1}(z))h_Z(z)dz}{\int\limits_{z^{(-j)}\in R_n^{(-j)}}\int\limits_{z^{(j)}=r_{j,n}^-}^{r_{j,n}^-+\upsilon_n}h_Z(z)dz}\right| \nonumber \\
&+&  \left|\frac{\int\limits_{z^{(-j)}\in R_n^{(-j)}}\int\limits_{z^{(j)}=r_{j,n}^-}^{r_{j,n}^-+\upsilon^*}f(\iota_h^{-1}(z))h_Z(z)dz}{\int\limits_{z^{(-j)}\in R_n^{(-j)}}\int\limits_{z^{(j)}=r_{j,n}^-}^{r_{j,n}^-+\upsilon_n}h_Z(z)dz}-\frac{\int\limits_{z^{(-j)}\in R_n^{(-j)}}\int\limits_{z^{(j)}=r_{j,n}^-}^{r_{j,n}^-+\upsilon^*}f(\iota_h^{-1}(z))h_Z(z)dz}{\int\limits_{z^{(-j)}\in R_n^{(-j)}}\int\limits_{z^{(j)}=r_{j,n}^-}^{r_{j,n}^-+\upsilon^*}h_Z(z)dz}\right| \nonumber \\
&\leq& \left|\frac{\int\limits_{z^{(-j)}\in R_n^{(-j)}}\int\limits_{z^{(j)}=r_{j,n}^-+\upsilon^*}^{r_{j,n}^-+\upsilon_n}f(\iota_h^{-1}(z))h_Z(z)dz}{\int\limits_{z^{(-j)}\in R_n^{(-j)}}\int\limits_{z^{(j)}=r_{j,n}^-}^{r_{j,n}^-+\upsilon_n}h_Z(z)dz}\right| +\left|\int\limits_{z^{(-j)}\in R_n^{(-j)}}\int\limits_{z^{(j)}=r_{j,n}^-}^{r_{j,n}^-+\upsilon^*}f(\iota_h^{-1}(z))h_Z(z)dz\right| \nonumber \\
&\times& \left| \frac{\int\limits_{z^{(-j)}\in R_n^{(-j)}}\int\limits_{z^{(j)}=r_{j,n}^-+\upsilon^*}^{r_{j,n}^-+\upsilon_n}h_Z(z)dz}{\left(\int\limits_{z^{(-j)}\in R_n^{(-j)}}\int\limits_{z^{(j)}=r_{j,n}^-}^{r_{j,n}^-+\upsilon_n}h_Z(z)dz\right)\left(\int\limits_{z^{(-j)}\in R_n^{(-j)}}\int\limits_{z^{(j)}=r_{j,n}^-}^{r_{j,n}^-+\upsilon^*}h_Z(z)dz\right)}\right| \nonumber \\
&\leq& M\frac{\zeta \text{Leb}\left(R_n^{(-j)}\right)(\upsilon_n-\upsilon^*)}{\zeta^{-1} \text{Leb}\left(R_n^{(-j)}\right)\upsilon_n} + M\frac{\zeta^2 \left(\text{Leb}\left(R_n^{(-j)}\right)\right)^2(\upsilon_n-\upsilon^*)\upsilon^*}{\zeta^{-2} \left(\text{Leb}\left(R_n^{(-j)}\right)\right)^2\upsilon_n\upsilon^*} \nonumber \\
&\leq& M \zeta^2\frac{\upsilon_n-\upsilon^*}{\upsilon_n} + M \zeta^4\frac{\upsilon_n-\upsilon^*}{\upsilon_n}.
\end{eqnarray*}
This shows that $|\mathbb{E}[Y|Z\in R_n] - \mathbb{E}[Y|Z\in\bar{R}_n]|\leq c^{-1}C(\upsilon_n-\upsilon^*)$. We can similarly show that $|\mathbb{E}[Y^2|Z\in R_n] - \mathbb{E}[Y^2|Z\in\bar{R}_n]|\leq c^{-1}C(\upsilon_n-\upsilon^*)$, which brings about $|\sigma^2(Y|Z\in R_n) - \sigma^2(Y|Z\in\bar{R}_n)|\leq c^{-1}C(\upsilon_n-\upsilon^*).$

We consider the split positions $r_{j,n}^-+{\upsilon_n/2}$ and $r_{j,n}^-+{\upsilon^*/2}$ for rectangles $R_n$ and $\bar{R}_n$, respectively. By applying similar arguments, we have 
$$\left|\frac{\npr\left(Z\in R_n\cap\{Z^{(j)}\leq r_{j,n}^-+{\upsilon_n/2}\}\right)}{\npr(Z\in R_n)}-\frac{\npr\left(Z\in \bar{R}_n\cap\{Z^{(j)}\leq r_{j,n}^-+{\upsilon^*/2}\}\right)}{\npr(Z\in \bar{R}_n)}\right|\leq c^{-1}C(\upsilon_n-\upsilon^*).$$

$$\left|\frac{\npr\left(Z\in R_n\cap\{Z^{(j)}> r_{j,n}^-+{\upsilon_n/2}\}\right)}{\npr(Z\in R_n)}-\frac{\npr\left(Z\in \bar{R}_n\cap\{Z^{(j)}> r_{j,n}^-+\upsilon^*/2\}\right)}{\npr(Z\in \bar{R}_n)}\right|\leq c^{-1}C(\upsilon_n-\upsilon^*).$$

$$|\sigma^2\left(Y|Z\in R_n\cap\{Z^{(j)}\leq r_{j,n}^-+\upsilon_n/2\}\right) - \sigma^2\left(Y|Z\in\bar{R}_n\cap\{Z^{(j)}\leq r_{j,n}^-+{\upsilon^*/2}\}\right)|\leq c^{-1}C(\upsilon_n-\upsilon^*).$$

$$|\sigma^2\left(Y|Z\in R_n\cap\{Z^{(j)}> r_{j,n}^-+{\upsilon_n/2}\}\right) - \sigma^2\left(Y|Z\in\bar{R}_n\cap\{Z^{(j)}> r_{j,n}^-+{\upsilon^*/2}\})|\leq c^{-1}C(\upsilon_n-\upsilon^*\right).$$
Combining the results obtained above, we prove (\ref{eqn:rfrw_high_delta_lower_bound_1}), which further shows that
$$\Delta(\upsilon^*)\leq (1+2^{-n})\Delta(\upsilon_n) + c^{-1}C(\upsilon_n-\upsilon^*).$$
When $n\rightarrow\infty$, $\Delta(\upsilon^*)=0$, which contradicts to Assumption \ref{assum:rfrw_high_Q_delta}. Therefore, $\inf_{c\leq \upsilon\leq 1}\Delta(\upsilon)>0$ holds for any $c>0$.
\end{proof}

\begin{lemma}
\label{lemma:rfrw_high_I_upper_bound}
Suppose that Assumption \ref{assump:rfrw_all}, and \ref{assump:rfrw_high_lip} are satisfied. For any rectangle $R$ and splitting feature $j\in Q_1$, it holds that
$$\mathcal{I}(R,j,\tau)\leq C(r_j^+-r_j^-),$$
where $C>0$ is a constant independent of rectangle $R$, splitting feature $j$, and relative split position $\tau$.
\end{lemma}

\begin{proof}
The proof of Lemma \ref{lemma:rfrw_high_I_upper_bound} is similar to Lemma B.6 in \cite{Chen2024-sup}. For any splitting feature $j\in Q_1$, rectangle $R$ with $R^{(-j)}:=\bigtimes_{i\neq j}[r_{i}^-,r_{i}^+]$, and the split position $r_{j,\tau}=\tau r_j^-+(1-\tau)r_j^+$ with $\tau \in[\xi,1-\xi]$, it holds that 
\begin{eqnarray*}
&&\left|\mathbb{E}[Y|Z\in R] - \frac{\int\limits_{z^{(-j)}\in R^{(-j)}}\int\limits_{z^{(j)}=r_{j}^-}^{r_{j}^+}f(\iota_h^{-1}(z))h_Z(z^{(-j)},r_{j,\tau})dz}{\int\limits_{z^{(-j)}\in R^{(-j)}}\int\limits_{z^{(j)}=r_{j}^-}^{r_{j}^+}h_Z(z^{(-j)},r_{j,\tau})dz}\right| \nonumber \\
&=& \left|\frac{\int\limits_{z^{(-j)}\in R^{(-j)}}\int\limits_{z^{(j)}=r_{j}^-}^{r_{j}^+}f(\iota_h^{-1}(z))h_Z(z)dz}{\int\limits_{z^{(-j)}\in R^{(-j)}}\int\limits_{z^{(j)}=r_{j}^-}^{r_{j}^+}h_Z(z)dz}-\frac{\int\limits_{z^{(-j)}\in R^{(-j)}}\int\limits_{z^{(j)}=r_{j}^-}^{r_{j}^+}f(\iota_h^{-1}(z))h_Z(z^{(-j)},r_{j,\tau})dz}{\int\limits_{z^{(-j)}\in R^{(-j)}}\int\limits_{z^{(j)}=r_{j}^-}^{r_{j}^+}h_Z(z^{(-j)},r_{j,\tau})dz}\right| \nonumber \\
&\leq& \left|\frac{\int\limits_{z^{(-j)}\in R^{(-j)}}\int\limits_{z^{(j)}=r_{j}^-}^{r_{j}^+}f(\iota_h^{-1}(z))h_Z(z)dz - \int\limits_{z^{(-j)}\in R^{(-j)}}\int\limits_{z^{(j)}=r_{j}^-}^{r_{j}^+}f(\iota_h^{-1}(z))h_Z(z^{(-j)},r_{j,\tau})dz}{\int\limits_{z^{(-j)}\in R^{(-j)}}\int\limits_{z^{(j)}=r_{j}^-}^{r_{j}^+}h_Z(z)dz}\right| \nonumber \\
&+&  \left|\int\limits_{z^{(-j)}\in R^{(-j)}}\int\limits_{z^{(j)}=r_{j}^-}^{r_{j}^+}f(\iota_h^{-1}(z))h_Z(z^{(-j)},r_{j,\tau})dz\right|  \nonumber \\
&\times& \left|\frac{\int\limits_{z^{(-j)}\in R^{(-j)}}\int\limits_{z^{(j)}=r_{j}^-}^{r_{j}^+}h_Z(z^{(-j)},r_{j,\tau})dz-\int\limits_{z^{(-j)}\in R^{(-j)}}\int\limits_{z^{(j)}=r_{j}^-}^{r_{j}^+}h_Z(z)dz}{\left(\int\limits_{z^{(-j)}\in R^{(-j)}}\int\limits_{z^{(j)}=r_{j}^-}^{r_{j}^+}h_Z(z)dz\right)\left(\int\limits_{z^{(-j)}\in R^{(-j)}}\int\limits_{z^{(j)}=r_{j}^-}^{r_{j}^+}h_Z(z^{(-j)},r_{j,\tau})dz\right)}\right| \nonumber \\
&\leq& ML_1\zeta (r_j^+-r_j^-)+ ML_1\zeta^3 (r_j^+-r_j^-).\nonumber
\end{eqnarray*}
Hence, it shows that 
$$\left|\mathbb{E}[Y|Z\in R] - \frac{\int\limits_{z^{(-j)}\in R^{(-j)}}\int\limits_{z^{(j)}=r_{j}^-}^{r_{j}^+}f(\iota_h^{-1}(z))h_Z(z^{(-j)},r_{j,\tau})dz}{\int\limits_{z^{(-j)}\in R^{(-j)}}\int\limits_{z^{(j)}=r_{j}^-}^{r_{j}^+}h_Z(z^{(-j)},r_{j,\tau})dz}\right|\leq C(r_j^+-r_j^-), $$
for some suitable constant $C>0$. Then, applying similar arguments, we can show that
$$\left|\mathbb{E}[Y|Z\in R\cap\{Z^{(j)}\leq r_{j,\tau}\}] - \frac{\int\limits_{z^{(-j)}\in R^{(-j)}}\int\limits_{z^{(j)}=r_{j}^-}^{r_{j}^+}f(\iota_h^{-1}(z))h_Z(z^{(-j)},r_{j,\tau})dz}{\int\limits_{z^{(-j)}\in R^{(-j)}}\int\limits_{z^{(j)}=r_{j}^-}^{r_{j}^+}h_Z(z^{(-j)},r_{j,\tau})dz}\right|\leq C(r_j^+-r_j^-).$$

$$\left|\mathbb{E}[Y|Z\in R\cap\{Z^{(j)}> r_{j,\tau}\}] - \frac{\int\limits_{z^{(-j)}\in R^{(-j)}}\int\limits_{z^{(j)}=r_{j}^-}^{r_{j}^+}f(\iota_h^{-1}(z))h_Z(z^{(-j)},r_{j,\tau})dz}{\int\limits_{z^{(-j)}\in R^{(-j)}}\int\limits_{z^{(j)}=r_{j}^-}^{r_{j}^+}h_Z(z^{(-j)},r_{j,\tau})dz}\right|\leq C(r_j^+-r_j^-).$$

\begin{eqnarray*}
   &&\left|\mathbb{E}[Y^2|Z\in R] - \frac{\int\limits_{z^{(-j)}\in R^{(-j)}}\int\limits_{z^{(j)}=r_{j}^-}^{r_{j}^+}\big(f(\iota_h^{-1}(z))\big)^2h_Z(z^{(-j)},r_{j,\tau})dz}{\int\limits_{z^{(-j)}\in R^{(-j)}}\int\limits_{z^{(j)}=r_{j}^-}^{r_{j}^+}h_Z(z^{(-j)},r_{j,\tau})dz}-\mathbb{E}[\epsilon^2]\right| 
    \leq C(r_j^+-r_j^-). \nonumber
\end{eqnarray*}

\begin{eqnarray*}
   &&\left|\mathbb{E}[Y^2|Z\in R\cap\{Z^{(j)}\leq r_{j,\tau}\}] - \frac{\int\limits_{z^{(-j)}\in R^{(-j)}}\int\limits_{z^{(j)}=r_{j}^-}^{r_{j}^+}\big(f(\iota_h^{-1}(z))\big)^2h_Z(z^{(-j)},r_{j,\tau})dz}{\int\limits_{z^{(-j)}\in R^{(-j)}}\int\limits_{z^{(j)}=r_{j}^-}^{r_{j}^+}h_Z(z^{(-j)},r_{j,\tau})dz}-\mathbb{E}[\epsilon^2]\right| \nonumber \\
    &\leq& C(r_j^+-r_j^-). \nonumber
\end{eqnarray*}

\begin{eqnarray*}
   &&\left|\mathbb{E}[Y^2|Z\in R\cap\{Z^{(j)}>r_{j,\tau}\}] - \frac{\int\limits_{z^{(-j)}\in R^{(-j)}}\int\limits_{z^{(j)}=r_{j}^-}^{r_{j}^+}\big(f(\iota_h^{-1}(z))\big)^2h_Z(z^{(-j)},r_{j,\tau})dz}{\int\limits_{z^{(-j)}\in R^{(-j)}}\int\limits_{z^{(j)}=r_{j}^-}^{r_{j}^+}h_Z(z^{(-j)},r_{j,\tau})dz}-\mathbb{E}[\epsilon^2]\right| \nonumber \\
    &\leq& C(r_j^+-r_j^-). \nonumber
\end{eqnarray*}
Combining the above results, we show that $\mathcal{I}(R,j,\tau)\leq C(r_j^+-r_j^-)$.
\end{proof}

Based on Lemma \ref{lemma:rfrw_high_delta_lower_bound}-\ref{lemma:rfrw_high_I_upper_bound}, we introduce the intuition and details of obtaining Lemma \ref{lemma:rfrw_high_split_q1}-\ref{lemma:rfrw_high_split_q_infty} below.

For the nodes in the tree that contain $z=(z^{(1)},\ldots,z^{(p)})^\top\in[0,1]^p$, let $\theta$ be the realization of the $\Theta$. Define $\Phi(z,s,d)$ as the collection of $\theta$, such that all the features in $Q_1$ are simultaneously chosen as candidates for at least $d$ splits in the first $s$ splits. Assume that $|Q_1|=o(s)$, i.e., $|Q|,|Q_1|=o(\log T)$, $m_{\text{try}} \asymp p$ and $p\rightarrow\infty$ with $\lim \inf p/T>0$. As $T\rightarrow\infty$, it is straightforward that $m_{\text{try}} \gg |Q_1|$ and hence $\lim_{s\rightarrow\infty}\npr_{\Theta}\left(\Theta\in\Phi(z,s,d)\right)=1$ for any fixed $d$. It implies that as the splitting process continues, the number of splits that a specific informative feature in $Q$ is chosen as the candidate will increase to infinity with probability going to one with respect to $\Theta$. 

Denote $\mathcal{S}(z,\Theta,s,j)$ as the number of splits along the feature $j$, given $\Theta$ and the number of splits $s$. We have the following results.

\begin{lemma}
\label{lemma:rfrw_high_split_q1}
Under Assumption \ref{assump:rfrw_all}, \ref{assump:rfrw_high_lip}, and \ref{assump:rfrw_high_signal}, suppose that $T\rightarrow\infty$, $p\rightarrow \infty$ with $\lim \inf p/T>0$, $m_{\text{try}} \asymp p$, and $d$ is fixed. It holds that
$$\npr\left(\inf_{z\in[0,1]^p:\Phi(z,s,d)\neq\emptyset} \inf_{\Theta\in\Phi(z,s,d)}\sum_{j\in Q_1}\mathcal{S}(z,\Theta,s,j)\geq d\right)=1.$$
\end{lemma}

\begin{proof}
    It is equivalent to show that as $T\rightarrow\infty$
    $$\npr\left(\inf_{z\in[0,1]^p:\Phi(z,s,d)\neq\emptyset} \inf_{\Theta\in\Phi(z,s,d)}\sum_{j\in Q_1}\mathcal{S}(z,\Theta,s,j)\leq d-1\right)=0.$$
    Note that
    \begin{eqnarray*}
    &&\left\{\inf_{z\in[0,1]^p:\Phi(z,s,d)\neq\emptyset} \inf_{\Theta\in\Phi(z,s,d)}\sum_{j\in Q_1}\mathcal{S}(z,\Theta,s,j)\leq d-1 \right\} \nonumber \\
    &\subseteq& \left\{ \sup_{R\in[0,1]^p:\text{Leb}(R)\geq \tilde{\xi}^{s-1}}\left(\sup_{i\notin Q_1}\sup_{\tilde{\xi}\leq\tau\leq1-\tilde{\xi}}\hat{\mathcal{I}}(R,i,\tau)-\sup_{j\in Q_1}\sup_{\tilde{\xi}\leq\tau\leq1-\tilde{\xi}}\hat{\mathcal{I}}(R,j,\tau)\right) \geq 0\right\}. \nonumber
    \end{eqnarray*}
    As shown in Lemma \ref{lemma:rfrw_high_emp_l_pop_l} that $\hat{\mathcal{I}}(R,j,\tau) \overset{P} \rightarrow \mathcal{I}(R,j,\tau)$, it suffices to show that
    $$\sup_{R\in[0,1]^p:\text{Leb}(R)\geq \tilde{\xi}^{s-1}}\left(\sup_{i\notin Q_1}\sup_{\tilde{\xi}\leq\tau\leq1-\tilde{\xi}}\mathcal{I}(R,i,\tau)-\sup_{j\in Q_1}\sup_{\tilde{\xi}\leq\tau\leq1-\tilde{\xi}}\mathcal{I}(R,j,\tau)\right) \geq 0.$$
    By Lemma \ref{lemma:rfrw_high_delta_lower_bound} and Assumption \ref{assum:rfrw_high_q1_q}, it holds that
    \begin{eqnarray*}
        &&\sup_{R\in[0,1]^p:\text{Leb}(R)\geq \tilde{\xi}^{s-1}}\left(\sup_{i\notin Q_1}\sup_{\tilde{\xi}\leq\tau\leq1-\tilde{\xi}}\mathcal{I}(R,i,\tau)-\sup_{j\in Q_1}\sup_{\tilde{\xi}\leq\tau\leq1-\tilde{\xi}}\mathcal{I}(R,j,\tau)\right) \nonumber \\
        &\leq& \sup_{R\in[0,1]^p:\text{Leb}(R)\geq \tilde{\xi}^{s-1}}\left(-(1-\psi)\sup_{j\in Q_1}\sup_{\tilde{\xi}\leq\tau\leq 1-\tilde{\xi}}\mathcal{I}(R,j,\tau) \right)\nonumber \\
        &\leq& -(1-\psi) \inf_{\upsilon\geq \tilde{\xi}^{s-1}}\Delta(\upsilon)\nonumber \\
        &<& 0. \nonumber
    \end{eqnarray*}
    Accordingly, as $T\rightarrow\infty$
    $$\npr\left( \sup_{R\in[0,1]^p:\text{Leb}(R)\geq \tilde{\xi}^{s-1}}\left(\sup_{i\notin Q_1}\sup_{\tilde{\xi}\leq\tau\leq1-\tilde{\xi}}\mathcal{I}(R,i,\tau)-\sup_{j\in Q_1}\sup_{\tilde{\xi}\leq\tau\leq1-\tilde{\xi}}\mathcal{I}(R,j,\tau)\right) \geq 0\right)=0,$$
    which finishes the proof.
\end{proof}

Lemma \ref{lemma:rfrw_high_split_q1} implies that when $d$ is sufficiently large, each informative feature in $Q$ will be split more than a predetermined number of times, which brings about Lemma \ref{lemma:rfrw_high_split_q}.

\begin{lemma}
\label{lemma:rfrw_high_split_q}
Under Assumption \ref{assump:rfrw_all}, \ref{assump:rfrw_high_lip}, and \ref{assump:rfrw_high_signal}, suppose that $T\rightarrow\infty$, $p\rightarrow \infty$ with $\lim \inf p/T>0$, and $m_{\text{try}} \asymp p$. For any fixed $d\in \mathbb{N}$, let $d^*$ satisfy $2C(1-\tilde{\xi})^{d^*-1}<\inf_{\upsilon\geq \tilde{\xi}^{d-1}}\Delta(\upsilon)$, where $C=C(d^*)$ is a constant depending on $d^*$. It holds that
$$\npr\left(\inf_{z\in[0,1]^p:\Phi(z,s,|Q_1|\max\{d,d^*\})\neq\emptyset}\inf_{\Theta\in\Phi(z,s,|Q_1|\max\{d,d^*\})}\min_{j\in Q}\mathcal{S}(\Theta,s,j)\geq d\right)=1.$$
\end{lemma}

\begin{proof}

    For each split, a total number of $m_{\text{try}}$ features are selected as candidates, and the probability of all the features in $Q_1$ simultaneously selected as candidates is
    \begin{eqnarray*}
        \frac{\binom{p-|Q_1|}{m_{\text{try}}-|Q_1|}}{\binom{p}{m_{\text{try}}}}&=&\left(1-\frac{|Q_1|}{p}\right)\left(1-\frac{|Q_1|}{p-1}\right)\ldots\left(1-\frac{|Q_1|}{m_{\text{try}}+1}\right) \nonumber \\
        &\geq& \left(1-\frac{|Q_1|}{m_{\text{try}}}\right)^{p-m_{\text{try}}}\geq \exp\left\{\frac{|Q_1|\lim\inf_{p\rightarrow\infty}\left(1-{p}/{m_{\text{try}}}\right)}{2}\right\} \equiv \tilde{P}>0.\nonumber
    \end{eqnarray*}
    When $p\rightarrow\infty$, it holds that
    \begin{equation}
    \label{eqn:rfrw_high_split_q_infty_1}
    \npr_\Theta\big(\Theta\in\Phi(z,s,|Q_1|\max\{d,d^*\})\big)\geq1-\sum_{j=0}^{|Q_1|\max\{d,d^*\}-1}\binom{s}{j}(1-\tilde{P})^{s-j}\tilde{P}^j.
    \end{equation}
    Since the right-hand side of (\ref{eqn:rfrw_high_split_q_infty_1}) does not depend on a specific $z$, and $s=O(\log T)$, as $T\rightarrow\infty$, it implies that
    \begin{equation}
    \label{eqn:rfrw_high_split_q_1}
    \inf_{z\in[0,1]^p}\npr_\Theta\big(\Theta\in\Phi(z,s,|Q_1|\max\{d,d^*\})\big)=1.
    \end{equation}
    Suppose we have $|Q_1|\max\{d,d^*\}$ splits that all the features in $Q_1$ are simultaneously selected as candidates in the first $s$ splits, there exists at least one feature $j_1\in Q_1$ that is split for more than $d^*$ times, and another feature $j_2\in Q$ that is split for less than $d$ times. 

    By Lemma \ref{lemma:rfrw_high_delta_lower_bound}, it holds that $\sup_{\tilde{\xi}\leq\tau\leq1-\tilde{\xi}}\mathcal{I}(R,j_2,\tau)\geq \inf_{\upsilon\geq\tilde{\xi}^{d-1}}\Delta(\upsilon)$. Consider the situation that all the features in $Q_1$ are selected as candidates, and the last split is placed along $j_1$. For this split along the feature $j_1$, it holds that $\sup_{\tilde{\xi}\leq\tau\leq1-\tilde{\xi}}\mathcal{I}(R,j_1,\tau)\leq C(1-\tilde{\xi})^{d^*-1}<\frac{1}{2}\inf_{\upsilon\geq\tilde{\xi}^{d-1}}\Delta(\upsilon)$ by Lemma \ref{lemma:rfrw_high_I_upper_bound}. Accordingly, the split should be placed along $j_2$ instead of $j_1$ with probability going to $1$, which brings about the contradiction. Since the above arguments do not depend on the specific observed data $z$ and $\theta$, Lemma \ref{lemma:rfrw_high_split_q} is proved.
\end{proof}

Based on Lemma \ref{lemma:rfrw_high_split_q1}-\ref{lemma:rfrw_high_split_q}, we can obtain Lemma \ref{lemma:rfrw_high_split_q_infty}, which ensures that the number of splits along each informative feature in $Q$ will increase to exceed any fixed integer with probability going to $1$ as $s\rightarrow\infty$.

\begin{lemma}
\label{lemma:rfrw_high_split_q_infty}
Under Assumption \ref{assump:rfrw_all}, \ref{assump:rfrw_high_lip}, and \ref{assump:rfrw_high_signal}, suppose that $T\rightarrow\infty$, $p\rightarrow \infty$ with $\lim \inf p/T>0$, and $m_{\text{try}} \asymp p$. For any positive integer $d$, it holds that
$$\plim_{T\rightarrow\infty} \inf_{z\in[0,1]^p}\npr_{\Theta}\left(\min_{j\in Q}\mathcal{S}(z,\Theta,s,j) \geq d\right)=1.$$
\end{lemma}

\begin{proof}
    By Lemma \ref{lemma:rfrw_high_split_q}, as $T\rightarrow\infty$ and $p\rightarrow\infty$, for all $z$, we have
    $$\npr\left(\Phi\left(z,s,|Q_1|\max\{d,d^*\}\right)\subseteq\left\{\theta:\min_{j\in Q}\mathcal{S}(z,\theta,s,j) \geq d\right\}\right)=1,$$
    which implies that    $$\npr\left(\inf_{z\in[0,1]^p}\npr_\Theta\big(\Theta\in\Phi\left(z,s,|Q_1|\max\{d,d^*\}\right)\big)\leq\inf_{z\in[0,1]^p}\npr_\Theta\left\{\min_{j\in Q}\mathcal{S}(z,\Theta,s,j) \geq d\right\}\right)=1.$$
    It follows from (\ref{eqn:rfrw_high_split_q_1}) that for any $c \in (0,1)$, as $T\rightarrow\infty$, we have
    $$\npr\left\{ \inf_{z\in[0,1]^p}\npr_{\Theta}\left(\min_{j\in Q}\mathcal{S}(z,\Theta,s,j) \geq d\right)>c\right\}=1,$$
    which finishes the proof.
\end{proof}

Since the results in Lemma \ref{lemma:rfrw_high_split_q_infty} hold for any $d$, it implies that the expected number of splits (with respect to $\Theta$) along the informative features in $Q$ will increase to infinity with probability going to 1, which ensures the upper bound on the right-hand side of (\ref{eqn:rfrw_high_h_optimal_fx}) goes to $0$ and obtains the consistency of high-dimensional RF-RW. 

Based on Lemma \ref{lemma:rfrw_high_split_q_infty}, we introduce the proof of Theorem \ref{theo:rfrw_high_consistent} below.

\begin{proof}[Proof of Theorem \ref{theo:rfrw_high_consistent}]
Based on Theorem \ref{theo:rfrw} and Theorem \ref{theo:rfrw_high_h_optimal_fx}, to obtain uniform consistency of high-dimensional RF-RW, it suffices to show that
$$\sup_{x\in\mathbb{R}^p}\mathbb{E}_\Theta\text{diam}\left([\iota_h\big(L_\Lambda(x)\big)]_Q\right)\overset{P}{\longrightarrow}0,$$
as $T\rightarrow\infty$. By Lemma \ref{lemma:rfrw_high_split_q_infty}, we can obtain that
\begin{eqnarray*}
    &&\sup_{z\in[0,1]^p} \mathbb{E}_\Theta\normalfont{\text{diam}}\big([L(z)]_Q\big) \nonumber \\
    &\leq& \sup_{z\in[0,1]^p}\mathbb{E}_\Theta \left [\normalfont{\text{diam}}\big([L(z)]_Q\big) \Big{|} \min_{j\in Q} \mathcal{S}(z,\Theta,s,j)
    \geq d\right]\npr_{\Theta}\left(\min_{j\in Q}\mathcal{S}(z,\Theta,s,j) \geq d\right)\nonumber \\
    &+&  \sup_{z\in[0,1]^p}\mathbb{E}_\Theta \left [\normalfont{\text{diam}}\big([L(z)]_Q\big) \Big{|} \min_{j\in Q} \mathcal{S}(z,\Theta,s,j)
    < d\right]\npr_{\Theta}\left(\min_{j\in Q}\mathcal{S}(z,\Theta,s,j) < d\right)\nonumber \\
    &\leq& (1-\tilde{\xi})^d\nonumber.
\end{eqnarray*}
Since $L(z)=\iota_h(L_{\bar\Lambda}(x))=\iota_h(L_{\Lambda}(x))$ for $x\in\mathbb{R}^p$, and $z=\iota_h(x)$, it follows that
\begin{equation}
\label{eqn:rfrw_high_consistent_proof}
\npr\left(\sup_{x\in\mathbb{R}^p} \mathbb{E}_\Theta\text{diam}\left([\iota_h\big(L_\Lambda(x)\big)]_Q\right)\leq (1-\tilde{\xi})^d\right)=1.
\end{equation}
Since (\ref{eqn:rfrw_high_consistent_proof}) holds for any $d$, it implies that $\sup_{x\in\mathbb{R}^p}\mathbb{E}_\Theta\text{diam}\left([\iota_h\big(L_\Lambda(x)\big)]_Q\right)\overset{P}{\longrightarrow}0$ as $T\rightarrow\infty$, which finishes the proof.
\end{proof}

\section{Additional simulations}\label{sec:appen_sim}

In this section, we present additional results of the simulations in Section \ref{sec:sim} and conduct further experiments of the overspecified scenario, where we consider more lag features than those in the data-generating processes.

Table \ref{table:rfrw_dgp1}-\ref{table:rfrw_dgp3} report the mean RMTSE values of the same simulation studies in Section \ref{sec:sim}, where we consider the scenarios of the methods include: fewer lag features than those in the true underlying process (underspecified); exact lag features (correctly specified); more lag features (overspecified).

\begin{table}[!ht]
\centering
\caption{Mean RMTSE values of M1, including 3 lag features (underspecified), 5 lag features (correctly specified), 15 lag features (overspecified)} 
\label{table:rfrw_dgp1}
\begin{tabular}{ c c c c c c c c c c} 
 \hline
 &  \multicolumn{3}{c}{3 lag features} & \multicolumn{3}{c}{5 lag features} & \multicolumn{3}{c}{15 lag features} \\
 \cmidrule(lr){2-4}\cmidrule(lr){5-7}\cmidrule(lr){8-10}
 Method & $500$ & $1000$ & $2000$ & $500$ & $1000$ & $2000$ & $500$ & $1000$ & $2000$\\
 \hline
  RF-RW-1& 2.12 & 2.09 & 2.08 & 1.56 & 1.51 & 1.47 & 1.44 & 1.39 & 1.34\\
  RF-RW-2 & \textbf{2.11} & \textbf{2.07} & \textbf{2.06} & \textbf{1.28} & \textbf{1.24} & \textbf{1.20} & \textbf{1.31} & \textbf{1.26} & \textbf{1.20}\\
  tsRF & 2.12 & 2.09 & 2.08 & 1.36 & 1.32 &  1.27 & 1.59 & 1.53 &  1.49 \\
  SVM & 2.22 & 2.20 & 2.18 & 1.73 & 1.63 &  1.54 & 2.01 & 1.91 &  1.84 \\
  ExtraTrees & 2.16 & 2.15 & 2.14 &  1.71 & 1.68 &  1.63 & 1.96 & 1.92 &  1.89 \\
  RF & 2.33 & 2.31 & 2.29 &  2.32 & 2.30 &  2.29 & 2.34 & 2.30 &  2.29 \\ 
  LSTM & 2.29 & 2.28 & 2.26 & 2.29 & 2.27 &  2.26 & 2.31 & 2.27 &  2.25 \\

 \hline
\end{tabular}
\end{table}

\begin{table}[!ht]
\centering
\caption{Mean RMSTE values of M2, including 2 lag features (underspecified), 3 lag features (correctly specified), 15 lag features (overspecified)} 
\label{table:rfrw_dgp2}
\begin{tabular}{ c c c c c c c c c c} 
 \hline
 &  \multicolumn{3}{c}{2 lag features} & \multicolumn{3}{c}{3 lag features} & \multicolumn{3}{c}{15 lag features} \\
 \cmidrule(lr){2-4}\cmidrule(lr){5-7}\cmidrule(lr){8-10}
 Method & $500$ & $1000$ & $2000$ & $500$ & $1000$ & $2000$ & $500$ & $1000$ & $2000$\\
 \hline
  RF-RW-1 & 1.77 & 1.74 & \textbf{1.73} & 1.49 & 1.42 & 1.38 &  1.55 & 1.46 & 1.38\\
  RF-RW-2 & 1.78 & 1.75 & 1.73 & \textbf{1.22} & \textbf{1.15} & \textbf{1.10} &  \textbf{1.24} & \textbf{1.17} & \textbf{1.10}\\
  tsRF & 1.77 & 1.74 & 1.73 & 1.49 & 1.43 & 1.38 &  1.65 & 1.58 & 1.53 \\
  SVM & 1.80 & 1.77 & 1.76 &  1.80 & 1.76 &  1.74 &  1.86 & 1.85 &  1.82 \\
  ExtraTrees & 1.76 & \textbf{1.73} & 1.73 & 1.65 & 1.62 & 1.60 &  1.73 & 1.71 & 1.69 \\
  RF & 1.79 & 1.76 & 1.75 &  1.79 & 1.75 & 1.75 &  1.79 & 1.76 & 1.74 \\ 
  LSTM & \textbf{1.75} & 1.73 & 1.73 & 1.75 & 1.73 & 1.74 &  1.76 & 1.74 & 1.74\\

 \hline
\end{tabular}
\end{table}

\begin{table}[!ht]
\centering
\caption{Mean RMSTE values of M3, including 6 lag features (underspecified), 8 lag features (correctly specified), 15 lag features (overspecified)} 
\label{table:rfrw_dgp3}
\begin{tabular}{ c c c c c c c c c c} 
 \hline
 &  \multicolumn{3}{c}{6 lag features} & \multicolumn{3}{c}{8 lag features} & \multicolumn{3}{c}{15 lag features} \\
 \cmidrule(lr){2-4}\cmidrule(lr){5-7}\cmidrule(lr){8-10}
 Method & $500$ & $1000$ & $2000$ & $500$ & $1000$ & $2000$ & $500$ & $1000$ & $2000$\\
 \hline
  RF-RW-1 & \textbf{1.63} & \textbf{1.57} & \textbf{1.50} & \textbf{1.61} & \textbf{1.54} & \textbf{1.44}  & \textbf{1.60}  & \textbf{1.54} & \textbf{1.43}\\
  RF-RW-2 & 1.69 & 1.64 & 1.56 & 1.66 & 1.60 & 1.50 & 1.65 & 1.59 & 1.48\\
  tsRF & 1.78 & 1.72 & 1.64 & 1.74 & 1.67 & 1.57 &  1.87 & 1.82 &  1.71\\
  SVM & 2.38 & 2.34 & 2.27 & 2.37  & 2.35 & 2.26  & 2.41  & 2.42 & 2.32  \\
  ExtraTrees & 2.19 & 2.17 & 2.10 & 2.17 & 2.17 & 2.09 & 2.24  & 2.25 & 2.17 \\
  RF & 2.47 & 2.48 & 2.41 & 2.48  & 2.49 & 2.41 & 2.47  & 2.49 &  2.41\\ 
  LSTM & 2.39 & 2.42 & 2.37 & 2.39 & 2.42 & 2.37 & 2.39  & 2.42 & 2.36\\

 \hline
\end{tabular}
\end{table}

\clearpage

\begin{figure}[!ht]
\centering
\includegraphics[width=1\textwidth]{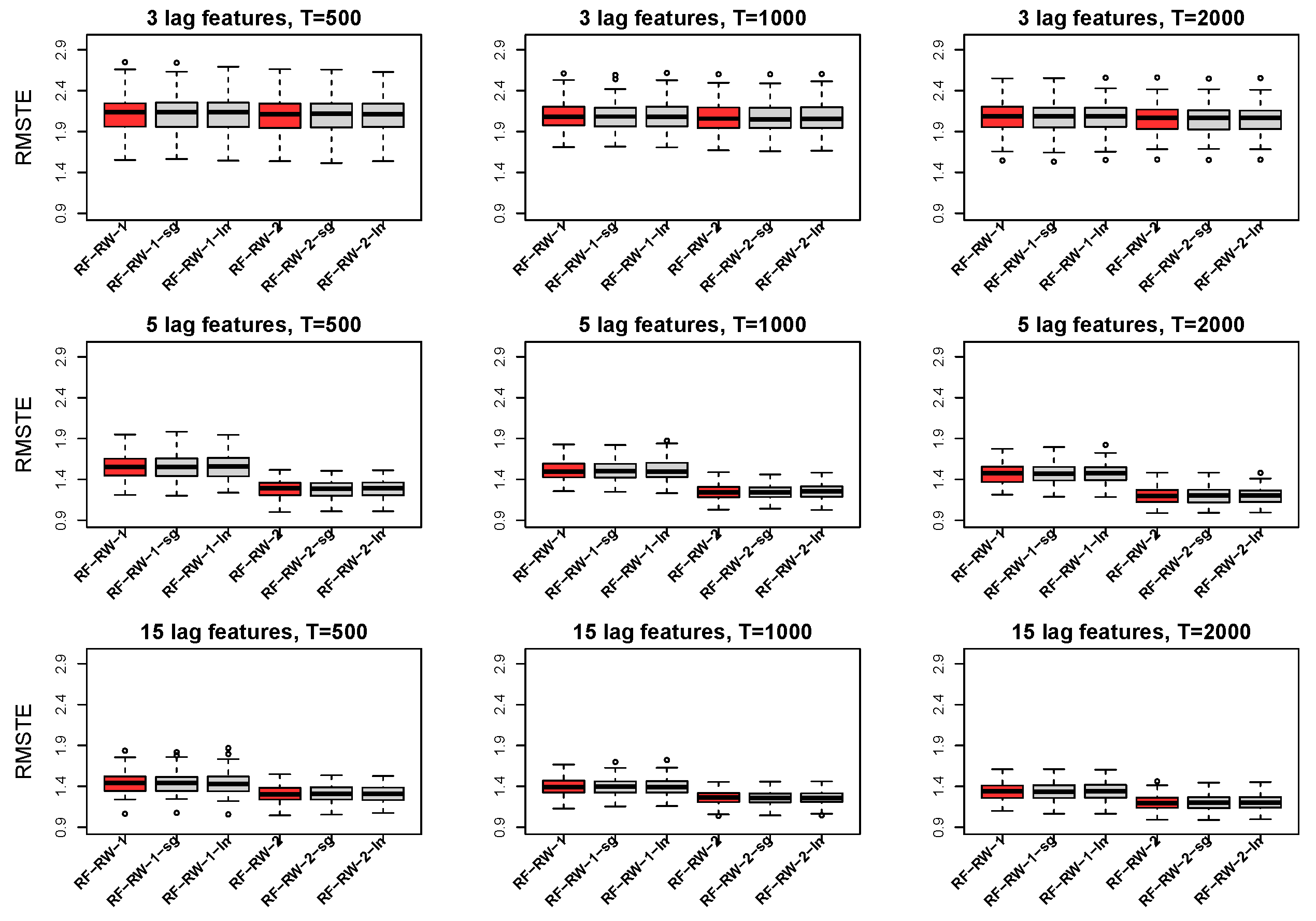}
\caption{RMSTE of M1 for RF-RW, with the random weights sampled from $\text{Exp}(1)$ (RF-RW-1, RF-RW-2), $\text{Lognormal}\left(\mu=-\log2/{2},\sigma^2=\log 2\right)$ (RF-RW-1-ln, RF-RW-2-ln), and $\text{Square root of Gamma}(\alpha=0.295,\theta=6.803)$ (RF-RW-1-sg, RF-RW-2-sg).}
\label{fig:rfrw_weights_compare}
\end{figure}

In Figure \ref{fig:rfrw_weights_compare}, we investigate how the distribution of random weights affects the predictive performance of RF-RW. We consider the random weights sampled from the following distributions: $\text{Exp}(1)$ (as in Section \ref{sec:sim}), $\text{Lognormal}\left(\mu=-\log2/{2},\sigma^2=\log 2\right)$, and $\text{Square root of Gamma}(\alpha=0.295,\theta=6.803)$. The results show that the predictive performance of RF-RW varies little across different distributions of random weights, indicating that RF-RW can be implemented with a wide range of random weight distributions.


\begin{figure}[!ht]
\centering
\includegraphics[width=1\textwidth]{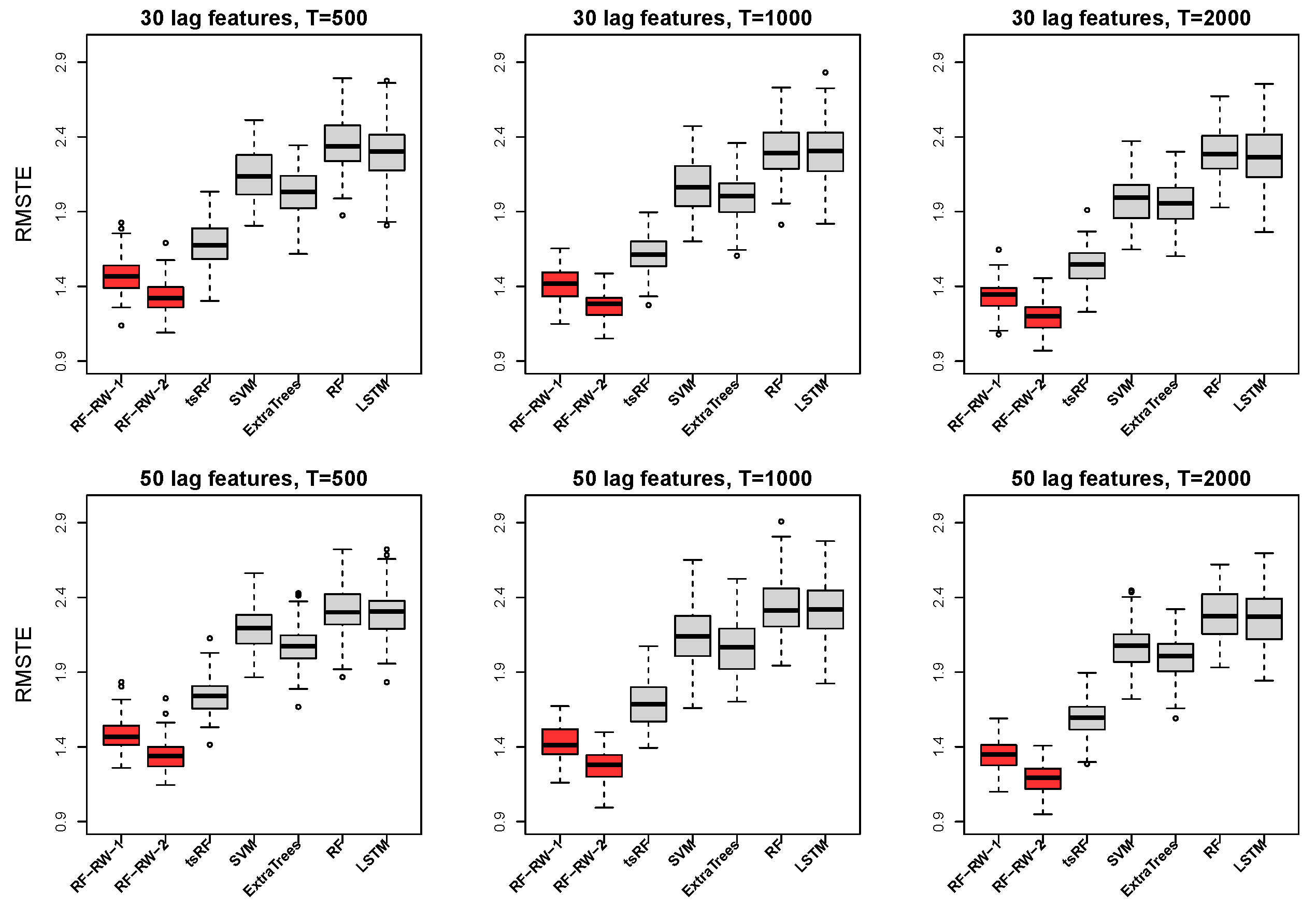}
\caption{RMSTE of M1, overspecified scenarios of methods including 30 and 50 lag features.}
\label{fig:rfrw_dgp1_30_50}
\end{figure}

\begin{table}[!ht]
\centering
\caption{Mean RMSTE values of M1, overspecified scenarios of methods including 30 and 50 lag features} 
\label{table:rfrw_dgp1_30_50}
\begin{tabular}{ c c c c c c c } 
 \hline
 &  \multicolumn{3}{c}{30 lag features} & \multicolumn{3}{c}{50 lag features}\\
 \cmidrule(lr){2-4}\cmidrule(lr){5-7}
 Method & $500$ & $1000$ & $2000$ & $500$ & $1000$ & $2000$ \\
 \hline
  RF-RW-1 & 1.47 &  1.41 & 1.34 & 1.48 & 1.43 & 1.35\\
  RF-RW-2 & \textbf{1.33} & \textbf{1.27} & \textbf{1.20} & \textbf{1.34} & \textbf{1.28} & \textbf{1.20} \\
  tsRF & 1.68 &  1.61 & 1.54 & 1.74 & 1.69 &  1.59 \\
  SVM &  2.14 & 2.07 & 1.98 & 2.19 & 2.14 &  2.07\\
  ExtraTrees & 2.04 & 2.00 & 1.95 & 2.07 & 2.06 &  2.00 \\
  RF &  2.35 & 2.31 & 2.28 & 2.33 & 2.33 &  2.28 \\ 
  LSTM & 2.32 & 2.29 & 2.26 & 2.30 & 2.32 &  2.27\\

 \hline
\end{tabular}
\end{table}


\begin{figure}[!ht]
\centering
\includegraphics[width=1\textwidth]{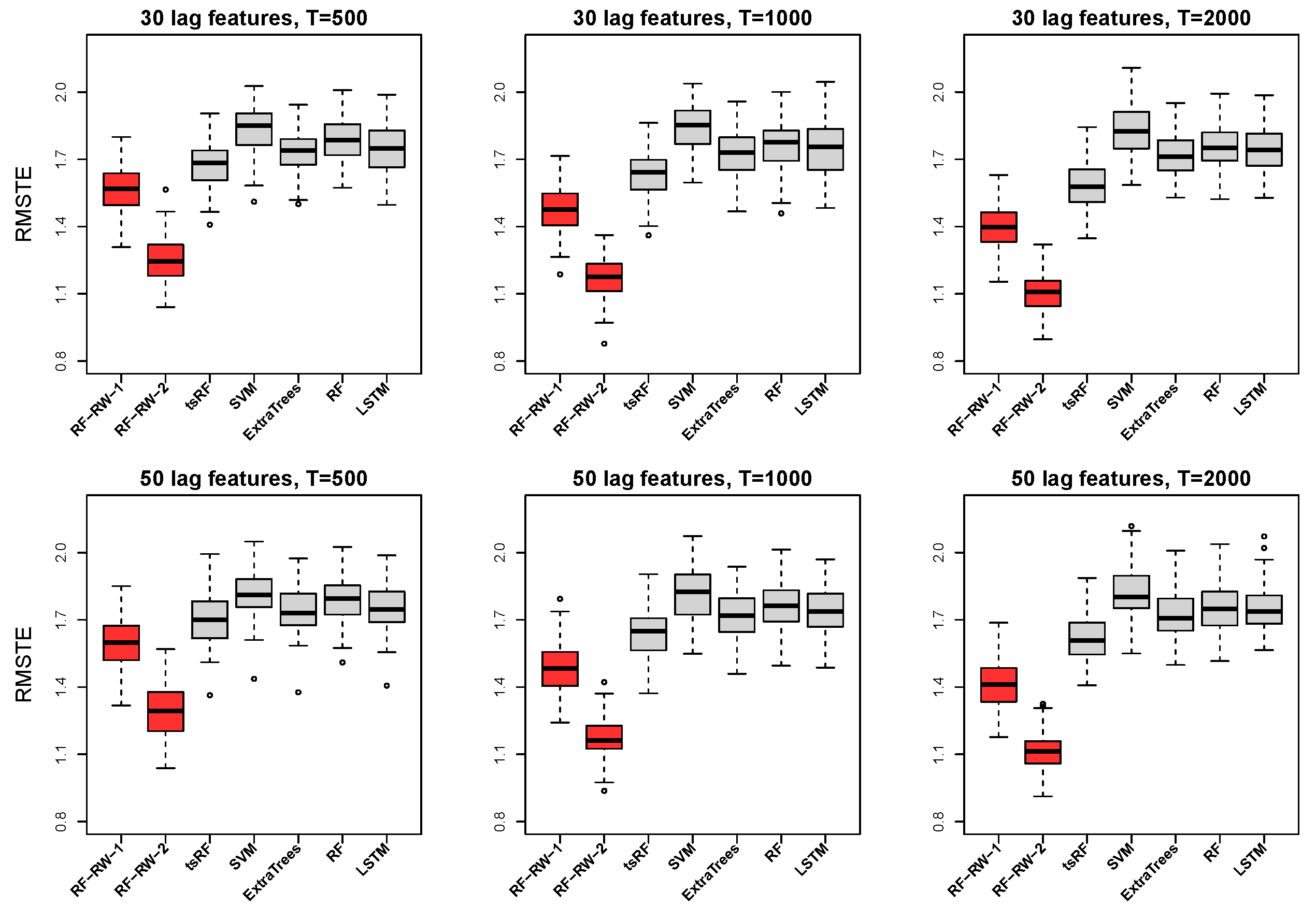}
\caption{RMSTE of M2, overspecified scenarios of methods including 30 and 50 lag features.}
\label{fig:rfrw_dgp2_30_50}
\end{figure}

\begin{table}[!ht]
\centering
\caption{Mean RMSTE values of M2, overspecified scenarios of methods including 30 and 50 lag features} 
\label{table:rfrw_dgp2_30_50}
\begin{tabular}{ c c c c c c c } 
 \hline
 &  \multicolumn{3}{c}{30 lag features} & \multicolumn{3}{c}{50 lag features}\\
 \cmidrule(lr){2-4}\cmidrule(lr){5-7}
 Method & $500$ & $1000$ & $2000$ & $500$ & $1000$ & $2000$ \\
 \hline
  RF-RW-1 & 1.57 & 1.48 & 1.40 & 1.60 & 1.49 & 1.41\\
  RF-RW-2 & \textbf{1.26} & \textbf{1.18} & \textbf{1.11} & \textbf{1.30} & \textbf{1.17} & \textbf{1.12} \\
  tsRF & 1.68 & 1.63 & 1.58 & 1.70 & 1.65 &  1.62 \\
  SVM &  1.83 & 1.84 & 1.83 & 1.82 & 1.82 &  1.82\\
  ExtraTrees & 1.73 & 1.72 & 1.71 & 1.74 & 1.72 &  1.72 \\
  RF &  1.79 & 1.76 & 1.75 & 1.79 & 1.75 &  1.75 \\ 
  LSTM & 1.75 & 1.75 & 1.75 & 1.76 & 1.74 &  1.75\\

 \hline
\end{tabular}
\end{table}

\begin{figure}[!ht]
\centering
\includegraphics[width=1\textwidth]{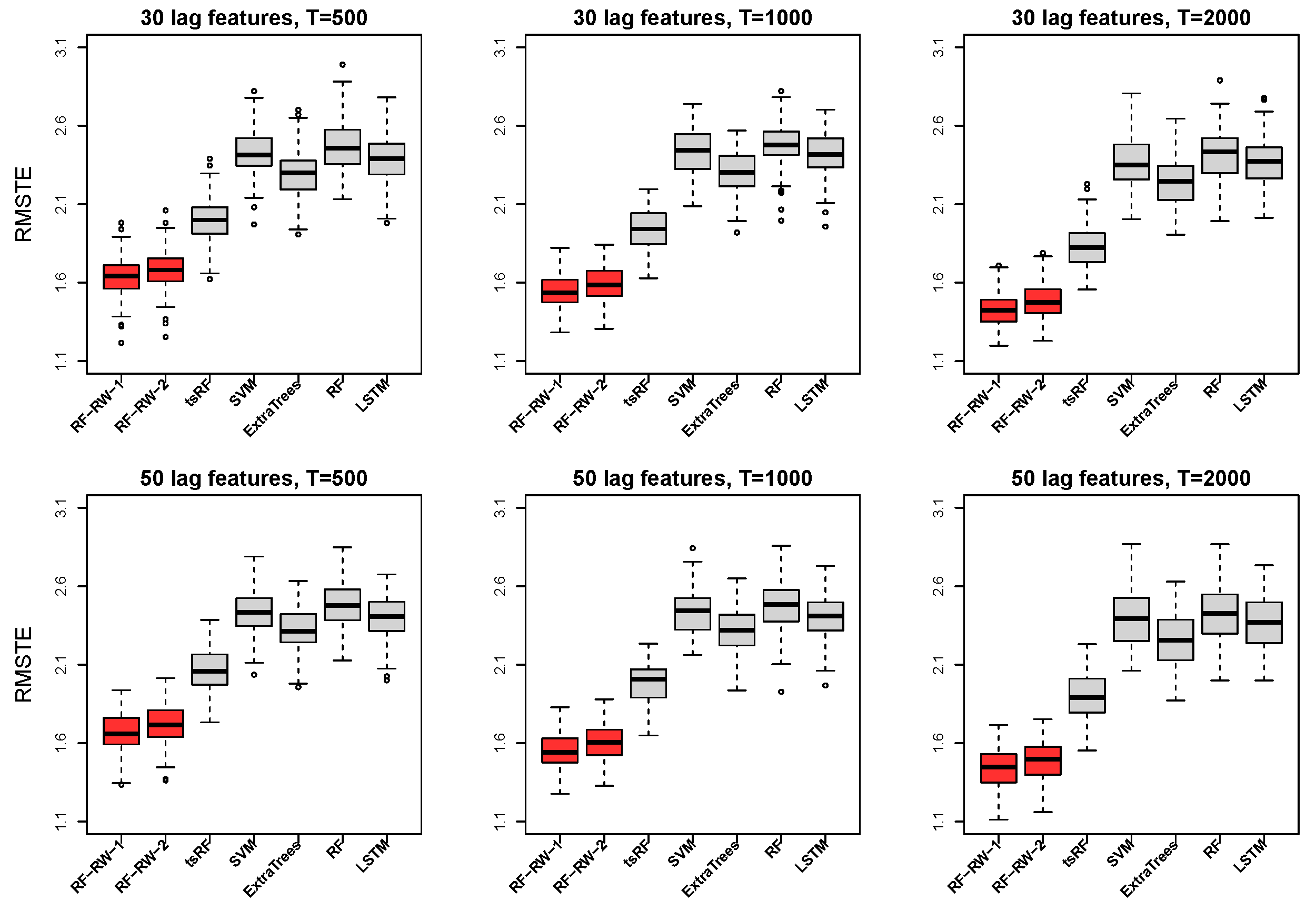}
\caption{RMSTE of M3, overspecified scenarios of methods including 30 and 50 lag features.}
\label{fig:rfrw_dgp3_30_50}
\end{figure}

\begin{table}[!ht]
\centering
\caption{Mean RMSTE values of M3, overspecified scenarios of methods including 30 and 50 lag features} 
\label{table:rfrw_dgp3_30_50}
\begin{tabular}{ c c c c c c c } 
 \hline
 &  \multicolumn{3}{c}{30 lag features} & \multicolumn{3}{c}{50 lag features}\\
 \cmidrule(lr){2-4}\cmidrule(lr){5-7}
 Method & $500$ & $1000$ & $2000$ & $500$ & $1000$ & $2000$ \\
 \hline
  RF-RW-1 & \textbf{1.64} & \textbf{1.55} & \textbf{1.43} & \textbf{1.67} & \textbf{1.55} & \textbf{1.44}\\
  RF-RW-2 & 1.68 & 1.59 & 1.48 & 1.72 & 1.60 & 1.49 \\
  tsRF & 2.00 & 1.94 & 1.83 & 2.06 & 1.99 &  1.89 \\
  SVM &  2.43 & 2.43 & 2.37 & 2.44 & 2.44 & 2.39 \\
  ExtraTrees & 2.30 & 2.30 & 2.24 & 2.32 & 2.31 & 2.26  \\
  RF & 2.48  & 2.48 & 2.41 & 2.48 & 2.47 &  2.42 \\ 
  LSTM & 2.39 & 2.41 & 2.37 & 2.39 & 2.40 & 2.37 \\

 \hline
\end{tabular}
\end{table}

We further investigate the predictive performance of the methods under the overspecified scenario. In particular, we consider the methods that include 30 and 50 lag features, which constitute more severe overspecification than the 15 lags used in the experiments of Section \ref{sec:sim}. The boxplots of the RMSTE of the data-generating processes are shown in Figure \ref{fig:rfrw_dgp1_30_50}-\ref{fig:rfrw_dgp3_30_50}, with the corresponding mean RMTSE values presented in Table \ref{table:rfrw_dgp1_30_50}-\ref{table:rfrw_dgp3_30_50}. The results once again show the superiority of the proposed RF-RW. In conjunction with the simulation results presented in Section \ref{sec:sim}, it supports the conclusion that the predictive performance of the proposed RF-RW is robust to the increase of lag features included,  which implies that RF-RW has the potential to mitigate the curse of dimensionality. 

In contrast, the predictive performance of the competitive \texttt{tsRF} continues to deteriorate as the number of lag features increases. This drawback of \texttt{tsRF} restricts its applicability in large-dimensional time series applications, where the methods often include more lag features to incorporate the temporal information.  

\putbib[literature_supp]
\end{bibunit}








\begin{thebibliography}{31}
\expandafter\ifx\csname natexlab\endcsname\relax\def\natexlab#1{#1}\fi

\bibitem[{Athey et~al.(2019)Athey, Tibshirani \& Wager}]{Athey2019}
\textsc{Athey, S.}, \textsc{Tibshirani, J.} \& \textsc{Wager, S.} (2019).
\newblock Generalized random forests.
\newblock \textit{The Annals of Statistics} \textbf{47}, 1179--1203.

\bibitem[{Biau(2012)}]{Biau2012}
\textsc{Biau, G.} (2012).
\newblock Analysis of a random forests model.
\newblock \textit{Journal of Machine Learning Research} \textbf{13},
  1063--1095.

\bibitem[{Breiman(2001)}]{Breiman2001}
\textsc{Breiman, L.} (2001).
\newblock Random forests.
\newblock \textit{Machine Learning} \textbf{45}, 5--32.

\bibitem[{Breiman et~al.(1984)Breiman, Friedman, Olshen \& Stone}]{Breiman1984}
\textsc{Breiman, L.}, \textsc{Friedman, J.~H.}, \textsc{Olshen, R.~A.} \&
  \textsc{Stone, C.~J.} (1984).
\newblock \textit{Classification And Regression Trees}.
\newblock Chapman and Hall/CRC.

\bibitem[{Chang \& Lin(2011)}]{Chang2011}
\textsc{Chang, C.-C.} \& \textsc{Lin, C.-J.} (2011).
\newblock Libsvm: A library for support vector machines.
\newblock \textit{ACM Transactions on Intelligent Systems and Technology}
  \textbf{2}, 1--27.

\bibitem[{Chen et~al.(2024)Chen, Xiao \& Yao}]{Chen2024}
\textsc{Chen, Q.}, \textsc{Xiao, Z.} \& \textsc{Yao, Q.} (2024).
\newblock Quantile control via random forest.
\newblock \textit{Journal of Econometrics} \textbf{249}, 105789.

\bibitem[{Chernozhukov et~al.(2013)Chernozhukov, Chetverikov \&
  Kato}]{Chernozhukov2013}
\textsc{Chernozhukov, V.}, \textsc{Chetverikov, D.} \& \textsc{Kato, K.}
  (2013).
\newblock Gaussian approximations and multiplier bootstrap for maxima of sums
  of high-dimensional random vectors.
\newblock \textit{The Annals of Statistics} \textbf{41}, 2786--2819.

\bibitem[{Chi et~al.(2022)Chi, Vossler, Fan \& Lv}]{Chi2022}
\textsc{Chi, C.~M.}, \textsc{Vossler, P.}, \textsc{Fan, Y.} \& \textsc{Lv, J.}
  (2022).
\newblock Asymptotic properties of high-dimensional random forests.
\newblock \textit{The Annals of Statistics} \textbf{50}, 3415--3438.

\bibitem[{Davis \& Nielsen(2020)}]{Davis2020}
\textsc{Davis, R.~A.} \& \textsc{Nielsen, M.~S.} (2020).
\newblock Modeling of time series using random forests: Theoretical
  developments.
\newblock \textit{Electronic Journal of Statistics} \textbf{14}, 3644--3671.

\bibitem[{Deng \& Zhang(2020)}]{Deng2020}
\textsc{Deng, H.} \& \textsc{Zhang, C.-H.} (2020).
\newblock Beyond gaussian approximation.
\newblock \textit{The Annals of Statistics} \textbf{48}, 3643--3671.

\bibitem[{Fan \& Yao(2003)}]{Fan2003}
\textsc{Fan, J.} \& \textsc{Yao, Q.} (2003).
\newblock \textit{Nonlinear Time Series: Nonparametric and Parametric Methods}.
\newblock Springer Science \& Business Media.

\bibitem[{Galasso et~al.(2022)Galasso, Cao \& Hochberg}]{Galasso2022}
\textsc{Galasso, J.}, \textsc{Cao, D.~M.} \& \textsc{Hochberg, R.} (2022).
\newblock A random forest model for forecasting regional covid-19 cases
  utilizing reproduction number estimates and demographic data.
\newblock \textit{Chaos, Solitons and Fractals} \textbf{156}, 111779.

\bibitem[{Gao(2007)}]{Gao2007}
\textsc{Gao, J.} (2007).
\newblock \textit{Nonlinear Time Series: Semiparametric and Nonparametric
  Methods}.
\newblock Chapman and Hall/CRC.

\bibitem[{Geurts et~al.(2006)Geurts, Ernst \& Wehenkel}]{Geurts2006}
\textsc{Geurts, P.}, \textsc{Ernst, D.} \& \textsc{Wehenkel, L.} (2006).
\newblock Extremely randomized trees.
\newblock \textit{Machine Learning} \textbf{63}, 3--42.

\bibitem[{Goehry et~al.(2023)Goehry, Yan, Goude, Massart \& Poggi}]{Goehry2023}
\textsc{Goehry, B.}, \textsc{Yan, H.}, \textsc{Goude, Y.}, \textsc{Massart, P.}
  \& \textsc{Poggi, J.~M.} (2023).
\newblock Random forests for time series.
\newblock \textit{REVSTAT-Statistical Journal} \textbf{21}, 283--302.

\bibitem[{Hastie et~al.(2009)Hastie, Tibshirani \& Friedman}]{Hastie2009}
\textsc{Hastie, T.}, \textsc{Tibshirani, R.} \& \textsc{Friedman, J.} (2009).
\newblock \textit{The Elements of Statistical Learning Data Mining, Inference,
  and Prediction}.
\newblock Springer.

\bibitem[{Jin et~al.(2001)Jin, Ying \& Wei}]{Jin2001}
\textsc{Jin, Z.}, \textsc{Ying, Z.} \& \textsc{Wei, L.~J.} (2001).
\newblock A simple resampling method by perturbing the minimand.
\newblock \textit{Biometrika} \textbf{88}, 381--390.

\bibitem[{Kane et~al.(2014)Kane, Price, Scotch \& Rabinowitz}]{Kane2014}
\textsc{Kane, M.~J.}, \textsc{Price, N.}, \textsc{Scotch, M.} \&
  \textsc{Rabinowitz, P.} (2014).
\newblock Comparison of arima and random forest time series models for
  prediction of avian influenza h5n1 outbreaks.
\newblock \textit{BMC Bioinformatics} \textbf{15}, 276.

\bibitem[{Masini et~al.(2023)Masini, Medeiros \& Mendes}]{Masini2023}
\textsc{Masini, R.~P.}, \textsc{Medeiros, M.~C.} \& \textsc{Mendes, E.~F.}
  (2023).
\newblock Machine learning advances for time series forecasting.
\newblock \textit{Journal of Economic Surveys} \textbf{37}, 76--111.

\bibitem[{Meinshausen(2006)}]{Meinshausen2006}
\textsc{Meinshausen, N.} (2006).
\newblock Quantile regression forests.
\newblock \textit{Journal of Machine Learning Research} \textbf{7}, 983--999.

\bibitem[{Paszke et~al.(2019)Paszke, Gross, Massa, Lerer, Bradbury, Chanan,
  Killeen, Lin, Gimelshein, Antiga, Desmaison, Köpf, Yang, DeVito, Raison,
  Tejani, Chilamkurthy, Steiner, Fang, Bai \& Chintala}]{Paszke2019}
\textsc{Paszke, A.}, \textsc{Gross, S.}, \textsc{Massa, F.}, \textsc{Lerer,
  A.}, \textsc{Bradbury, J.}, \textsc{Chanan, G.}, \textsc{Killeen, T.},
  \textsc{Lin, Z.}, \textsc{Gimelshein, N.}, \textsc{Antiga, L.},
  \textsc{Desmaison, A.}, \textsc{Köpf, A.}, \textsc{Yang, E.},
  \textsc{DeVito, Z.}, \textsc{Raison, M.}, \textsc{Tejani, A.},
  \textsc{Chilamkurthy, S.}, \textsc{Steiner, B.}, \textsc{Fang, L.},
  \textsc{Bai, J.} \& \textsc{Chintala, S.} (2019).
\newblock Pytorch: an imperative style, high-performance deep learning library.
\newblock In \textit{Proceedings of the 33rd International Conference on Neural
  Information Processing Systems}. Curran Associates Inc.

\bibitem[{Rady et~al.(2021)Rady, Fawzy \& Fattah}]{Rady2021}
\textsc{Rady, E. H.~A.}, \textsc{Fawzy, H.} \& \textsc{Fattah, A. M.~A.}
  (2021).
\newblock Time series forecasting using tree based methods.
\newblock \textit{Journal of Statistics Applications and Probability}
  \textbf{10}, 229--244.

\bibitem[{Saha et~al.(2023)Saha, Basu \& Datta}]{Saha2023}
\textsc{Saha, A.}, \textsc{Basu, S.} \& \textsc{Datta, A.} (2023).
\newblock Random forests for spatially dependent data.
\newblock \textit{Journal of the American Statistical Association}
  \textbf{118}, 665--683.

\bibitem[{Scornet et~al.(2015)Scornet, Biau \& Vert}]{Scornet2015}
\textsc{Scornet, E.}, \textsc{Biau, G.} \& \textsc{Vert, J.~P.} (2015).
\newblock Consistency of random forests.
\newblock \textit{The Annals of Statistics} \textbf{43}, 1716--1741.

\bibitem[{Scornet \& Hooker(2025)}]{Scornet2025}
\textsc{Scornet, E.} \& \textsc{Hooker, G.} (2025).
\newblock Theory of random forests: a review.
\newblock \textit{HAL} Hal-05006431.

\bibitem[{Shiraishi et~al.(2024)Shiraishi, Nakamura \& Shibuki}]{Shiraishi2024}
\textsc{Shiraishi, H.}, \textsc{Nakamura, T.} \& \textsc{Shibuki, R.} (2024).
\newblock Time series quantile regression using random forests.
\newblock \textit{Journal of Time Series Analysis} \textbf{45}, 639--659.

\bibitem[{Teräsvirta et~al.(2010)Teräsvirta, Tjøstheim \&
  Granger}]{terasvirta2010}
\textsc{Teräsvirta, T.}, \textsc{Tjøstheim, D.} \& \textsc{Granger, C. W.~J.}
  (2010).
\newblock \textit{Modelling Nonlinear Economic Time Series}.
\newblock Oxford University Press.

\bibitem[{Wager \& Athey(2018)}]{Wager2018}
\textsc{Wager, S.} \& \textsc{Athey, S.} (2018).
\newblock Estimation and inference of heterogeneous treatment effects using
  random forests.
\newblock \textit{Journal of the American Statistical Association}
  \textbf{113}, 1228--1242.

\bibitem[{Wager \& Walther(2015)}]{Wager2015}
\textsc{Wager, S.} \& \textsc{Walther, G.} (2015).
\newblock Adaptive concentration of regression trees, with application to
  random forests.
\newblock \textit{arXiv:} 1503.06388.

\bibitem[{Yeşilkanat(2020)}]{YEŞİLKANAT2020}
\textsc{Yeşilkanat, C.~M.} (2020).
\newblock Spatio-temporal estimation of the daily cases of covid-19 in
  worldwide using random forest machine learning algorithm.
\newblock \textit{Chaos, Solitons and Fractals} \textbf{140}, 110210.

\bibitem[{Zheng et~al.(2018)Zheng, Zhu, Li \& Xiao}]{Zheng2018}
\textsc{Zheng, Y.}, \textsc{Zhu, Q.}, \textsc{Li, G.} \& \textsc{Xiao, Z.}
  (2018).
\newblock Hybrid quantile regression estimation for time series models with
  conditional heteroscedasticity.
\newblock \textit{Journal of the Royal Statistical Society Series B}
  \textbf{80}, 975--993.

\end{thebibliography}


\begin{thebibliography}{8}
\expandafter\ifx\csname natexlab\endcsname\relax\def\natexlab#1{#1}\fi

\bibitem[{Chen et~al.(2024)Chen, Xiao \& Yao}]{Chen2024-sup}
\textsc{Chen, Q.}, \textsc{Xiao, Z.} \& \textsc{Yao, Q.} (2024).
\newblock Quantile control via random forest.
\newblock \textit{Journal of Econometrics} \textbf{249}, 105789.

\bibitem[{Davis \& Nielsen(2020)}]{Davis2020-sup}
\textsc{Davis, R.~A.} \& \textsc{Nielsen, M.~S.} (2020).
\newblock Modeling of time series using random forests: Theoretical
  developments.
\newblock \textit{Electronic Journal of Statistics} \textbf{14}, 3644--3671.

\bibitem[{De~la Pena \& Giné(1999)}]{de1999}
\textsc{De~la Pena, V.} \& \textsc{Giné, E.} (1999).
\newblock \textit{Decoupling: From Dependence to Independence}.
\newblock Springer Science \& Business Media.

\bibitem[{Merlevède et~al.(2009)Merlevède, Peligrad \& Rio}]{Merlevède2009}
\textsc{Merlevède, F.}, \textsc{Peligrad, M.} \& \textsc{Rio, E.} (2009).
\newblock Bernstein inequality and moderate deviations under strong mixing
  conditions.
\newblock In \textit{High dimensional probability V: the Luminy volume},
  vol.~5. Institute of Mathematical Statistics, pp. 273--293.

\bibitem[{Rio(1993)}]{rio1993}
\textsc{Rio, E.} (1993).
\newblock Covariance inequalities for strongly mixing processes.
\newblock \textit{Annales de l'IHP Probabilités et statistiques} \textbf{29},
  587--597.

\bibitem[{Shiraishi et~al.(2024)Shiraishi, Nakamura \&
  Shibuki}]{Shiraishi2024-sup}
\textsc{Shiraishi, H.}, \textsc{Nakamura, T.} \& \textsc{Shibuki, R.} (2024).
\newblock Time series quantile regression using random forests.
\newblock \textit{Journal of Time Series Analysis} \textbf{45}, 639--659.

\bibitem[{Wager \& Athey(2018)}]{Wager2018-sup}
\textsc{Wager, S.} \& \textsc{Athey, S.} (2018).
\newblock Estimation and inference of heterogeneous treatment effects using
  random forests.
\newblock \textit{Journal of the American Statistical Association}
  \textbf{113}, 1228--1242.

\bibitem[{Wager \& Walther(2015)}]{Wager2015-sup}
\textsc{Wager, S.} \& \textsc{Walther, G.} (2015).
\newblock Adaptive concentration of regression trees, with application to
  random forests.
\newblock \textit{arXiv:} 1503.06388.

\end{thebibliography}
\end{document}